\documentclass[12pt]{article} %
\usepackage[utf8]{inputenc}
\usepackage[tracking=true, letterspace=50, expansion=false]{microtype}

\usepackage{amsmath}
\usepackage{amssymb}
\usepackage{amsthm}
\usepackage{thmtools}

\usepackage[margin=1in]{geometry} %
\usepackage{setspace} %
\usepackage[title]{appendix} %
\usepackage{booktabs} %
\usepackage[inline]{enumitem} %
\usepackage{subcaption}

\usepackage[
 bibencoding=utf8,
 maxbibnames=9, 
 minbibnames=4, 
 backend=biber,%
 sortlocale=en_US,%
 style=authoryear,%
 doi=true,%
 url=true,
 useprefix=true, %
 eprint=true%
]{biblatex}

\newtheorem{theorem}{Theorem}[section]
\newtheorem{lemma}{Lemma}[section]
\newtheorem{corollary}{Corollary}[section]
\newtheorem{assumption}{Assumption}[section]
\theoremstyle{definition}
\newtheorem{remark}{Remark}[section]
\usepackage{etoolbox}
\AtEndEnvironment{example}{\qed}%
\newtheorem{example}{Example}[section]
\AtEndEnvironment{algorithm}{\qed}%
\newtheorem{algorithm}{Algorithm}[section]

\usepackage{graphicx}

\DeclareMathOperator{\cv}{cv}

\DeclareMathOperator{\bias}{bias}
\DeclareMathOperator{\maxbias}{\overline{bias}}

\DeclareMathOperator*{\argmin}{arg\,min}
\DeclareMathOperator{\var}{var}
\DeclareMathOperator{\sd}{sd}
\DeclareMathOperator{\Pen}{Pen}
\DeclareMathOperator{\se}{se}

\usepackage{mathtools}
\DeclarePairedDelimiter\hor{[}{)}

\DeclarePairedDelimiter\abs{\lvert}{\rvert}
\DeclarePairedDelimiter\norm{\lVert}{\rVert}

\renewcommand*{\arraystretch}{1.2}
\newcommand{\ND}{\mathcal{N}} 
\newcommand{\sRk}{\mathcal{G}} 
\newcommand{\Rmax}{\mkern 1.5mu\overline{\mkern-1.5mu R\mkern-1.5mu}%
  \mkern 1.5mu} %

\usepackage{xcolor}%
\definecolor{webbrown}{rgb}{.6,0,0}

\usepackage[nolist]{acronym}
\begin{acronym}
  \acro{OLS}{ordinary least squares}
  \acro{ATE}{average treatment effect}
  \acro{TE}{treatment effect}
  \acro{MSE}{mean squared error}
  \acro{CI}{confidence interval}
  \acro{DGP}{data generating process}
  \acro{FLCI}{fixed-length confidence interval}
  \acro{FGLS}{feasible generalized least squares}
  \acro{EHW}{Eicker-Huber-White}
  \acro{FWL}{Frisch–Waugh–Lovell}
  \acro{MPE}{marginal propensity to earn}
  \acro{RHC}{right heart catheterization}
\end{acronym}

\newcommand\RMSE{\ensuremath\operatorname{R}_{\textnormal{MSE}}}

\allowdisplaybreaks

\usepackage[normal, online, flushleft]{threeparttable}%
\usepackage{adjustbox}
\usepackage{tikz} %
\usepackage[capposition=bottom]{floatrow}

\usepackage{hyperref}
\hypersetup{%
  breaklinks = true,%
  colorlinks = true,%
  anchorcolor = webbrown,%
  citecolor = webbrown,%
  filecolor = webbrown,%
  linkcolor = webbrown,%
  menucolor = webbrown,%
  urlcolor= webbrown,%
  citebordercolor= 1 0 0,%
  menubordercolor=1 0 0,%
  urlbordercolor=1 0 0,%
  runbordercolor=1 0 0,%
  pdftitle=Bias-Aware Inference in Regularized Regression Models,%
  pdfauthor=Timothy B. Armstrong \&
  Michal Kolesár \& Soonwoo Kwon}
\usepackage{cleveref}
\makeatletter
\@for\@tempa:=theorem,lemma,corollary,assumption,remark,example,algorithm\do{%
  \expandafter\xdef\csname theH\@tempa\endcsname{%
    \noexpand\theHsection.\noexpand\arabic{\@tempa}}}
\makeatother
\crefname{theorem}{theorem}{theorems}
\crefname{assumption}{assumption}{assumptions}
\crefname{corollary}{corollary}{corollaries}
\crefname{algorithm}{algorithm}{algorithms}
\crefname{remark}{remark}{remarks}
\crefname{proposition}{proposition}{propositions}
\crefname{appsec}{appendix}{appendices}

\title{Bias-Aware Inference in Regularized Regression Models\thanks{An earlier
    version of this paper was circulated under the title ``Optimal Inference in
    Regularized Regression Models''. Parts of this paper incorporate material
    from Section 4 of the working paper \textcite{ArKo16optimal}, which was
    taken out in the published version \parencite{ArKo18optimal}. We thank
    Victor Chernozhukov for helpful comments and discussion. We thank Mark Li
    and Ulrich Müller for sharing their code. Armstrong acknowledges support
    from National Science Foundation Grant SES-2049765. Kolesár acknowledges
    support by the Sloan Research Fellowship and from National Science
    Foundation Grant SES-2049356. Kwon acknowledges support from the Annenberg
    Institute. Juan Yamin Silva provided excellent research assistance.}}
\author{Timothy B. Armstrong\thanks{email: \texttt{timothy.armstrong@usc.edu}}\\
  University of Southern California \and
  Michal Kolesár\thanks{email: \texttt{mkolesar@princeton.edu}}\\
  Princeton University \and Soonwoo Kwon\thanks{email:
    \texttt{soonwoo\_kwon@brown.edu}}\\ Brown University} \date{\today}

\begin{document}

\maketitle

\begin{abstract}
  We consider inference on a scalar regression coefficient under a constraint on
  the magnitude of the control coefficients. A class of estimators based on a
  regularized propensity score regression is shown to exactly solve a tradeoff
  between worst-case bias and variance. We derive \acp{CI} based on these
  estimators that are bias-aware: they account for the possible bias of the
  estimator. Under homoskedastic Gaussian errors, these estimators and \acp{CI}
  are near-optimal in finite samples for mean squared error and \ac{CI} length. We also
  provide conditions for asymptotic validity of the \acp{CI} with an unknown and
  possibly heteroskedastic error distribution, and derive novel optimal rates of
  convergence under high-dimensional asymptotics that allow the number of
  regressors to increase more quickly than the number of observations. Extensive
  simulations and two empirical applications illustrate the performance of our
  methods.
\end{abstract}

\clearpage

\section{Introduction}

We study estimation and inference for a scalar coefficient $\beta$ in a linear
regression model
\begin{equation}\label{eq:linear_regression}
  Y_{i} = w_{i}\beta + z_{i}'\gamma + \varepsilon_{i}, \qquad i=1, \dotsc, n,
\end{equation}
when the dimension $k$ of the controls $z_{i}$ is large. In such settings, the
\ac{OLS} estimator is often uninformative, exhibiting a large variance; the
estimator is not even defined when $k>n$. This motivates modifying the \ac{OLS}
objective function to penalize large values of $\gamma$, thereby lowering
variance at the cost of introducing bias.

The most popular of these approaches are the lasso
\parencite{tibshirani_regression_1996} and other variants of $\ell_{1}$
penalization \parencite[e.g.,][]{CaTa07,BeChWa11}. There is a large literature
\parencite[see, e.g.,][for a review]{buhlmann_statistics_2011} showing favorable
\ac{MSE} properties of these estimators when $\gamma$ is (approximately) sparse.
For inference, several papers have developed \acp{CI} based on ``double lasso''
estimators \parencite[see, among
others,][]{belloni_inference_2014,javanmard_confidence_2014,van_de_geer_asymptotically_2014,zhang_confidence_2014}.
These methods come with large-sample coverage guarantees even when $k\gg n$,
provided that the sparsity level of $\gamma$ grows slowly enough relative to
$n$. However, because the theory only restricts the rate of growth of the
sparsity level, it is unclear what exact degree of sparsity this approach
implicitly requires to guarantee good performance in a given sample.
Furthermore, in many applications in economics, there are no compelling
substantive arguments for why only a small number of the controls should be able
to account for the bulk of the confounding.

We propose bounding the \emph{magnitude} of the control coefficients, rather
than their sparsity level, by assuming that $\Pen(\gamma)\leq C$. The penalty
function $\Pen(\cdot)$ formalizes the notion of magnitude, and it can
incorporate any restrictions on $\gamma$ that place it in a convex symmetric
set. In many applications, high-dimensional controls arise in robustness checks,
where researchers examine the sensitivity of the baseline results to the
inclusion of additional controls, such as squares, interactions, or other
transformations of the initial controls, or additional covariates that are
believed to be less important.\footnote{\label{fn:examples_high_dim}Prominent
  examples include \textcite{JoMoRe19,ckmmlrs21,dhar2022,boudreau24}.} In such
cases, one can take the penalty to be an $\ell_{p}$ norm of the additional
controls. We focus on two specifications in particular. The first is an
$\ell_{1}$ penalty that can be interpreted as bounding the \emph{total effect}
of the additional controls (the maximal impact on the outcome of simultaneously
shifting each of them by up to one standard deviation). The second is a weighted
$\ell_{2}$ penalty that bounds their \emph{explanatory power} (their partial
$R^{2}$). Following the proposals in the recent sensitivity analysis literature
\parencite{altonji_selection_2005,oster_unobservable_2019,CiHa20,diegert2026},
the penalty bound $C$ can be calibrated by relating the penalty to the total
effect or explanatory power of the baseline controls in the baseline
specification.

We obtain sharp finite-sample results deriving near-optimal estimators and
\acp{CI} under the chosen penalty constraint in the idealized scenario in which
the regression errors $\varepsilon_{i}$ are Gaussian with a known homoskedastic
variance. The estimator that optimally trades off the finite-sample worst-case
bias against variance takes a simple form, which can be understood as a
penalized version of the Frisch-Waugh formula. First, one runs a penalized
propensity score regression of $w_i$ on $z_i$ using $\Pen(\cdot)$ as the penalty
function. Second, one uses the propensity score residuals as an instrument in
the univariate regression of $Y_{i}$ on $w_{i}$. \Acp{CI} based on these
estimators can be constructed by using a critical value that incorporates the
worst-case bias of the estimator, which we show is obtained automatically as a
byproduct of the first-step propensity score regression. Because these \acp{CI}
are \emph{bias-aware}---they account for the potential finite-sample bias of the
estimator---they are valid in finite samples in this idealized Gaussian setup.
We show how to tune the penalty parameter in the propensity score regression to
optimize the \ac{MSE} of the resulting estimator, or the length of the resulting
\ac{CI}.

In the more realistic scenario with heteroskedasticity and an unknown error
distribution, bias-aware \acp{CI} can be formed using heteroskedasticity-robust
variance estimators. We give conditions for their asymptotic validity, allowing
for high-dimensional asymptotics with $k\gg n$. We also study rates of
convergence and the efficiency of the bias-aware \acp{CI} relative to the
unpenalized regression. We show that when $k<n$, meaningfully improving upon the
unpenalized regression by bounding the partial $R^{2}$ of the controls is only
possible when the sample size is small to moderate. In the high-dimensional
setting with $k\gg n$, consistency requires shrinking the bound on the partial
$R^{2}$ to zero, which may be hard to motivate. By contrast, the $\ell_{1}$
penalty bound that limits the total effect of the controls achieves fast
convergence rates even when $C$ does not shrink with $n$. Furthermore, we show
that this convergence rate cannot be improved even if one additionally imposes
the same $\ell_1$ bound on the coefficients in the propensity score regression,
as well as a certain degree of sparsity in both regressions.

Our approach requires an explicit specification of the penalty bound $C$. Our
efficiency bounds show that it is impossible to automate the choice of $C$ when
forming \acp{CI}; fundamentally, one cannot obtain a sufficiently informative
data-driven upper bound for $C$. We show, however, that it is possible to obtain
a one-sided \emph{lower} \ac{CI} for $C$, which can be used as a specification
check to ensure that the chosen value is not too low. We also discuss how a
careful calibration of $C$ can be supplemented by computing its breakdown value,
the largest value of $C$ such that the empirical finding of interest, such as
rejecting a particular null hypothesis, holds.

The requirement to explicitly choose $C$ may seem like a limitation of our
approach relative to sparsity-based approaches, which do not require explicitly
bounding the degree of sparsity. However, good finite-sample performance of such
methods \emph{does} require bounding the degree of sparsity implicitly, and such
implicit bounds are hard to calculate in a given problem. We demonstrate this
issue in a Monte Carlo analysis, where we show that double-lasso \acp{CI}
display moderate to severe undercoverage even in designs that are apparently
sparse. Indeed, we view the explicit specification of $C$ as an advantage of our
approach, because our coverage guarantees and efficiency bounds are based on
explicit assumptions rather than ``asymptotic promises'' that are hard to
evaluate in a particular sample.

Two empirical applications illustrate the theory. The first application
investigates the robustness of the analysis in \textcite{imbens2001lottery} to
treatment effect heterogeneity. The sample size is moderate, and in line with
the theory, we find that bounding the partial $R^{2}$ of the additional controls
yields nonnegligible efficiency gains over the unpenalized regression. Assuming
an $\ell_{1}$ penalty bound yields tighter inference. The point estimates remain
stable over a range of specifications, indicating that the original analysis is
not sensitive to dropping the assumption of constant treatment effect. The
second application examines the sensitivity to linear functional form
assumptions of the conclusions in \textcite{connors_rhc_1996} and
\textcite{hirano_rhc_2001} that \ac{RHC} increases mortality. Due to the large
sample size, the bound on partial $R^{2}$ has little bite; the $\ell_{1}$
penalty bound suggests that the effects of \ac{RHC} are likely smaller in
magnitude than in the baseline regression, though they remain statistically
significant.

Our results relate to several strands of literature. Our procedures and
efficiency bounds apply the general theory of estimation and inference on linear
functionals in convex Gaussian models developed in
\textcite{IbKh85,donoho94,low95,ArKo18optimal}, and add to a growing literature
applying this approach to various settings
\parencite{ArKo20ate,ArKo20sensitivity,kolesar_inference_2018,ImWa19,RaRo19,NoRo24,kwon_inference_2020}.
\textcite{muralidharan_factorial_2020} apply the approach in the present paper
to experiments with factorial designs and bounds on interaction effects.

The idea of using propensity score residuals to estimate $\beta$ goes back at
least to the work of \textcite{robinson_root-n-consistent_1988} on the partially
linear model. We provide a novel finite-sample justification for this idea, as
well as an exact result giving the optimal penalization in this regression. Our
setup allows for a general form of $\Pen(\cdot)$, and yields existing estimators
in a few special cases, including the linear estimator of
\textcite{heckman_minimax_1988} in the partially linear model, and the ridge
regression estimator considered in \textcite{li82} under an unweighted
$\ell_{2}$ penalty. Finally, \textcite{li_linear_2020} consider the weighted
$\ell_2$ penalty that bounds the explanatory power of the controls, and develop
a bias-aware \ac{CI} based on a likelihood ratio statistic, which is numerically
shown to be close to optimal under homoskedasticity and a particular weighted
average length criterion. However, unlike our \ac{CI}, the Li--M\"uller \ac{CI}
may end up being longer than the unpenalized regression \ac{CI}, as one of our
empirical applications in \Cref{sec:empirical-app} illustrates.

The next section presents our finite-sample results in the idealized model with
Gaussian errors. \Cref{sec:impl-with-non} discusses implementation in the more
realistic setting with an unknown error distribution. \Cref{sec:theor-prop}
gives efficiency results and derives rates of convergence.
\Cref{sec:nonconvex_efficiency} compares our approach to \acp{CI} motivated by
sparsity constraints. The performance of our methods is evaluated in a Monte
Carlo study in \Cref{sec:simulation-results}, while \Cref{sec:empirical-app}
illustrates them in two empirical applications. Proofs, extensions, and
auxiliary results appear in appendices.

\section{Method}\label{sec:finite-sample-result}

To motivate our method, we set up an idealized version of our model with
Gaussian homoskedastic errors. We then show how to construct estimators and
\acp{CI} in this model that optimize the finite-sample tradeoff between bias and
variance.

\subsection{Setup}\label{sec:setup}

We write the model in \cref{eq:linear_regression} in vector form as
\begin{equation}\label{eq:linear_regression_vector}
  Y = w\beta + Z\gamma + \varepsilon,
\end{equation}
where $w=(w_1,\dotsc, w_n)'\in\mathbb{R}^n$ stacks the variable of interest, and
$Z=(z_{1}, \dotsc, z_{n})'\in \mathbb{R}^{n\times k}$ is a matrix of control
variables. The design matrix $X=(w, Z)$ is treated as fixed. To obtain
finite-sample results, we further assume that the errors are normal and
homoskedastic, $\varepsilon\sim\ND(0, \sigma^{2}I_{n})$, with $\sigma^{2}$
known. To ensure informative inference on $\beta$ when the number of controls
$k$ is large relative to the sample size $n$ (and possibly exceeds it), we
assume the control coefficients $\gamma$ are ``small'' in the sense that their
parameter space is given by
\begin{equation}\label{eq:Gamma_def}
  \Gamma=\Gamma(\Pen;C)= \{\gamma\in \sRk\colon \Pen(\gamma)\le C\}
\end{equation}
for some seminorm $\Pen(\cdot)$ on a linear subspace $\sRk$ of $\mathbb{R}^{k}$,
and some penalty bound $C$. We do not restrict $\beta$, so that the model
parameters satisfy $(\beta, \gamma')'\in\mathbb{R}\times \Gamma$. The
requirement that $\Pen(\cdot)$ be a seminorm means that it satisfies the
triangle inequality
($\Pen(\gamma+\tilde\gamma)\le \Pen(\gamma)+\Pen(\tilde\gamma)$), and
homogeneity ($\Pen(c\gamma)=\abs{c}\Pen(\gamma)$ for any scalar $c$), but,
unlike a norm, it is not necessarily positive definite ($\Pen(\gamma)=0$ does
not imply $\gamma=0$). This allows us to restrict only a subset of the control
coefficients.

For the choice of $\Pen(\gamma)$, we focus on two leading examples, both of
which partition the controls into a set of $k_{1}\geq 0$ unrestricted baseline
controls and a set of $k_{2}=k-k_{1}$ additional controls, $Z=(Z_{1}, Z_{2})$.
We partition $\gamma=(\gamma_{1}', \gamma_{2}')'$ accordingly. We denote the
$\ell_{p}$ norm of a vector $a$ by $\norm{a}_{p}$, and the projection onto the
column space of a matrix $A$ by $H_{A}$.

\begin{example}[$\ell_{2}$ penalty]\label{example:l2}
  Let $M$ be a $k_{2}\times k_{2}$ matrix that scales the additional controls,
  and set $\Pen(\gamma)=\norm{M\gamma_{2}}_{2}$. Setting $M=I_{k_{2}}$ weights
  the coefficients equally, while setting $M=(Z_{2}'(I-H_{Z_{1}})Z_{2}/n)^{1/2}$
  corresponds to the specification considered in \textcite{li_linear_2020}. The
  latter specification has an intuitive interpretation: the population partial
  $R^{2}$ from adding $Z_{2}$ to the regression of $Y-w\beta$ onto $Z_{1}$ is
  given by
  $R^{2}_{Y-w\beta}=\norm{M\gamma_{2}}_{2}^{2}/(\norm{M\gamma_{2}}_{2}^{2}
  +\sigma^{2})$. If $w$ is a treatment dummy, imposing \cref{eq:Gamma_def} thus
  amounts to bounding the partial $R^{2}$ from the regression of the untreated
  outcome on the additional controls by $C^{2}/(C^{2}+\sigma^{2})$. This
  specification for the penalty is therefore natural if the researcher wishes to
  limit the \emph{explanatory power} of the additional controls. We refer to it
  as the $R^{2}_{Y-w\beta}$ penalty in the sequel.
\end{example}
\begin{example}[$\ell_{1}$ penalty]\label{example:l1}
  We focus on the unweighted $\ell_{1}$ penalty, and after standardizing the
  controls to have zero mean and unit variance, set
  $\Pen(\gamma)=\norm{\gamma_{2}}_{1}$. This specification limits the
  \emph{total effect} of the controls:
  $\norm{\gamma_{2}}_{1}=\sum_{j=k_{1}+1}^{k}\abs{\gamma_{j}}$ is the maximal
  impact on the outcome by simultaneously shifting each control by up to one
  standard deviation.
\end{example}

Specification of the parameter space $\Gamma$ requires selecting the penalty
bound $C$, in addition to picking the penalty; here we take $C$ as given, and
defer a discussion of its choice to \Cref{sec:impl-with-non}.

Formulating the parameter space $\Gamma$ in terms of a seminorm is not
restrictive in the sense that essentially any convex set $\Gamma$ that is
symmetric ($\gamma\in\Gamma$ implies $-\gamma\in\Gamma$) can be defined in this
way \parencite[see][Proposition 5, p.~26]{yosida_functional_1995}. Furthermore,
the assumption that the controls enter linearly in
\cref{eq:linear_regression_vector} is not essential: as we show in
\Cref{sec:partly-linear-model}, our setup also covers the partially linear model
that only places smoothness restrictions on how the controls enter the outcome
equation. Although we rule out nonconvex constraints on $\Gamma$, such as
sparsity, our results nonetheless have implications for such settings, as we
discuss in \Cref{sec:nonconvex_efficiency}.

\subsection{Linear estimators and CIs}\label{sec:line-estim-cis}

We now show how to build finite-sample valid \acp{CI} based on linear estimators
of the form $\hat{\beta}=a'Y$ for some $n$-vector of weights $a$ that may depend
on the design matrix $X$ and the known variance $\sigma^{2}$. We then show how
to choose the weights $a$ optimally. \Cref{sec:linear_efficiency} below shows
that when $a$ is optimally chosen, the resulting estimators and \acp{CI} are
optimal or near-optimal among all procedures, not just linear ones.

Since the estimator $\hat{\beta}$ is linear, under Gaussian errors, it is
normally distributed with mean $a'(w\beta+Z\gamma)$ and known standard deviation
$\sd(\hat\beta) = \sigma\norm{a}_{2}$. Consequently, the $z$-statistic
$(\hat{\beta}-\beta)/\sd(\hat\beta)$ follows a $\ND(b,1)$ distribution, where
$b=(a'(w\beta+Z\gamma)-\beta)/\sd(\hat\beta)$ denotes the bias-standard
deviation ratio. While we cannot directly compute $b$ because the bias depends
on the unknown parameters $\beta$ and $\gamma$, under the parameter space
restriction in \cref{eq:Gamma_def}, we know $\abs{b}$ is bounded by
$\maxbias_{\Gamma}(\hat\beta)/\sd(\hat\beta)$, where
$\maxbias_{\Gamma}(\hat\beta)$ denotes the \emph{worst-case bias} of
$\hat{\beta}$ over the parameter space. To ensure that the worst-case bias is
finite, it is necessary that $a'w=1$; under this condition,
\begin{equation}\label{eq:maxbias_def}
  \maxbias_{\Gamma}(\hat\beta)
  = \sup_{\gamma\in \Gamma}a'Z\gamma.
\end{equation}
Thus, we can form a two-sided \ac{CI} for $\beta$ by replacing the usual
critical value based on quantiles of the standard normal distribution with a
critical value that accounts for the non-zero mean of the $z$-statistic:
\begin{equation}\label{eq:linear-CI}
  \hat{\beta} \pm
  \cv_{\alpha}\left(\maxbias_{\Gamma}(\hat\beta)/\sd(\hat\beta)\right)\cdot \sd(\hat{\beta}).
\end{equation}
Here $\cv_{\alpha}(B)$ denotes the $1-\alpha$ quantile of the folded normal
distribution, $\abs{\ND(B, 1)}$.\footnote{\label{fn:compute_cv}The critical
  value $\cv_{\alpha}(B)$ can be computed as the square root of the $1-\alpha$
  quantile of a non-central $\chi^{2}$ distribution with $1$ degree of freedom
  and non-centrality parameter $B^{2}$.} This \ac{CI} is \emph{bias-aware} in
that the critical value $\cv_{\alpha}(\cdot)$ reflects the potential
finite-sample bias of $\hat{\beta}$.

In the terminology of \textcite{donoho94}, the \ac{CI} in \cref{eq:linear-CI} is
a \emph{fixed-length} \ac{CI}, since its length depends only on the non-random
weights $a$ and the known variance $\sigma^{2}$, but not on the outcome data $Y$
or the parameter vector $(\beta, \gamma')'$. This is useful, as it allows us to
pick the weights $a$ that minimize the length of the resulting \ac{CI} without
affecting its coverage properties. To do this, observe that the \ac{CI} length
depends on the $n$-vector of weights $a$ only through the standard deviation and
the worst-case bias, and it is increasing in both $\sd(\hat{\beta})$ and
$\maxbias_{\Gamma}(\hat\beta)$. As a result, we can find the optimal weights by
first minimizing variance subject to a bound $B$ on the worst-case bias,
\begin{equation}\label{eq:a_optimization}
  \min_{a\in\mathbb{R}^{n}} a'a
  \quad\text{s.t.}\quad
  \sup_{\gamma\in \Gamma} a'Z\gamma \le B,\qquad a'w=1,
\end{equation}
and then varying the bound $B$ to find the bias-variance tradeoff yielding the
shortest \ac{CI}.

While solving \cref{eq:a_optimization} for a range of bounds $B$ may appear
cumbersome, applying the general theory of
\textcite{IbKh85,donoho94,low95,ArKo18optimal} to our setting allows us to
rewrite it as a tractable convex optimization problem. Solving that problem
yields our main computational result:
\begin{theorem}\label{thm:optimization_equivalence}
  Let $\pi_{\lambda}$ be a penalized propensity score regression estimator that
  solves
  \begin{equation}\label{eq:pi_optimization}
    \min_{\pi} \norm{w-Z\pi}^{2}_{2} \quad\text{s.t.}\quad \Pen(\pi)\leq t,
  \end{equation}
  for some strictly positive bound $t=t(\lambda)$. Suppose that the propensity
  score residual $\tilde{w}_{\lambda}=w-Z\pi_{\lambda}$ is not identically zero,
  and that $\Pen(\pi_{\lambda})>0$. Then the worst-case bias and standard
  deviation of the estimator
  \begin{equation}\label{eq:optimal_estimator}
    \hat{\beta}_{\lambda}=a_{\lambda}'Y,\quad a_{\lambda}=\frac{\tilde{w}_{\lambda}}{\tilde{w}_{\lambda}'w},
  \end{equation}
  are given by
  \begin{equation}\label{eq:betahat_lambda_def}
    \maxbias_{\Gamma}(\hat\beta_\lambda)=C \overline{B}_\lambda,
    \quad\text{and}\quad \sd_{\lambda} =
    \sigma \norm{a_{\lambda}}_{2},\quad\text{where}\quad
    \overline B_\lambda
    =\frac{a_{\lambda}'Z\pi_\lambda }{\Pen(\pi_\lambda)}.
 \end{equation}
 Furthermore, the weights $a_{\lambda}$ solve \cref{eq:a_optimization} with the
 bound $B=C\overline{B}_{\lambda}$.
\end{theorem}

According to \Cref{thm:optimization_equivalence}, to find the optimal weights,
we first run a penalized propensity score regression of $w$ on $Z$ (we use the
term propensity score regression even though $w$ may not be binary). It does not
matter precisely how the tuning parameter $\lambda$ maps to the penalty bound
$t$. We may take $\lambda$ to correspond to the Lagrange multiplier in a
Lagrangian formulation of \cref{eq:pi_optimization}, but we can also take
$\lambda=t$ or other parametrizations depending on what is most convenient. We
then use the residuals $\tilde{w}_{\lambda}$ from the penalized propensity score
regression as instruments in the univariate regression of $Y$ on $w$. The
resulting estimator $\hat{\beta}_{\lambda}$ solves the bias-variance tradeoff in
\cref{eq:a_optimization} for a particular value of $B$ that depends on the
penalty weight $\lambda$.

Since \cref{eq:betahat_lambda_def} gives explicit formulas for the worst-case
bias and standard deviation of $\hat{\beta}_{\lambda}$, finding the penalty
weight $\lambda^{*}_{\textnormal{CI}}$ that optimizes the length of the
resulting \ac{CI} is a simple univariate optimization problem. We can use the
same approach to find the best estimator. This is because the worst-case
\ac{MSE} of any linear estimator is likewise increasing in its worst-case bias
and its standard deviation. Specifically, the penalty weights
$\lambda^{*}_{\textnormal{MSE}}$ and $\lambda^{*}_{\textnormal{CI}}$ that
optimize the \ac{MSE} and \ac{CI} length, respectively, solve
\begin{equation}\label{eq:optimal_lambda}
  \lambda^{*}_{\textnormal{MSE}} = \argmin_{\lambda} \sd_\lambda^{2} + (C \overline B_\lambda)^2,
  \quad \lambda^{*}_{\textnormal{CI}}
  = \argmin_{\lambda} \cv_{\alpha}(C\overline B_\lambda/\sd_\lambda)\sd_\lambda,
\end{equation}
with $\sd_\lambda$ and $\overline B_\lambda$ given in
\cref{eq:betahat_lambda_def}, and the optimization is over penalty weights such
that $\norm{w-Z\pi_{\lambda}}_{2}>0$. The optimal linear estimator and the
optimal linear \ac{CI} are then given by
\begin{equation}
  \label{eq:optimal_linear}
  \hat{\beta}_{\lambda^{*}_{\textnormal{MSE}}}\quad\text{and}\quad
  \hat{\beta}_{\lambda^{*}_{\textnormal{CI}}} \pm \cv_{\alpha}(C\overline
  B_{\lambda^{*}_{\textnormal{CI}}}/\sd_{\lambda^{*}_{\textnormal{CI}}})\sd_{\lambda^{*}_{\textnormal{CI}}}.
\end{equation}

Using the \emph{penalized} propensity score residuals as instruments for $w$ to
calculate $\hat{\beta}_{\lambda}$ can be understood as a penalized version of
the Frisch-Waugh formula. Indeed, if we vary the penalty weight such that
$t\to \infty$, $\hat{\beta}_{\lambda}$ converges to the long regression estimate
$\hat{\beta}_{\textnormal{long}}=\frac{w'(I-H_{Z})Y}{w'(I-H_{Z})w}$ that uses
the unpenalized propensity score residuals, $(I-H_{Z})w$, as instruments,
provided they are not identically zero (which ensures that the condition
$\norm{\tilde{w}_{\lambda}}_{2}>0$ in~\Cref{thm:optimization_equivalence} holds
for all $\lambda$). This estimator minimizes variance among all linear
estimators that are unbiased, so \Cref{thm:optimization_equivalence} reduces to
the Gauss-Markov theorem in this case. If we vary the penalty weight such that
$t\to 0$ and $\Pen(\cdot)$ is a norm on $\gamma_{2}$, then
$\hat{\beta}_{\lambda}$ converges to the short regression estimate
$\hat{\beta}_{\textnormal{short}}=\frac{w'(I-H_{Z_{1}})Y}{w'(I-H_{Z_{1}})w}$
that only includes the unrestricted controls $Z_{1}$. This estimator minimizes
variance among all linear estimators with finite worst-case bias. In other
words, the short and long regressions are corner solutions of the bias-variance
tradeoff, in which weight is placed entirely on variance or on bias.

\begin{example}[continues=example:l2]
  When $\Pen(\pi)=\norm{M\pi_{2}}_{2}$ for some matrix $M$, one can use the
  Lagrangian formulation for~\cref{eq:pi_optimization} to show that
  $\hat{\beta}_{\lambda}$ is equivalent to a ridge regression of $Y$ onto $X$
  that minimizes
  $\norm{Y-w\beta-Z_{1}\gamma_{1}-Z_{2}\gamma_{2}}_{2}^{2} +
  \lambda\norm{M\gamma_{2}}_{2}^{2}$. The optimality of ridge regression in this
  setting was shown by \textcite{li82}; we obtain this result as a special case
  of \Cref{thm:optimization_equivalence}.

  For the $R^{2}_{Y-w\beta}$ penalty that sets
  $M=(Z_{2}'(I-H_{Z_{1}})Z_{2}/n)^{1/2}$, the ridge estimator further
  simplifies to a weighted average of the long and short regression estimates.
  For ease of interpretation, we write the penalty bound as
  $C=\sigma \Rmax/(1-\Rmax^{2})^{1/2}$ where $\Rmax^{2}$ is the implied bound on
  $R^{2}_{Y-w\beta}$, the partial $R^{2}$ from adding $Z_{2}$ to the regression
  of the untreated outcome on $Z_{1}$. If the long regression estimator is
  infeasible (i.e., the unpenalized propensity score residual $(I-H_{Z})w$ is
  identically zero), the short regression estimator is optimal. Otherwise, to
  find the optimal bias-variance tradeoff, it is convenient to set $\lambda$ to
  correspond directly to the weight on the short regression,
  \begin{equation}\label{eq:li_mueller_beta}
    \hat{\beta}_{\lambda}
    =\lambda\hat{\beta}_{\textnormal{short}}
    + (1-\lambda)\hat{\beta}_{\textnormal{long}}.
  \end{equation}
  \Cref{thm:optimization_equivalence} then implies this estimator has worst-case
  bias given by
  \begin{equation}\label{eq:lambda_li_mueller}
    C\overline B_\lambda
    =\lambda\cdot \sigma_{\textnormal{short}}
    \sqrt{
      \frac{n \Rmax^{2} R_{w}^{2}}{1-\Rmax^{2}}},
    \quad\text{and}\quad \sd_{\lambda} =
    \sqrt{1+(1-\lambda)^{2}R_{w}^{2}/(1-R_{w}^{2})}\sigma_{\textnormal{short}},
 \end{equation}
 where $\sigma_{\textnormal{short}}=\sigma/\sqrt{w'(I-H_{Z_{1}})w}$ denotes the
 short regression standard deviation, and
 $R_{w}^{2}=1-\frac{w'(I-H_{Z})w}{w'(I-H_{Z_{1}})w}$ is the sample partial
 $R^{2}$ from adding $Z_{2}$ to the propensity score regression of $w$ onto
 $Z_{1}$. Plugging these expressions into~\cref{eq:optimal_lambda} then allows
 one to compute the optimal weight on the short regression.
\end{example}

\begin{example}[continues=example:l1]
  In this case, the solution to~\eqref{eq:pi_optimization} is given by a variant
  of the lasso estimate \parencite{tibshirani_regression_1996} that only
  penalizes $\pi_{2}$.

  The resulting estimator $\hat{\beta}_{\lambda}$ is related to estimators
  proposed for constructing \acp{CI} using the lasso \parencite[see, among
  others,][]{zhang_confidence_2014,javanmard_confidence_2014,van_de_geer_asymptotically_2014,belloni_inference_2014}.
  These papers propose estimators for $\beta$ that combine lasso estimates from
  the outcome regression of $Y$ onto $X$ with lasso estimates from the
  propensity score regression, which yields an estimate that is nonlinear in
  $Y$. In contrast, our estimator only uses lasso estimates for the propensity
  score regression, and is linear in $Y$. We give a detailed comparison between
  our estimator and this ``double lasso'' approach in
  \Cref{sec:nonconvex_efficiency}.
\end{example}

\begin{remark}[Heterogeneous treatment effects]\label{remark:hte}
  If $w_{i}$ is as good as randomly assigned conditional on $z_{i}$, the
  coefficient $\beta$ in \cref{eq:linear_regression} can be interpreted as the
  \ac{ATE} of a one-unit increase in the variable $w$. This interpretation
  requires that the individual \acp{TE} are mean independent of $z_{i}$. To
  allow for \ac{TE} heterogeneity in our framework, the simplest approach is to
  consider the model with interactions,
  $Y_{i}=w_{i}z_{i}'\delta+z_{i}'\gamma+\varepsilon_{i}$, where $z'\delta$ now
  represents the conditional \ac{ATE} for units with $z_{i}=z$. Suppose the
  parameter of interest is the weighted average of the conditional \acp{ATE},
  $\beta=\sum_{i=1}^{n}c_{i}z_{i}'\delta$, where the weights $c_{i}$ sum to one.
  With binary $w_{i}$, the unweighted case $c_{i}=1/n$ gives the (conditional on
  the sample) \ac{ATE}, while setting $c_{i}=w_{i}/\sum_{j=1}^{n}w_{j}$ gives
  the \ac{ATE} for the treated. Then, as noted by \textcite[Section
  5.3]{imbens_recent_2009}, we can rewrite the interacted model as
  \begin{equation*}
    Y_i = w_i\beta
    +w_{i}(z_{i,-1}-\sum_{j=1}^{n}
    c_{j}z_{j,-1})'\delta_{-1} + z_{i}'\gamma + \varepsilon_i,
  \end{equation*}
  where we assume that the first element of $z_{i}$ is constant, so that
  $z_{i,-1}$ is the vector of controls excluding a constant and $\delta_{-1}$ is
  the corresponding subvector of $\delta$. This is exactly in the form of
  \cref{eq:linear_regression}, with the control vector consisting of \emph{both}
  the original controls $z_{i}$ and the interaction of the treatment with the
  demeaned controls, $w_{i}(z_{i,-1}-\sum_{j=1}^{n}c_{j}z_{j,-1})$. We use this
  approach in our first empirical application in \Cref{sec:empirical-app}.
  \Cref{theorem:weighted_te} in \Cref{sec:heter-treatm-effects} shows that when
  \ac{TE} are heterogeneous, and one uses our approach without including the
  interactions, the estimator $\hat{\beta}_{\lambda}$ provides inference on a
  weighted average of treatment effects; however, some of the weights may be
  negative. \Cref{sec:heter-treatm-effects} also discusses how our approach can
  be adapted to allow the conditional \acp{ATE} to be nonlinear in $z_{i}$.
\end{remark}
\section{Implementation with non-Gaussian and heteroske\-dastic errors}\label{sec:impl-with-non}

We now discuss practical implementation issues, allowing $\varepsilon$ to be
non-Gaussian and heteroskedastic. As a baseline, we propose the following
implementation:
\begin{algorithm}[Baseline implementation]\mbox{}\label{algorithm:baseline}
  \begin{description}
  \item[Input] Data $(Y, X)$, penalty $\Pen(\cdot)$, penalty bound $C$, and
    initial estimates of residuals
    $\hat{\varepsilon}_{\textnormal{init}, 1}, \dotsc,
    \hat{\varepsilon}_{\textnormal{init}, n}$.
  \item[Output] Estimator and \ac{CI} for $\beta$.
  \end{description}
  \begin{enumerate}
  \item\label{initial_variance_estimator_step} Compute an initial variance estimator,
    $\hat{\sigma}^{2}=\frac{1}{n}\sum_{i=1}^{n}
    \hat{\varepsilon}_{\textnormal{init}, i}^2$, assuming homoskedasticity.
  \item\label{solution_path_step} For each $\lambda$, use
    \cref{eq:optimal_estimator} to compute $\hat{\beta}_{\lambda}$, and use
    \cref{eq:betahat_lambda_def} to compute $\overline{B}_{\lambda}$ and
    $\sd_{\lambda}$, with $\hat{\sigma}$ in place of $\sigma$ in the formula for
    $\sd_\lambda$. For the $R^{2}_{Y-w\beta}$ penalty, use the closed-form
    expressions in~\cref{eq:li_mueller_beta} and \cref{eq:lambda_li_mueller}.
  \item\label{standard_error_step} Compute $\lambda^*_{\textnormal{MSE}}$ and
    $\lambda^*_{\textnormal{CI}}$ as in~\cref{eq:optimal_lambda}, assuming
    homoskedasticity. Compute the robust standard error
    $\widehat{\se}_{\lambda^{*}_{\textnormal{CI}}, \textnormal{rob}}^{2} =
    \sum_{i=1}^n a_{\lambda^{*}_{\textnormal{CI}}, i}^{2}
    \hat\varepsilon_{\textnormal{init}, i}^2$.
  \end{enumerate}
  Return the estimator $\hat\beta_{\lambda^*_{\textnormal{MSE}}}$ and the
  \ac{CI}
  $\hat{\beta}_{\lambda^{*}_{\textnormal{CI}}} \pm
  \cv_{\alpha}(C\overline{B}_{\lambda^{*}_{\textnormal{CI}}}/
  \widehat{\se}_{\lambda^{*}_{\textnormal{CI}},
    \textnormal{rob}})\widehat{\se}_{\lambda^{*}_{\textnormal{CI}},
    \textnormal{rob}}$.
\end{algorithm}

The baseline implementation assumes homoskedasticity for the purpose of
computing the weights $a_{\lambda}$, but the \ac{CI} is computed using robust
standard errors to ensure asymptotic validity under heteroskedastic errors.
\Cref{sec:se_consistency} gives formal asymptotic validity results. The key
regularity condition is that the maximal Lindeberg weight
$\operatorname{Lind}(a_{\lambda^{*}_{\textnormal{CI}}})$, where
\begin{equation}\label{eq:lind_weight}
  \operatorname{Lind}(a)
  =\frac{\max_{1\leq i\leq n} a_{i}^2}{\sum_{j=1}^{n}a_{j}^2},
\end{equation}
vanishes fast enough relative to the estimation error in
$\widehat{\se}_{\lambda^{*}_{\textnormal{CI}}, \textnormal{rob}}$ (as measured
by the $\ell_{2}$ convergence rate of $\hat\varepsilon_{\textnormal{init}}$).
\Cref{sec:se_consistency} also discusses the Lindeberg weight diagnostics in our
empirical applications (specifically, see \Cref{remark:empirical-lindeberg}), as
well as a modification of the baseline algorithm that ensures this regularity
condition holds. When $C=\infty$ or $C=0$, the baseline implementation reduces
to the ubiquitous practice of reporting (long or short regression) \ac{OLS} with
heteroskedasticity-robust standard errors. While it is possible to exploit
heteroskedasticity for computing the optimal weights $a_{\lambda}$, we do not
pursue the generalization here.

The following remarks discuss the implementation choices in more detail.

\begin{remark}[Choice of initial residuals]\label{remark:residuals} Following
  the analysis in \Cref{sec:se_consistency}, we can take residuals from a
  regularized outcome regression of $Y$ on $X$ as the initial residual estimates
  $\hat\varepsilon_{\textnormal{init}, i}$. In our simulations and empirical
  applications, we use the lasso, with penalty chosen via 10-fold
  cross-validation.
\end{remark}
\begin{remark}[Choice of $C$]\label{choice_of_C_remark}
  To help calibrate the penalty bound $C$, we follow the recent sensitivity
  analysis literature
  \parencite{altonji_selection_2005,oster_unobservable_2019,CiHa20,diegert2026}
  that advocates for calibrating the importance of unobservables by relating
  them to observable analogues. We focus on the setting in
  \Cref{example:l2,example:l1} in which only a subset of the controls is
  penalized; the penalized controls $Z_{2}$ then play a role analogous to the
  unobserved confounders in the sensitivity analysis literature. For the
  $\ell_{1}$ penalty, we calibrate $C$ by setting it to
  $\rho\norm{\hat{\gamma}_{\textnormal{short}}}_{1}$ for some $\rho$, where
  $\hat{\gamma}_{\textnormal{short}}$ denotes the \ac{OLS} coefficient on the
  baseline controls $Z_{1}$ (with the intercept excluded).\footnote{Formally,
    one should account for sampling uncertainty in
    $\hat{\gamma}_{\textnormal{short}}$ to ensure validity of the \ac{CI} under
    the assumption
    $\norm{\gamma_{2}}_{1}\leq \rho\norm{\gamma_{\textnormal{short}}}_{1}$,
    where $\gamma_{\textnormal{short}}$ denotes the population short regression
    coefficient. For instance, one can combine a first-stage CI for
    $\norm{\gamma_\text{short}}_{1}$ with a Bonferroni correction. In our
    simulations in \Cref{sec:simulation-results}, however, we find that our
    calibration leads to valid coverage when this assumption holds even without
    additional corrections for sampling uncertainty.} This bounds the total
  effect of the penalized controls on the outcome in the long regression to be
  at most $\rho$ times the total effect of $Z_{1}$ in the short regression. As a
  conservative baseline specification capturing the belief that the baseline
  controls are more important, we propose setting $\rho=1$. This is analogous to
  the common practice in the sensitivity analysis literature to assume the
  ``effect'' of unobservables is no larger than the ``effect'' of observables
  \parencite[e.g.][]{altonji_selection_2005,oster_unobservable_2019}; here we
  formalize ``effect'' by the $\ell_{1}$ norms in the short and long
  regression.\footnote{If one wishes to penalize all controls, one can select
    the $\ell_{1}$ penalty bound by selecting a subset of controls $Z_{1}$ that
    one deems most important, and set
    $C=\rho\norm{\hat{\gamma}_{\textnormal{short}}}_{1}$ for some $\rho$. A
    natural default calibration is $\rho=2$, so that the total effect of
    \emph{all} controls in the long regression is at most twice as large as the
    effect of the important controls in the short regression.}

  For the $R^{2}_{Y-w\beta}$ penalty (or other versions of the $\ell_{2}$
  bound), our theoretical analysis in \Cref{sec:theor-prop} implies that an
  analogous calibration will typically not yield substantive improvements over
  the long regression if the long regression is feasible. Specifically, consider
  implementing the $R^{2}_{Y-w\beta}$ penalty by bounding
  $R^{2}_{Y-w\beta}\leq \Rmax^{2}$ for some bound $\Rmax^{2}$. We show in
  \Cref{sec:scope-impr-over} that ensuring gains over the long regression
  requires setting $\Rmax^{2}\lesssim 1/n$, which has two implications. First,
  the implied values of $\Rmax^{2}$ necessary for substantive improvements over
  the long regression are unreasonably small once the sample size is in the
  thousands, as is the case in our second empirical application in
  \Cref{sec:right-heart-cath}; we therefore do not recommend using
  $R^{2}_{Y-w\beta}$ bounds or other versions of the $\ell_{2}$ penalty in such
  cases. Second, setting $\Rmax^{2}$ to, say, the short regression analogue will
  have little bite even when the sample size is only in the hundreds.
  Nonetheless, in some applications with moderate sample sizes, it may be
  possible to use direct arguments for setting $\Rmax^{2}$ to smaller values,
  allowing for nonnegligible gains over the long regression. For example, if the
  additional controls are transformations of the baseline controls as in
  functional form robustness checks, one may set $\Rmax^{2}$ to a fraction of
  the short regression $R^{2}$, formalizing the empirical regularity that
  higher-order terms have less explanatory power than main effects. We
  illustrate this in our first empirical application in
  \Cref{sec:marg-prop-cons}.

  To complement the analysis under default choices of $\rho$ or $\Rmax^{2}$, we
  recommend varying them by way of sensitivity analysis, and reporting a
  ``breakdown value'', the largest value such that the empirical finding of
  interest holds. This is analogous to the ubiquitous practice of reporting
  breakdown values for sensitivity parameters in the sensitivity analysis
  section of empirical papers.

  Careful \emph{a priori} calibration of the penalty bound is necessary: by
  \Cref*{optimality_corollary}\ref*{item:flci} in \Cref*{sec:linear_efficiency}
  below, one cannot start with an initial conservative choice for $C$ and hope
  to substantively improve upon the \ac{CI} based on this conservative choice by
  using a data-driven rule for picking $C$, without sacrificing coverage.
  However, it is possible to form a lower \ac{CI}
  $\hor{\hat{\underline C}, \infty}$ for $C$ as a specification check to assess
  the plausibility of a given penalty bound. We present such a \ac{CI} in
  \Cref{sec:C_lower_CI} for the case where $\Pen(\gamma)$ takes the form of an
  $\ell_p$ constraint.
\end{remark}
\begin{remark}[Computational issues]\label{remark:computation}
  Step~\ref{solution_path_step} involves computing the solution path of a
  regularized regression estimator. Efficient algorithms exist for computing
  these paths under the $\ell_{1}$ penalty and its variants
  \parencite{efron2004lars,rosset_piecewise_2007}. Under the $R^{2}_{Y-w\beta}$ penalty,
  the estimator $\hat{\beta}_{\lambda}$ has a closed form, so that our algorithm
  can again be implemented in a computationally efficient manner. For other
  types of penalties, the convexity of the optimization problem in
  \cref{eq:pi_optimization} can be exploited to yield efficient implementation.
  We also note that since the solution path $\pi_{\lambda}$ does not depend on
  $C$, it only needs to be computed once, even when multiple choices of $C$ are
  considered in a sensitivity analysis.
\end{remark}

\section{Theoretical properties}\label{sec:theor-prop}

The optimal linear estimator and \ac{CI} in \cref{eq:optimal_linear} trade off
bias and variance optimally among procedures that are linear in the outcomes
$Y$. \Cref{sec:linear_efficiency} shows that this estimator and \ac{CI} are in
fact highly efficient among all procedures, not just linear ones.
\Cref{sec:scope-impr-over} shows that reasonable calibrations of the penalty
$R^{2}_{Y-w\beta}$ yield slow convergence rates when the regression is
high-dimensional, and that achieving meaningful improvements over the long
regression requires imposing $R^{2}_{Y-w\beta}\leq \Rmax^{2}$ for some
$\Rmax^{2}\lesssim 1/n$. \Cref{sec:asymptotics} derives analogous results for
the unweighted $\ell_{2}$ penalty and shows that, by contrast, the rate results
under the $\ell_{1}$ penalty are much better. As in
\Cref{sec:finite-sample-result}, this section assumes that the errors
$\varepsilon$ are Gaussian and homoskedastic. We use the notation
$R_{\textnormal{CI}}(C)=\cv_{\alpha}(C\overline{B}_{\lambda^{*}_{\textnormal{CI}}}
/\sd_{\lambda^{*}_{\textnormal{CI}}})\sd_{\lambda^{*}_{\textnormal{CI}}}$ to
denote the half-length of the optimal linear \ac{CI}.

\subsection{Efficiency among nonlinear procedures}\label{sec:linear_efficiency}

To state the efficiency result for the optimal linear estimator and \ac{CI} in
\cref{eq:optimal_linear}, let $\var_{\beta,\gamma}(\hat{\beta})$ denote the
variance of any estimator $\hat{\beta}=\hat{\beta}(X,Y)$, and let
\begin{equation*}
  \RMSE(\hat{\beta};\Gamma)
  =\sup_{\beta\in\mathbb{R}, \gamma\in\Gamma}E_{\beta, \gamma}[(\hat{\beta}-\beta)^{2}]
\end{equation*}
denote its worst-case \ac{MSE} performance over the parameter space, where
$E_{\beta, \gamma}$ denotes expectation under $(\beta, \gamma')'$. We refer to
the interval $[\hat{\beta}-\hat{\chi},\hat{\beta}+\hat{\chi}]$ with half-length
$\hat{\chi}=\hat{\chi}(X,Y)$ as a $1-\alpha$ \ac{CI} for $\beta$ if it satisfies
the coverage requirement
\begin{equation*}
  \inf_{\beta\in\mathbb{R}, \gamma\in\Gamma}
  P_{\beta, \gamma}\left(\abs{\hat{\beta}-\beta}\leq \hat{\chi}\right)\ge 1-\alpha,
\end{equation*}
where $P_{\beta, \gamma}$ denotes probability under $(\beta, \gamma')'$.

\begin{corollary}\label{optimality_corollary}
  Let the optimal linear estimator and the optimal linear \ac{CI} be given
  in~\cref{eq:optimal_linear}. Define $\hat\beta_\lambda$,
  $\overline{B}_\lambda$ and $\sd_{\lambda}$ as in
  \Cref{thm:optimization_equivalence}. Let $\tilde \beta$ and
  $\tilde\beta\pm \tilde\chi$ denote some other, possibly nonlinear, estimator
  and some other $1-\alpha$ \ac{CI}, also possibly nonlinear in $Y$. Let
  $\maxbias_{\Gamma}(\tilde\beta)=\sup_{\beta\in\mathbb{R},\gamma\in\Gamma}\abs{E[\tilde{\beta}-\beta]}$
  denote its worst-case bias.
  \begin{enumerate}[label=({\roman*})]
  \item\label{item:bias-variance} For any $\lambda$,
    $\sup_{\beta\in\mathbb{R}, \gamma\in\Gamma} \var_{\beta,
      \gamma}(\tilde\beta) \leq \sd^{2}_\lambda$ implies
    $\maxbias_{\Gamma}(\tilde\beta) \geq C \overline B_\lambda$, and
    $\maxbias_{\Gamma}(\tilde\beta) \leq C \overline B_\lambda$ implies
    $\sup_{\beta\in\mathbb{R}, \gamma\in\Gamma} \var_{\beta, \gamma}(\tilde\beta)
    \geq \sd_\lambda^{2}$.
  \item\label{item:mse} The worst-case \ac{MSE} improvement of $\tilde\beta$
    over $\hat\beta_{\lambda^*_{\textnormal{MSE}}}$ is bounded by
    \begin{equation*}
    \frac{\RMSE(\tilde\beta)}{\RMSE(\hat\beta_{\lambda^*_{\textnormal{MSE}}})}\ge
    \kappa^*_{\textnormal{MSE}}(X, \sigma, \Gamma)\ge 0.8,
  \end{equation*}
  where $\kappa^*_{\textnormal{MSE}}(X, \sigma, \Gamma)$ is given in
  \Cref{sec:proof-coroll-optimality}.
\item\label{item:flci} The ratio of the expected length of the \ac{CI}
  $\tilde\beta\pm\tilde \chi$ to that of the optimal linear \ac{CI} at
  $\gamma=0$ and any $\beta$ is bounded below by
    \begin{equation*}
      \frac{E_{\beta,0}[\tilde\chi]}{R_{\textnormal{CI}}(C)}
      \ge \kappa^*_{\textnormal{CI}}(X, \sigma, \Gamma,\alpha),
  \end{equation*}
  where $\kappa^*_{\textnormal{CI}}(X, \sigma, \Gamma,\alpha)$ is given in
  \Cref{sec:proof-coroll-optimality} and is at least $0.717$ when $\alpha=0.05$.
  \end{enumerate}
\end{corollary}

\Cref{optimality_corollary} follows from the general results in
\textcite{donoho94}, \textcite{low95} and \textcite{ArKo18optimal} for
estimation of linear functionals in Gaussian models with convex parameter
spaces. By construction, the estimator $\hat{\beta}_{\lambda}$ minimizes
variance among all linear estimators with a bound $C\overline{B}_{\lambda}$ on
the bias (or equivalently, it minimizes bias among all linear estimators with a
bound $\sd_{\lambda}^{2}$ on the variance).
\Cref{optimality_corollary}\ref{item:bias-variance} shows that this optimality
property is retained if we enlarge the class of estimators to all estimators,
including nonlinear ones. As a result, the optimal linear estimator
$\hat{\beta}_{\lambda^{*}_{\textnormal{MSE}}}$ continues to perform well among
all estimators, including nonlinear ones: by
\Cref{optimality_corollary}\ref{item:mse}, its worst-case \ac{MSE} efficiency is
at least 80\%. The exact efficiency bound
$\kappa^*_{\textnormal{MSE}}(X, \sigma, \Gamma)$ depends on the design matrix,
noise level, and particular choice of the parameter space, and can be computed
explicitly in particular applications. We have found that typically the
efficiency is considerably higher than the 80\% lower bound.

Finally, \Cref{optimality_corollary}\ref{item:flci} shows that it is not
possible to substantively improve upon the optimal linear \ac{CI} in terms of
expected length when $\gamma=0$, even if we ``direct power'' at $\gamma=0$
(potentially at the expense of longer expected length when $\gamma\neq 0$). The
construction of the \ac{CI} may appear conservative: its length depends on the
worst-case bias over the parameter space for $\gamma$, which, as the proof of
\Cref{thm:optimization_equivalence} shows, is attained when $\Pen(\gamma)=C$.
Therefore, one may think that it may be possible to improve the \ac{CI}'s
performance when in fact $\Pen(\gamma)\ll C$.
\Cref{optimality_corollary}\ref{item:flci} shows that this is not the case, even
in the best-case scenario in which $\gamma=0$. This generalizes the classic
result that when the long regression is feasible and we set $C=\infty$,
one-sided $t$-tests based on the long regression are uniformly most powerful.
The exact efficiency bound
$\kappa^*_{\textnormal{CI}}(X, \sigma, \Gamma, \alpha)$ can be computed
explicitly in particular applications, and we have found that it is typically
considerably higher than $71.7\%$.

A consequence of \Cref{optimality_corollary}\ref{item:flci} is that it is
impossible to form a CI that is adaptive with respect to the penalty bound $C$.
In the present setting, an adaptive CI would have length that automatically
reflects the true regularity $\Pen(\gamma)$ while maintaining coverage under a
conservative \emph{a priori} bound on $\Pen(\gamma)$. However, according to
\Cref{optimality_corollary}\ref{item:flci}, any \ac{CI} must have expected width
that reflects the conservative \emph{a priori} bound $C$ rather than the true
regularity $\Pen(\gamma)$, even when $\Pen(\gamma)$ is much smaller than the
conservative \emph{a priori} bound $C$. This is why careful \emph{a priori}
calibration of the penalty bound $C$ is necessary, as discussed in
\Cref{choice_of_C_remark}.

\subsection{Efficiency and rates of convergence under the
  \texorpdfstring{$R^{2}_{Y-w\beta}$}{R2} penalty}\label{sec:scope-impr-over}

We now give two results that limit the scope of applicability of the
$R^{2}_{Y-w\beta}$ penalty. First, we show that when the long regression is
feasible, meaningful improvements are only possible for small sample sizes.
Second, we show that in high-dimensional settings, reasonable calibrations of
the penalty bound yield slow convergence rates. As in \Cref{choice_of_C_remark},
for ease of interpretation, rather than bounding $C$, we bound
$R^{2}_{Y-w\beta}\leq \Rmax^{2}$ for some bound $\Rmax^{2}$.

For the first result, assume that $R_{w}^{2}$, the partial $R^{2}$ from adding
$Z_{2}$ to the propensity score regression, does not equal $1$, so that the long
regression is feasible. It then follows from \cref{eq:lambda_li_mueller} and the
definition of the optimal linear \ac{CI} in~\cref{eq:optimal_linear} that its
length relative to the long regression \ac{CI}
$\hat{\beta}_{\textnormal{long}}\pm z_{1-\alpha/2}
\sigma_{\textnormal{short}}/\sqrt{1-R_{w}^{2}}$ is given by
\begin{equation}\label{eq:relative_efficiency_lr}
  \operatorname{RE}(\kappa,R_{w}^{2};\alpha)
  =
  \min_{\lambda\in[0,1]}\cv_{\alpha}\left(
    \kappa    \frac{\lambda\cdot
      \sqrt{R_{w}^{2}(1-R_{w}^{2})}}{\sqrt{R_{w}^{2}(1-\lambda)^{2}+1-R_{w}^{2}}}
  \right)\frac{  \sqrt{R_{w}^{2}(1-\lambda)^{2}+1-R_{w}^{2}}}{z_{1-\alpha/2}},
\end{equation}
where
$\kappa^{2}=\kappa^{2}(n,\Rmax^{2})=n\Rmax^{2}/(1-\Rmax^{2})$
and $z_{1-\alpha/2}$ is the usual critical value corresponding to the
$1-\alpha/2$ quantile of a standard normal distribution. In particular, the
relative efficiency depends on the data only through $R_{w}^{2}$, the partial
$R^{2}$ from adding $Z_{2}$ to the propensity score regression. Note the maximal
possible efficiency gain is capped at the ratio of short and long regression
standard errors, $\operatorname{RE}(0,R_{w}^{2};\alpha)=\sqrt{1-R_{w}^{2}}$. It
is achieved when we set $\Rmax^{2}=0$ (which amounts to
assuming that $\gamma_{2}=0$), in which case the short regression is optimal and
unbiased.

Before discussing large-sample implications of the formula in
\cref{eq:relative_efficiency_lr}, we first use it to derive a simple bound on
how small $\kappa$ and hence the sample size needs to be in order to achieve a
meaningful improvement over the long regression. In particular, suppose we wish
to improve over the long regression by at least $100r\%$, so that
$\operatorname{RE}(\kappa,R_{w}^{2};\alpha)\leq 1-r$. Using the fact that
$\cv_{\alpha}(B)\geq B+z_{1-\alpha}$ to lower bound the left-hand side of this
inequality and solving for $\kappa$ then yields, for
$r> 1-z_{1-\alpha}/z_{1-\alpha/2}$ ($\approx 16\%$ when $\alpha=0.05$)
\begin{equation}\label{eq:kappa_improvement}
  \kappa
  \leq (1-r) z_{1-\alpha/2} \sqrt{\frac{R_{w}^{2}}{1-R_{w}^{2}}}- \sqrt{z_{1-\alpha}^{2} - (1-r)^{2} z_{1-\alpha/2}^{2}}.
\end{equation}
Consider a favorable scenario in which $R_{w}^{2}=0.96$, so that the long
regression standard errors are 5 times as big as the short regression standard
errors. Then, to achieve 20\% improvement over the long regression for
$\alpha=0.05$, by \cref{eq:kappa_improvement}, we require $\kappa\leq 7.2$,
which requires $n\leq 980$ for $\Rmax^{2}=0.05$ and $n\leq 5,100$ for
$\Rmax^{2}=0.01$. This suggests that once the sample size is in the thousands,
reasonable calibrations for $\Rmax^{2}$ will not yield appreciable efficiency
gains. For example, in our second empirical application in
\Cref{sec:right-heart-cath}, $R_{w}^{2}\approx 0.71$ and $n=5{,}735$, so
\cref{eq:kappa_improvement} implies that a necessary condition for a 20\%
efficiency gain is $\Rmax^{2}\leq 0.0007$.

For modest sample sizes, meaningful efficiency gains are possible if the long
regression is appreciably noisier than the short regression. In our first
empirical example in~\Cref{sec:marg-prop-cons}, $R_{w}^{2}\approx 0.76$ and
$n=496$, yielding the corresponding necessary condition $\Rmax^{2}\leq 0.011$.
Thus, the moderate sample size leaves more scope for gains than in the second
application.

More generally, \cref{eq:kappa_improvement} implies that improving upon the long
regression requires setting $\Rmax^{2}\lesssim 1/(n(1-R_{w}^{2}))$, or
equivalently $C\lesssim 1/\sqrt{n(1-R_{w}^{2})}$. This yields the rate
$\Rmax^{2}\lesssim 1/n$ when $R_{w}^{2}$ is bounded away from one. Under
asymptotics where $R_{w}^{2}\to 1$ (such as when $k\to\infty$, and the variance
of the propensity score residual shrinks as we add more controls, which captures
the notion that the long regression is feasible, but noisy), $\Rmax^{2}$ may
shrink more slowly, though reasonable calibrations of $\Rmax^{2}$ will typically
still be uninformative as the finite-sample bounds for $\kappa$ above
illustrate.

Combining the relative efficiency calculations above with the observation that
when $R_{w}^{2}=1$, the optimal linear \ac{CI} is based on the short regression
alone then yields the second result, which gives the rate of convergence of the
optimal linear estimator.

\begin{theorem}\label{theorem:li_mueller_rate}
  Suppose that $\eta/n \leq \sigma^{2}/w'(I-H_{Z_{1}})w\leq 1 /(n\eta)$ for
  some small constant $\eta>0$, and that $R^{2}_{Y-w\beta} \leq \Rmax^{2}$ for
  some $\Rmax^{2}<1$. Let $\kappa^{2}=n\Rmax^{2}/(1-\Rmax^{2})$, and define
  $C\overline{B}_{\lambda}$ and $\sd_{\lambda}$ as in
  \cref{eq:lambda_li_mueller}. Then there exists a finite constant
  $K_{\eta,\alpha}$ depending only on $\eta$ and $\alpha$ such that
  \begin{equation*}
    \frac{1}{K_{\eta,\alpha}}
    \frac{1+\min\{\kappa,1/\sqrt{1-R_{w}^{2}}\}}{\sqrt{n}}
    \leq     R_{\textnormal{CI}}(C)\leq K_{\eta,\alpha}
    \frac{1+\min\{\kappa,1/\sqrt{1-R_{w}^{2}}\}}{\sqrt{n}}.
  \end{equation*}
\end{theorem}

We note that \Cref{theorem:li_mueller_rate} does not restrict $R_{w}^{2}$ (and
thus accommodates scenarios where $R_{w}^{2}$ converges to, or equals one), or
how fast $k$ may grow relative to $n$.
In the case where $R_{w}^2=1$ (which holds generically when $k\geq n$), the term
$\min\{\kappa,1/\sqrt{1-R_{w}^2}\}$ reduces to
$\kappa=\sqrt{n\Rmax^2/(1-\Rmax^2)}$ and the CI length is proportional to
$\sqrt{\Rmax^2/(1-\Rmax^2)}$ times a constant that does not depend on $n$.
Thus, consistent estimation requires taking the bound $\Rmax^2$ to zero with the
sample size.
This slow rate of convergence is not a feature of the particular way the $R^{2}$
penalty normalizes the $\ell_{2}$ norm of $\gamma$---in the next subsection we
give matching rates of convergence for the optimal linear \ac{CI} under the
unweighted version of the $\ell_{2}$ penalty.

\subsection{Rates of convergence under unweighted \texorpdfstring{$\ell_{p}$}{lp} norms}\label{sec:asymptotics}

For ease of notation, we assume all coefficients are constrained, so that
$\Pen(\gamma)=\norm{\gamma}_{p}$. We consider asymptotics under which $k$ grows
at least as fast as $n$; if $k/n\to 0$ then the optimal rate is $1/\sqrt{n}$ and
it can be achieved by the long regression.

Let $q\in[1,\infty]$ denote the Hölder conjugate of $p$, satisfying $1/p+1/q=1$.
We will show that when $\Pen(\gamma)=\norm{\gamma}_{p}$, the rate of convergence
for the optimal linear \ac{CI} is at least
\begin{equation}\label{eq:convergence-rates}
  R_{p}^{*}(k,C,R_{w}^{2})=
  \frac{1+\min\{C r_{p}(k),1/\sqrt{1-R_{w}^{2}}\}}{\sqrt{n}},\quad\text{where}\quad
  r_{p}(k) =
  \begin{cases}
    k^{1/q} & \text{if $p>1$,} \\
    \sqrt{\log k} & \text{if $p=1$.}
  \end{cases}
\end{equation}
Furthermore, for $p=1$ and $p=2$, we will give lower bounds that match this rate
for a range of conditions on how $(C,k)$ changes with the sample size. If
$k\gg n$ and $C$ does not decrease to zero with $n$, these rates require $p<2$
(so that $q>2$) for consistency. When $p=1$, we can then allow $k$ to grow
exponentially with $n$, whereas setting $1<p<2$ allows for $k$ to grow at a
polynomial rate in $n$ that depends on $p$. We note that when the optimal rate
is slower than $n^{-1/2}$, the bias term is not asymptotically negligible,
making it necessary to use our bias-aware critical values even in large samples.

\begin{remark}\label{remark:l1_vs_l2}
  These stark differences in convergence rates can be explained by the fact that
  the constraint $\norm{\gamma}_{p}\leq C$ becomes more stringent with smaller
  $p$, since $\norm{\gamma}_{p}\leq \norm{\gamma}_{r}$ for $r<p$. Specifically,
  consider an $\ell_{p}$ constraint
  $\{\gamma\in\mathbb{R}^{k}\colon \norm{\gamma}_{p}\leq C_{p}\}$, for some
  $p>1$, but with the bound $C_{p}$ calibrated so that the parameter space has
  the same volume as the parameter space under the $\ell_{1}$ constraint,
  $\{\gamma\in\mathbb{R}^{k}\colon \norm{\gamma}_{1}\leq C_{1}\}$, for some
  fixed bound $C_{1}$. This ensures that the strength of the parameter space
  restriction is the same under the $\ell_{p}$ as under the $\ell_{1}$
  constraint. As $k\to\infty$ this requires setting
  $C_{p}\asymp k^{1/p-1}$.\footnote{To see this, note the volume of a unit
    $\ell_{p}$ ball is given by $(2\Gamma(1+1/p))^{k}/\Gamma(1+k/p)$, so that
    $C_{p}=
    \frac{C_{1}}{\Gamma(1+1/p)}\left(\frac{\Gamma(1+k/p)}{\Gamma(1+k)}\right)^{1/k}
    =\frac{C_{1}}{\Gamma(1+1/p)}\frac{e^{1-1/p}}{p^{1/p}}k^{1/p-1}(1+O(1/k))$.}
  In the high-dimensional setting with $R_{w}^{2}=1$, this implies the
  convergence rate
  $R_{p}^{*}(k,C_{p},1)\asymp C_{p}r_{p}(k)/\sqrt{n}\asymp 1/\sqrt{n}$, which is
  the same rate as the convergence rate under the $\ell_{1}$ constraint, up to
  log factors. In other words, if $C_{p}$ is calibrated so that the strength of
  the $\ell_{p}$ constraint is the same across different values of $p$, the
  convergence rates are also the same.
  While these calculations are useful for understanding the differences in
  convergence rates, we do not mean to suggest that the volume of a parameter
  space should be used as a metric to gauge its plausibility in practice: since
  taking $C\to 0$ rules out even a single coefficient being bounded away from
  zero, imposing an $\ell_{1}$ constraint with a fixed $C$ appears more
  defensible than imposing an $\ell_{2}$ constraint, but taking $C\to 0$.
\end{remark}

To derive the upper bound, we assume that the design matrix $X$ is well-behaved
in that $X\in\mathcal{E}^{p}_{n}(\eta)$, where $\mathcal{E}^{p}_{n}(\eta)$
denotes the set of design matrices $X$ for which there exists
$\delta\in\mathbb{R}^{k}$ such that
\begin{align*}
  \frac{1}{n}\norm{w-Z\delta}_{2}^{2}&\leq \frac{1}{\eta},
  & \frac{1}{n}w'(w-Z\delta)&\geq \eta,
  & \frac{1}{n}\norm{Z'(w-Z\delta)}_{q}
  & \le \frac{r_{p}(k)}{\sqrt{n}\eta}.
\end{align*}
\Cref{upper_bound_event_lemma} in~\Cref{sec:verif-high-level} shows that
$X\in\mathcal{E}_{n}^{p}(\eta)$ holds with high probability for $\eta>0$ small
enough when $w_{i}, Z_{i}$ are drawn i.i.d.\ from a distribution satisfying mild
conditions on moments and covariances if $\delta$ is defined as the coefficient
from the population version of the propensity score regression, so long as the
population propensity score residual is not degenerate. For the condition
$X\in\mathcal{E}^{2}_{n}(\eta)$, a second sufficient condition that allows for a
population propensity score residual that is asymptotically degenerate is that
$(w_{i},Z_{i})$ are drawn i.i.d.\ from a distribution such that $w_{i}$ is not
degenerate and has bounded second moments, and $Z_{i}$ is sub-Gaussian. Then
$X\in\mathcal{E}^{2}_{n}(\eta)$ with high probability for $\eta>0$ small enough
if we let $\delta=0$. This follows by the operator norm bound,
$\norm{Z'w}_{2}\leq \norm{Z}_{2}\norm{w}_{2}$, and the standard result in random
matrix theory that $\norm{Z}_{2}\lesssim \sqrt{k}$ \parencite[e.g.,][Theorems
5.39 and 5.58]{vershynin12}.

\begin{theorem}\label{thm:upper_rate_bound}
  Suppose that for some constant $\eta>0$, $\eta \leq w'w/n\leq 1/\eta$,
  $\eta\leq\sigma\leq 1/\eta$, and $X\in\mathcal{E}^{p}_{n}(\eta)$. Set
  $\Pen(\gamma)=\norm{\gamma}_{p}$. Then there exists a finite constant
  $K_{\eta,\alpha}$ depending only on $\eta$ and $\alpha$ such that
  $R_{\textnormal{CI}}(C)\le K_{\eta, \alpha}R_{p}^{*}(k,C,R_{w}^{2})$.
\end{theorem}

We now give a matching lower bound for $p=2$. To state the result, let
$d_{\max}$ denote the largest singular value of $Z$, and let $d_{\min}$ denote
its smallest non-zero singular value.

\begin{theorem}\label{theorem:l2_rates}
  Suppose that for some constant $\eta>0$, $\eta \leq w'w/n\leq 1/\eta$,
  $\eta\leq\sigma\leq 1/\eta$, $d_{\max}\geq \sqrt{k}\eta$ and
  $d_{\max}/d_{\min}\leq 1/\eta$. %
  Set $\Pen(\gamma)=\norm{\gamma}_{2}$. %
  Then there exists a finite constant $K_{\eta,\alpha}$ depending only on $\eta$
  and $\alpha$ such that
  \begin{equation*}
    R_{\textnormal{CI}}(C)\geq  \frac{R_{2}^{*}(k,C,R_{w}^{2})}{K_{\eta,\alpha}}.
  \end{equation*}
\end{theorem}

If $Z$ has independent sub-Gaussian rows, standard results in random matrix
theory \parencite[e.g.,][Theorems 5.39 and 5.58]{vershynin12} imply that the
conditions on the singular values of $Z$ hold with probability approaching one
provided that the limit of $k/n$ lies outside the interval
$[1/\tilde{\eta},\tilde{\eta}]$, where $\tilde{\eta}$ depends only on the
sub-Gaussian norm of the rows. If $Z$ has i.i.d.\ Gaussian rows, it is
sufficient to take any $\tilde{\eta}>1$ \parencite[e.g.,][Corollary
5.35]{vershynin12}. Furthermore, these conditions have also been shown to hold
for more general distributions and $\tilde{\eta}$ large enough, provided that
$Z$ has finite fourth moments \parencite[e.g.,][]{SrVe13}.

Hence, for a large class of design matrices $X$, the convergence rate of the
optimal linear CI under the unweighted $\ell_{2}$ penalty matches that under the
$R^{2}_{Y-w\beta}$ penalty. In particular, rate improvements upon the long
regression require setting $C\lesssim 1/\sqrt{k(1-R_{w}^{2})}$. If
$R_{w}^{2}=1$, consistency requires taking $C\to 0$, which rules out even a
single coefficient being bounded away from zero.

We next consider a lower bound for the case $p=1$. Rather than imposing
conditions on $X$ in a fixed design setting that hold with high probability
under random design, we directly consider a random design setting, and we do not
condition on $X$ when requiring coverage of \acp{CI}. This allows us to
strengthen the conclusion of our theorem by showing that the rate in
\Cref{thm:upper_rate_bound} is sharp even if one imposes a linear model for
$w_i$ given $z_i$ along with sparsity and $\ell_1$ bounds on the coefficients in
this model. The random design model will also imply that $R_{w}^{2}=1$ for $n$
large enough when $k\geq n$. In particular, we assume that
\begin{align*}
  Y&=w\beta+Z\gamma + \varepsilon, \quad \varepsilon\mid Z, w \sim \ND(0,\sigma^2I_n), \\
  w&=Z\delta + v, \quad v\mid Z \sim \ND(0,\sigma_{v}^2 I_n), \\
  z_{ij} &\sim \ND(0,1)\quad\text{i.i.d.\ over $i, j$}.
\end{align*}
We use $P_{\vartheta}$ and $E_{\vartheta}$ for probability and expectation when
$Y, X$ follow this model with parameters
$\vartheta=(\beta, \gamma', \delta', \sigma^2,\sigma^2_{v})'$. Let
$\sigma_0^2>0$ and $\sigma^2_{v, 0}>0$ be given and let $\Theta(C, s, \eta)$
denote the set of parameters
$\vartheta=(\beta, \gamma', \delta', \sigma^2, \sigma^2_{v})$ where
$\abs{\sigma^2-\sigma^2_0}\le \eta$,
$\abs{\sigma_{v}^2-\sigma^2_{v, 0}}\le \eta$, $\norm{\gamma}_{1}\le C$,
$\norm{\delta}_{1}\le C$, $\norm{\gamma}_0\le s$ and $\norm{\delta}_0\le s$.

\begin{theorem}\label{thm:l1_lower_bound}
  Let $\hat\beta\pm \hat\chi$ be a \ac{CI} satisfying
  $P_{\vartheta}(\abs{\beta-\hat{\beta}}\leq\hat{\chi}) \geq 1-\alpha$ for all
  $\vartheta$ in $\Theta(C, C\cdot K \sqrt{n/\log k}, \eta)$ where $\alpha<1/2$.
  Suppose $k\to\infty$, $C\sqrt{(\log k)/n}\to 0$, and
  $c_{0}\sqrt{(\log k)/n}\le C\le \sqrt{k/n}\cdot k^{-\tilde \eta}$ for some
  $c_{0}>0$ and $\tilde{\eta}>0$. Then, there exists $c$ such that, if $K$ is
  large enough and $\eta\to 0$ slowly enough, the expected length of this
  \ac{CI} under the parameter vector $\vartheta^*$ given by $\beta=0$,
  $\gamma=0$, $\delta=0$, $\sigma^2=\sigma^2_0$, $\sigma^2_{v}=\sigma^2_{v, 0}$
  satisfies $E_{\vartheta^*}[\hat\chi] \ge c\cdot n^{-1/2}(1+C\sqrt{\log k})$
  once $n$ is large enough.
\end{theorem}

\Cref{thm:l1_lower_bound} follows from similar arguments to
\textcite{cai_lasso_2017} and \textcite{javanmard_debiasing_2018}, who provide
similar bounds for the case where only a sparsity bound is imposed. According to
\Cref{thm:l1_lower_bound}, imposing sparsity does not allow one to improve upon
the \ac{CI} that uses only the $\ell_1$ bound $\norm{\gamma}_1\le C$ (thereby
attaining the rate in \Cref{thm:upper_rate_bound}), unless one imposes a
sparsity bound that is smaller than $C\sqrt{n/\log k}$. We provide further
comparison with \acp{CI} that impose sparsity in the next section.

The condition $C\le \sqrt{k/n}\cdot k^{-\tilde{\eta}}$ in
\Cref{thm:l1_lower_bound} places a lower bound on the rate at which $k$ must
increase. If $C$ is a positive constant, it requires that we take
$\tilde{\eta}<1/2$ and that $k\gtrsim n^{1/(1-2\tilde{\eta})}$, so that $k$ must
grow faster than $n$. If $k/n$ converges to a constant that is strictly smaller
than one, the optimal convergence rate is $1/\sqrt{n}$ and it can be achieved by
the long regression. We conjecture that if $k/n$ converges to a constant greater
than one, or if $k/n$ diverges more slowly than any positive power of $n$, the
optimal rate is between $1/\sqrt{n}$ and $\sqrt{\log k}/\sqrt{n}$. However,
showing this would require a different argument. To our knowledge, this case is
also not covered in the literature on sparse high-dimensional regression.

\section{Comparison with sparsity constraints}\label{sec:nonconvex_efficiency}

Several authors have considered \acp{CI} for $\beta$ using ``double lasso''
estimators \parencite[see, among
others,][]{belloni_inference_2014,javanmard_confidence_2014,van_de_geer_asymptotically_2014,zhang_confidence_2014}.
These \acp{CI} are valid under the parameter space
\begin{equation}\label{eq:sparsity_parameter_space}
  \widetilde{\Gamma}(s) = \{\gamma\colon\norm{\gamma}_0\le s\},
\end{equation}
where $\norm{\gamma}_0=\#\{j\colon \gamma_j\ne 0\}$ is the $\ell_0$ ``norm,''
which indexes the sparsity of $\gamma$, and with $s$ increasing slowly enough
relative to $n$ and $k$. Since $\norm{\gamma}_0$ is not a true norm or seminorm
(it is nonconvex), this parameter space is not covered by our setup.
Nonetheless, as we show in \Cref{connection_with_l1_estimator_sec}, if the
sparsity assumption is used to bound the $\ell_{1}$ loss of a preliminary lasso
estimator, arguments from \Cref{sec:finite-sample-result} lead to estimators and
\acp{CI} that are analogous to those proposed in the double lasso literature.

When should one use a double lasso \ac{CI}, and when should one use the approach
in the present paper? In principle, this depends on the \emph{a priori}
assumptions one is willing to make, and whether they are best captured by a
sparsity bound or a bound on convex penalty function, such as the $\ell_{1}$ or
$\ell_{2}$ norm. In many settings, it may be difficult to motivate the
assumption that a regression function has a sparse approximation, whereas upper
bounds on the magnitude of the coefficients may be more
plausible.\footnote{\label{fn:lq_nonconvex}One may also formalize bounds on the
  magnitude of the coefficients by imposing a nonconvex $\ell_{p}$ penalty with
  $p\in(0,1)$. However, such constraints are harder to interpret than the
  $\ell_{1}$ and $\ell_{2}$ bounds that are the focus of this paper.}

A key advantage of the \acp{CI} and estimators we propose is that they have
sharp finite-sample optimality properties and coverage guarantees in the fixed
design Gaussian model with known error variance. While this is an idealized
setting, the worst-case bias calculations do not depend on the error
distribution, and remain the same under non-Gaussian, heteroskedastic errors.
Our approach directly accounts for the potential finite-sample bias of the
estimator, rather than relying on ``asymptotic promises'' about rates at which
certain constants involved in bias terms converge to zero.

On the flip side, our \acp{CI} require an explicit choice of the penalty bound
$C$ in order to form a ``bias-aware'' \ac{CI}. In contrast, \acp{CI} based on
double lasso estimators do not require explicitly choosing the regularity (in
this case, the sparsity $s$), since they ignore bias. This is justified under
asymptotics in which $s$ increases more slowly than $\sqrt{n}/\log k$, which
ensures that the double lasso bias is asymptotically negligible. However, such
asymptotic promises are difficult to evaluate in a given finite-sample setting.
Indeed, as shown by \textcite{wuthrich_omitted_2021} and confirmed in our Monte
Carlo simulations in \Cref{sec:simulation-results} below, the double lasso CI
leads to undercoverage in finite samples even in relatively sparse settings. To
ensure good finite-sample coverage of the \ac{CI} in~\cref{eq:zz_CI}, one needs
to ensure that the actual finite-sample bias is negligible relative to the
standard deviation of the estimator. But since any bias bound depends on the
sparsity index $s$ (as in the bound in \cref{eq:lasso_gamma_l1_bound}), this
gets us back to having to explicitly specify $s$.

Thus, \acp{CI} that ignore bias such as conventional \acp{CI} based on double
lasso estimators do not avoid the problem of specifying $s$ or $C$: they merely
make such choices implicit in their asymptotic promises. These issues show up
formally in the asymptotic analysis of such \acp{CI}. In particular, double
lasso \acp{CI} require the ``ultra sparse'' asymptotic regime
$s=o(\sqrt{n}/\log k)$, and they undercover asymptotically in the ``moderately
sparse'' regime where $s$ increases more slowly than $n/\log k$ with
$s\gg \sqrt{n}/\log k$. Indeed, \Cref{thm:l1_lower_bound} above, as well as the
results of \textcite{cai_lasso_2017} and \textcite{javanmard_debiasing_2018}
show that it is impossible to avoid explicitly specifying $s$ if one allows for
the moderately sparse regime.

On the other end of the spectrum, in the ``low dimensional'' regime where
$k\ll n$, the double lasso \ac{CI} is asymptotically equivalent to the usual
\ac{CI} based on the long regression. Thus, the double lasso \ac{CI} cannot be
used when the goal is to use \emph{a priori} information on $\gamma$ to improve
upon the \ac{CI} based on the long regression \parencite[as in, for
example,][]{muralidharan_factorial_2020}. In contrast, our approach optimally
incorporates the bound $C$ regardless of the asymptotic regime.

We emphasize that the discussion above focuses on \acp{CI} rather than
estimation. Debiased lasso estimators may perform well in terms of mean squared
error across a range of regimes. In particular, they may perform well when the
$\ell_{1}$ norm of $\gamma$ is small (but $\gamma$ is not sparse), even if they
do not match the performance of the optimal linear estimator we propose that
explicitly targets this case.

\section{Simulation results}\label{sec:simulation-results}

We now illustrate the performance of our methods in simulations, focusing on the
$\ell_{1}$ penalty from \Cref{example:l1} and $R^{2}_{Y-w\beta}$ penalty from
\Cref{example:l2}.

We consider a design taken from \textcite{belloni_inference_2014}, with data
generated from a random regressor model that supplements
\cref{eq:linear_regression} with a propensity score regression
\begin{equation*}
  w_{i}=z_{i}'\pi+\sigma_{\tilde{w}}\tilde{w}_{i},
\end{equation*}
with $\tilde{w}_{i}$ and $\varepsilon_{i}$ independent standard Gaussian. The
design matrix $Z$ is random, with i.i.d.\ rows $z_{i}\sim\ND (0, \Sigma)$ and
$\Sigma_{ij}=2^{-\abs{i-j}}$. Similar to \textcite{wuthrich_omitted_2021}, we
tweak the \textcite{belloni_inference_2014} design by considering regression
coefficients that are of similar magnitude rather than decaying, so that we can
separately vary the degree of sparsity and the signal-to-noise ratio.
Specifically, we set
\begin{equation*}
  \gamma_{j}= \pi_{j} =
  \begin{cases}
    c_{1} & \text{if $j\leq k_{1}$,}\\
    c_{2} & \text{if $k_{1}<j\leq k_{1}+s$,} \\
    0 & \text{otherwise.}
\end{cases}
\end{equation*}
The \ac{DGP} is therefore indexed by 9 parameters: $n$, $k_{1}$, $k_{2}$, $s$,
$\beta$, $\sigma$, $\sigma_{\tilde{w}}$, $c_{1}$ and $c_{2}$. We fix $k_{1}=5$,
$\sigma=1$, and $\beta=0$, and consider $n \in \{500, 1000\}$,
$k_{2} \in \{0.2n, 0.8n\}$, $s\in \{10, 20, 100\}$, and
$\sigma_{\tilde{w}} \in \{0.5, 1\}$. We calibrate $c_{1}$ and $c_{2}$ by fixing
$R^{2}_{Z}$, the population $R^{2}$ from the regression of $Y$ on $Z$, and
fixing the ratio
$\rho = \norm{\gamma_{2}}_{1} / \norm{\gamma_{\textnormal{short}}}_{1}$, where
$\gamma_{\textnormal{short}}$ is the population short regression coefficient.
This allows us to directly control the signal-to-noise ratio, and the validity
of the rule-of-thumb calibration that sets $\rho=1$ and replaces
$\gamma_{\textnormal{short}}$ with an estimate. We consider
$R_{Z}^{2}\in\{0.1, 0.5\}$ and $\rho\in\{0.2, 1, 1.4, 2\}$. This gives a total
of 192 \acp{DGP}.

We consider the following six methods for constructing nominal 95\% \acp{CI} for
$\beta$:
\begin{description}
\item[AKK1] Implements \Cref{algorithm:baseline} with the penalty given by the
  $\ell_{1}$ norm of $\gamma_{2}$, the last $k_{2}$ regression coefficients.
  Following \Cref{remark:residuals}, the initial residual estimates use the
  lasso estimator (that only penalizes $\gamma_{2}$), with the penalty chosen
  via 10-fold cross-validation. Following \Cref{choice_of_C_remark}, we set
  $C=\rho\norm{\widehat{\gamma}_{\textnormal{short}}}_{1}$ with $\rho=1$, where
  $\widehat{\gamma}_{\textnormal{short}}$ are \ac{OLS} estimates from a short
  regression that only includes the first $k_{1}$ controls.
\item[BCH] The double lasso procedure of \textcite{belloni_inference_2014},
  implemented via the R package \texttt{hdm} without penalizing the $k_{1}$
  baseline controls.
\item[LR] The \ac{OLS} estimator based on the long regression with HC1 standard
  errors.
\item[Oracle AKK1] The oracle version of AKK1, which sets the penalty bound $C$
  to the actual value of $\norm{\gamma_{2}}_{1}$, and uses knowledge of the
  variance of the error term $\varepsilon_{i}$.
\item[Oracle AKK2] Implements \Cref{algorithm:baseline} with the
  $R^{2}_{Y-w\beta}$ penalty, using knowledge of the actual value of
  $R^{2}_{Y-w\beta}$ to set the bound $\Rmax^{2}$, and the variance of the error
  term $\varepsilon_{i}$.
\item[Oracle LM] The procedure proposed by \textcite{li_linear_2020},
  implemented using knowledge of the actual value of $R^{2}_{Y-w\beta}$ and the
  variance of the error term.
\end{description}

We focus on the oracle versions of the $R^{2}_{Y-w\beta}$ penalty and the
\textcite{li_linear_2020} procedure for two reasons. First, it ensures the
comparison between these two methods is not clouded by tuning parameter choices;
second, it allows us to illustrate the theoretical results
from~\Cref{sec:scope-impr-over} that unless the sample size is moderate,
efficiency gains relative to the long regression are modest even under the oracle
calibration.

\begin{table}[p]
  \renewcommand*{\arraystretch}{1.15}
  \centering
  \begin{adjustbox}{max width=\linewidth}
  \begin{threeparttable}
    \caption{Finite-sample coverage of nominal 95\% confidence intervals and
      their average length relative to the long regression
      ($n=500$).}\label{tab:sim_n=500}
    \begin{tabular}{@{}l c rrr rrr rrr rrr@{}}
      \toprule
      & & \multicolumn{6}{c}{$k_{2} = 100$} & \multicolumn{6}{c}{$k_{2} = 400$} \\
      \cmidrule(lr){3-8}\cmidrule(lr){9-14}
      & & \multicolumn{3}{c}{Feasible methods} & \multicolumn{3}{c}{Oracles}
        & \multicolumn{3}{c}{Feasible methods} & \multicolumn{3}{c}{Oracles} \\
      \cmidrule(lr){3-5}\cmidrule(lr){9-11}\cmidrule(lr){6-8}\cmidrule(lr){12-14}
      $\rho$ & \multicolumn{1}{c}{$R^{2}_{Y-w\beta}$}
        & \multicolumn{1}{c}{AKK1} & \multicolumn{1}{c}{BCH} & \multicolumn{1}{c}{LR}
        & \multicolumn{1}{c}{AKK1} & \multicolumn{1}{c}{AKK2} & \multicolumn{1}{c}{LM}
                             & \multicolumn{1}{c}{AKK1} & \multicolumn{1}{c}{BCH} & \multicolumn{1}{c}{LR}
                                   & \multicolumn{1}{c}{AKK1} & \multicolumn{1}{c}{AKK2} & \multicolumn{1}{c}{LM} \\
      \midrule

      \multicolumn{14}{@{}c@{}}{Panel A: Coverage Probability, minimum across DGPs}\\
      \addlinespace[0.35em]

      & \multicolumn{13}{c}{$s = 10$}\\
      0.2 & 0.011 & 94.5 & 92.5 & 94.3 & 94.6 & 94.4 & 94.3 & 96.6 & 90.8 & 94.0 & 94.6 & 98.1 & 96.9\\
      1.0 & 0.106 & 92.8 & 92.4 & 94.2 & 94.8 & 94.8 & 94.5 & 93.4 & 90.3 & 94.1 & 95.1 & 94.9 & 95.2\\
      1.4 & 0.139 & 92.0 & 92.2 & 92.7 & 93.6 & 93.1 & 93.2 & 90.4 & \textcolor{red}{87.9} & 93.5 & 94.7 & 94.6 & 95.3\\
      2.0 & 0.175 & \textcolor{red}{88.9} & 91.5 & 93.7 & 94.7 & 94.6 & 94.5 & \textcolor{red}{84.8} & \textcolor{red}{88.0} & 94.0 & 94.3 & 95.3 & 94.8\\[1em]

      & \multicolumn{13}{c}{$s = 20$}\\
      0.2 & 0.006 & 94.9 & 92.9 & 94.8 & 95.0 & 95.3 & 94.8 & 95.8 & 91.5 & 94.4 & 93.9 & 97.5 & 96.4\\
      1.0 & 0.078 & 93.0 & \textcolor{red}{85.2} & 94.2 & 94.2 & 94.4 & 94.7 & 92.7 & \textcolor{red}{74.4} & 93.7 & 93.9 & 94.4 & 94.7\\
      1.4 & 0.109 & 90.8 & \textcolor{red}{88.3} & 94.3 & 94.1 & 93.7 & 94.6 & 91.4 & \textcolor{red}{83.6} & 93.7 & 95.3 & 94.2 & 94.2\\
      2.0 & 0.145 & \textcolor{red}{89.0} & \textcolor{red}{85.0} & 94.1 & 94.0 & 94.1 & 95.2 & \textcolor{red}{82.3} & \textcolor{red}{81.1} & 93.4 & 93.9 & 94.2 & 94.4\\[1em]

      & \multicolumn{13}{c}{$s = 100$}\\
      0.2 & 0.001 & 94.8 & 94.5 & 94.4 & 95.0 & 94.9 & 95.2 & 95.1 & 93.2 & 93.3 & 94.1 & 95.3 & 95.1\\
      1.0 & 0.028 & 94.2 & \textcolor{red}{59.5} & 94.3 & 95.1 & 94.9 & 94.6 & 95.8 & \textcolor{red}{59.8} & 94.8 & 96.1 & 96.9 & 96.0\\
      1.4 & 0.045 & 93.8 & \textcolor{red}{37.5} & 94.2 & 94.7 & 94.6 & 94.4 & 94.7 & \textcolor{red}{27.9} & 94.6 & 97.0 & 96.4 & 96.0\\
      2.0 & 0.071 & 91.7 & \textcolor{red}{12.5} & 93.7 & 94.3 & 94.2 & 94.6 & 90.9 & \textcolor{red}{5.9} & 93.8 & 97.0 & 95.6 & 95.5\\

      \addlinespace[0.85em]
      \multicolumn{14}{@{}c@{}}{Panel B: Relative Length, mean across DGPs}\\
      \addlinespace[0.35em]

      & \multicolumn{13}{c}{$s = 10$}\\
      0.2 & 0.011 & 0.96 & 0.89 &  & 0.92 & 0.97 & 0.97 & 0.55 & 0.43 &  & 0.46 & 0.70 & 0.70\\
      1.0 & 0.106 & 0.94 & 0.90 &  & 0.95 & 1.00 & 0.97 & 0.51 & 0.44 &  & 0.52 & 0.93 & 0.90\\
      1.4 & 0.139 & 0.92 & 0.90 &  & 0.96 & 1.00 & 0.97 & 0.49 & 0.44 &  & 0.53 & 0.95 & 0.91\\
      2.0 & 0.175 & 0.91 & 0.90 &  & 0.96 & 1.00 & 0.96 & 0.47 & 0.44 &  & 0.54 & 0.96 & 0.90\\[1em]

      & \multicolumn{13}{c}{$s = 20$}\\
      0.2 & 0.006 & 0.96 & 0.89 &  & 0.92 & 0.96 & 0.96 & 0.56 & 0.43 &  & 0.46 & 0.64 & 0.64\\
      1.0 & 0.078 & 0.94 & 0.89 &  & 0.96 & 0.99 & 0.98 & 0.52 & 0.43 &  & 0.53 & 0.91 & 0.88\\
      1.4 & 0.109 & 0.93 & 0.90 &  & 0.97 & 1.00 & 0.97 & 0.50 & 0.43 &  & 0.55 & 0.94 & 0.90\\
      2.0 & 0.145 & 0.92 & 0.90 &  & 0.97 & 1.00 & 0.97 & 0.49 & 0.44 &  & 0.56 & 0.95 & 0.91\\[1em]

      & \multicolumn{13}{c}{$s = 100$}\\
      0.2 & 0.001 & 0.96 & 0.89 &  & 0.92 & 0.93 & 0.94 & 0.56 & 0.43 &  & 0.46 & 0.52 & 0.52\\
      1.0 & 0.028 & 0.95 & 0.88 &  & 0.97 & 0.98 & 0.99 & 0.55 & 0.43 &  & 0.55 & 0.80 & 0.79\\
      1.4 & 0.045 & 0.95 & 0.87 &  & 0.98 & 0.99 & 0.98 & 0.54 & 0.42 &  & 0.58 & 0.85 & 0.84\\
      2.0 & 0.071 & 0.94 & 0.87 &  & 0.98 & 0.99 & 0.98 & 0.52 & 0.42 &  & 0.61 & 0.90 & 0.87\\
      \bottomrule
    \end{tabular}
  \begin{tablenotes}
\footnotesize
\item \begin{minipage}{0.99\linewidth} \emph{Notes}: Panel A reports, for each
    method, the minimum coverage probability of nominal 95\% \acp{CI} over
    $\sigma_{\tilde{w}} \in \{0.5, 1\}$ and $R_{Z}^{2} \in \{0.1, 0.5\}$, the
    four \acp{DGP} that vary. Panel B reports mean relative length across the
    four \acp{DGP}. Within each \ac{DGP}, relative length is the method's
    average \ac{CI} length across simulation draws divided by the average length
    of the long regression \ac{CI}. 1,000 Monte Carlo draws. Cells with coverage
    below 90\% highlighted in red.
\end{minipage}
\end{tablenotes}
  \end{threeparttable}
  \end{adjustbox}
\end{table}

\begin{table}[p]
  \renewcommand*{\arraystretch}{1.15}
  \centering
  \begin{adjustbox}{max width=\linewidth}
  \begin{threeparttable}
    \caption{Finite-sample coverage of nominal 95\% confidence intervals and
      their average length relative to the long regression
      ($n=1000$).}\label{tab:sim_n=1000}
    \begin{tabular}{@{}l c rrr rrr rrr rrr@{}}
      \toprule
      & & \multicolumn{6}{c}{$k_{2} = 200$} & \multicolumn{6}{c}{$k_{2} = 800$} \\
      \cmidrule(lr){3-8}\cmidrule(lr){9-14}
      & & \multicolumn{3}{c}{Feasible methods} & \multicolumn{3}{c}{Oracles}
        & \multicolumn{3}{c}{Feasible methods} & \multicolumn{3}{c}{Oracles} \\
      \cmidrule(lr){3-5}\cmidrule(lr){9-11}\cmidrule(lr){6-8}\cmidrule(lr){12-14}
      $\rho$ & \multicolumn{1}{c}{$R^2_{Y-w\beta}$}
      & \multicolumn{1}{c}{AKK1} & \multicolumn{1}{c}{BCH} & \multicolumn{1}{c}{LR}
      & \multicolumn{1}{c}{AKK1} & \multicolumn{1}{c}{AKK2} & \multicolumn{1}{c}{LM}
      & \multicolumn{1}{c}{AKK1} & \multicolumn{1}{c}{BCH} & \multicolumn{1}{c}{LR}
      & \multicolumn{1}{c}{AKK1} & \multicolumn{1}{c}{AKK2} & \multicolumn{1}{c}{LM} \\
      \midrule

      \multicolumn{14}{@{}c@{}}{Panel A: Coverage Probability, minimum across DGPs}\\
      \addlinespace[0.35em]

      & \multicolumn{13}{c}{$s = 10$}\\
      0.2 & 0.011 & 94.7 & 92.8 & 94.8 & 95.1 & 95.0 & 94.6 & 95.9 & 90.2 & 93.4 & 95.3 & 95.9 & 95.1\\
      1.0 & 0.106 & 94.3 & 94.1 & 94.5 & 94.7 & 94.2 & 95.2 & 93.4 & 91.4 & 94.5 & 94.2 & 95.0 & 94.6\\
      1.4 & 0.139 & 92.1 & 93.7 & 94.8 & 93.5 & 95.0 & 95.0 & 91.2 & 93.0 & 93.8 & 94.4 & 93.6 & 94.2\\
      2.0 & 0.175 & 90.5 & 94.0 & 93.9 & 93.3 & 93.7 & 94.1 & \textcolor{red}{84.7} & 93.8 & 94.3 & 94.6 & 94.7 & 95.0\\[1em]

      & \multicolumn{13}{c}{$s = 20$}\\
      0.2 & 0.006 & 94.7 & 91.3 & 94.5 & 95.2 & 95.3 & 94.9 & 94.7 & 91.4 & 93.7 & 94.2 & 97.6 & 96.2\\
      1.0 & 0.078 & 93.3 & \textcolor{red}{89.6} & 94.4 & 94.4 & 94.3 & 94.8 & 92.9 & \textcolor{red}{86.7} & 93.7 & 93.7 & 94.4 & 94.0\\
      1.4 & 0.109 & 92.6 & \textcolor{red}{85.0} & 94.5 & 94.4 & 94.7 & 94.4 & 90.2 & \textcolor{red}{83.8} & 94.0 & 94.4 & 94.1 & 94.1\\
      2.0 & 0.145 & 90.3 & \textcolor{red}{83.7} & 93.8 & 93.9 & 94.1 & 94.9 & \textcolor{red}{84.0} & \textcolor{red}{75.7} & 94.7 & 94.4 & 94.6 & 95.3\\[1em]

      & \multicolumn{13}{c}{$s = 100$}\\
      0.2 & 0.001 & 93.5 & 92.8 & 92.7 & 93.8 & 93.5 & 92.8 & 95.5 & 94.5 & 95.0 & 95.5 & 96.3 & 96.0\\
      1.0 & 0.028 & 93.6 & \textcolor{red}{38.1} & 93.5 & 93.9 & 93.9 & 93.7 & 95.4 & \textcolor{red}{31.7} & 94.0 & 95.7 & 95.3 & 94.9\\
      1.4 & 0.045 & 94.6 & \textcolor{red}{16.5} & 94.8 & 94.8 & 95.3 & 94.7 & 94.8 & \textcolor{red}{7.7} & 94.2 & 96.4 & 95.6 & 95.4\\
      2.0 & 0.071 & 93.5 & \textcolor{red}{3.8} & 93.7 & 94.2 & 94.4 & 94.2 & \textcolor{red}{88.3} & \textcolor{red}{1.5} & 94.4 & 96.6 & 95.0 & 95.3\\

      \addlinespace[0.85em]
      \multicolumn{14}{@{}c@{}}{Panel B: Relative Length, mean across DGPs}\\
      \addlinespace[0.35em]

      & \multicolumn{13}{c}{$s = 10$}\\
      0.2 & 0.011 & 0.96 & 0.89 &  & 0.92 & 0.98 & 0.99 & 0.56 & 0.44 &  & 0.46 & 0.78 & 0.78\\
      1.0 & 0.106 & 0.94 & 0.90 &  & 0.95 & 1.00 & 0.98 & 0.52 & 0.44 &  & 0.52 & 0.96 & 0.96\\
      1.4 & 0.139 & 0.93 & 0.90 &  & 0.96 & 1.00 & 0.98 & 0.50 & 0.44 &  & 0.54 & 0.97 & 0.96\\
      2.0 & 0.175 & 0.92 & 0.90 &  & 0.96 & 1.00 & 0.97 & 0.49 & 0.44 &  & 0.55 & 0.98 & 0.95\\[1em]

      & \multicolumn{13}{c}{$s = 20$}\\
      0.2 & 0.006 & 0.96 & 0.89 &  & 0.92 & 0.97 & 0.98 & 0.57 & 0.44 &  & 0.46 & 0.72 & 0.72\\
      1.0 & 0.078 & 0.95 & 0.90 &  & 0.96 & 1.00 & 0.99 & 0.53 & 0.44 &  & 0.54 & 0.95 & 0.95\\
      1.4 & 0.109 & 0.94 & 0.90 &  & 0.97 & 1.00 & 0.98 & 0.52 & 0.44 &  & 0.56 & 0.97 & 0.95\\
      2.0 & 0.145 & 0.93 & 0.90 &  & 0.97 & 1.00 & 0.97 & 0.50 & 0.44 &  & 0.57 & 0.97 & 0.96\\[1em]

      & \multicolumn{13}{c}{$s = 100$}\\
      0.2 & 0.001 & 0.96 & 0.89 &  & 0.92 & 0.94 & 0.95 & 0.57 & 0.44 &  & 0.47 & 0.57 & 0.58\\
      1.0 & 0.028 & 0.96 & 0.88 &  & 0.97 & 0.99 & 0.99 & 0.56 & 0.43 &  & 0.56 & 0.87 & 0.87\\
      1.4 & 0.045 & 0.95 & 0.88 &  & 0.98 & 0.99 & 0.99 & 0.55 & 0.43 &  & 0.59 & 0.92 & 0.92\\
      2.0 & 0.071 & 0.94 & 0.88 &  & 0.98 & 1.00 & 0.99 & 0.53 & 0.43 &  & 0.62 & 0.95 & 0.94\\

      \bottomrule
    \end{tabular}
   \begin{tablenotes}
  \item\footnotesize\emph{Notes}: See \Cref{tab:sim_n=500}.
  \end{tablenotes}
  \end{threeparttable}
\end{adjustbox}
\end{table}

\Cref{tab:sim_n=500} reports the simulation results for \acp{DGP} with $n=500$,
and \Cref{tab:sim_n=1000} reports them for \acp{DGP} with $n=1000$. In line with
the theory, all oracle \acp{CI} have coverage close to nominal across the board,
as does the long regression. For AKK1, the coverage is also close to nominal so
long as $\rho\leq 1$, so that the population restriction underlying the rule of
thumb is satisfied. Furthermore, in these designs, the method is robust to mild
violations of the population constraint, with coverage remaining above 90\% when
$\rho=1.4$. Only when $\rho=2$ do we start to see mild undercoverage in select
designs, as the red cells, which highlight coverage below 90\%, indicate. By
contrast, the BCH method displays mild undercoverage even in sparse designs with
$s=10$, and coverage dips below 75\% once $s=20$. Coverage drops off to almost
zero in the non-sparse designs with $s=100$. These results illustrate the
concern discussed in \Cref{sec:nonconvex_efficiency} that asymptotic sparsity
requirements may be difficult to evaluate in finite samples.

Panel~B shows that AKK1 regularization yields substantial length reductions over
the long regression, especially when $k_{2}$ is large relative to $n$: with
$k_{2} = .8n$, the AKK1 \acp{CI} are roughly half as long as the long regression
\ac{CI}. With $k_{2}=.2n$, the gains are more modest, with the AKK1 \acp{CI} at
91--96\% of the long regression length. By contrast, the oracle AKK2 and LM
procedures that bound $R^{2}_{Y-w\beta}$ display much more modest efficiency
gains, even though they use the true value of $R^{2}_{Y-w\beta}$. To help put
the gains in context, the tables report the value of $R^{2}_{Y-w\beta}$ in the
second column. In line with the theory in \Cref{sec:scope-impr-over}, the
percentage gains in length are typically in the single digits unless
$R^{2}_{Y-w\beta}$ is less than about 3\%; both methods perform similarly.

\section{Empirical applications}\label{sec:empirical-app}

We revisit two empirical applications to illustrate the performance of our
methods using real data. The first application conducts inference on the
\ac{MPE} out of unearned income using data from \textcite{imbens2001lottery},
and the second reanalyzes the effectiveness of \ac{RHC}, using data from
\textcite{connors_rhc_1996}.

\subsection{Effect of unearned income on labor earnings}\label{sec:marg-prop-cons}

The \ac{MPE} out of unearned income is a key structural parameter in labor and
public economics. While unearned income is typically endogenous,
\textcite{imbens2001lottery} argue that in their sample of $n=496$ winners of
major and minor prizes in the Massachusetts lottery in 1984--88, the magnitude
of winnings is as good as random conditional on individual characteristics that
proxy for the frequency of lottery ticket purchases. The lottery winnings are
paid out over 20 years, so that in a regression of the average social security
earnings in the 6 years after the lottery, $Y_{i}$, on yearly lottery payments,
$w_{i}$, and individual controls, the coefficient on $w_{i}$ may be interpreted
as the \ac{MPE}, under the assumption that the \ac{MPE} is constant.

We relax the assumption that the \ac{MPE} is constant using a specification
taken from \textcite{li_linear_2020}: we augment the baseline set of $k_{1}=7$
individual controls $z_{i1}$ consisting of the intercept, two continuous
controls (years of education and age), and 4 binary controls (indicators for
male, college, age over 55, and age over 65) with $k_{2}=25$ additional controls
$z_{i2}$ that are constructed by taking demeaned cross-products of the four
baseline binary controls and their interactions with $w_{i}$ and dropping
collinear terms. All non-binary covariates contained in $z_{i1}$ and $z_{i2}$
are standardized. Following the discussion in \Cref{remark:hte}, the coefficient
on $w_{i}$ in this specification can be interpreted as the average \ac{MPE},
allowing for heterogeneity in the \ac{MPE} with respect to the binary controls.
By contrast, the short regression estimand in a regression that only includes
$z_{i1}$ does not yield the average \ac{MPE} in the presence of such
heterogeneity.

\begin{table}[t]
  \centering
  \begin{threeparttable}
  \caption{Performance of penalized regression methods in empirical applications.}\label{tab:empirical_results}
  \begin{tabular}{@{}lc r@{,~}l c@{}}
    \toprule
    Method & Point estimate & \multicolumn{2}{c}{95\% Confidence interval} & \ac{CI} length \\
    \midrule
    \addlinespace[0.8ex]
    \multicolumn{5}{c}{{Panel A\@: \textcite{imbens2001lottery}}} \\
    \addlinespace[0.4ex]
    Short regression & $-0.0522$ & $[-0.0736$ & $-0.0309]$ & $0.0427$ \\
    Long regression & $-0.0493$ & $[-0.1158$ & $0.0171]$ & $0.1329$ \\
    AKK1 ($\rho=1$) & $-0.0600$ & $[-0.0982$ & $-0.0218]$ & $0.0764$ \\
    AKK1 ($\rho=2$) & $-0.0567$ & $[-0.1065$ & $-0.0068]$ & $0.0997$ \\
    BCH & $-0.0557$ & $[-0.0804$ & $-0.0310]$ & $0.0494$ \\
    AKK2 ($\Rmax^2=.05$) & $-0.0497$ & $[-0.1099$ & $0.0104]$ & $0.1203$ \\
    \addlinespace[0.3ex]
    Li--M\"uller ($\Rmax^2=.05$) & $-0.0504$ & $[-0.1168$ & $0.0161]$ & $0.1329$ \\
    \addlinespace[1.0ex]
    \multicolumn{5}{c}{{Panel B\@: \textcite{connors_rhc_1996}}} \\
    \addlinespace[0.4ex]
    Short regression & $0.0635$ & $[0.0385$ & $0.0885]$ & $0.0500$ \\
    Long regression & $0.0482$ & $[0.0026$ & $0.0938]$ & $0.0912$ \\
    AKK1 ($\rho=1$) & $0.0410$ & $[0.0053$ & $0.0768]$ & $0.0715$ \\
    AKK1 ($\rho=2$) & $0.0444$ & $[0.0060$ & $0.0828]$ & $0.0768$ \\
    BCH & $0.0562$ & $[0.0297$ & $0.0827]$ & $0.0530$ \\
    AKK2 ($\Rmax^2=.05$) & $0.0484$ & $[0.0030$ & $0.0938]$ & $0.0908$ \\
    \addlinespace[0.3ex]
    Li--M\"uller ($\Rmax^2=.05$) & $0.0482$ & $[0.0025$ & $0.0940]$ & $0.0915$ \\
    \bottomrule
  \end{tabular}
  \begin{tablenotes}
  \item\footnotesize\emph{Notes:} AKK1 uses an $\ell_{1}$ penalty bound with the
    rule-of-thumb calibration $\rho=1$; $\rho=2$ is reported as a robustness
    check. AKK2 and the \textcite{li_linear_2020} method impose
    $R^2_{Y-w\beta}\leq 0.05$. For variance calculations, all methods use
    residual estimates based on a lasso estimator that penalizes only
    $\gamma_{2}$, with penalty chosen by 10-fold cross-validation.
  \end{tablenotes}
\end{threeparttable}
\end{table}

We consider the same five methods as in the simulations: AKK1, BCH, long
regression, AKK2, and the \textcite{li_linear_2020} method. In addition, we
report the short regression that omits the interactions with the baseline
controls.\footnote{To make the methods more comparable and not conflate the
  comparison with differences in standard error construction, for variance
  calculations, all methods use residual estimates based on a lasso estimator
  that penalizes only $\gamma_{2}$, with penalty chosen by 10-fold
  cross-validation.} The results are reported in Panel~A of
\Cref{tab:empirical_results}.

The \ac{MPE} estimate in the long regression is close to the short regression
estimate $-0.052$ that only includes the baseline controls $z_{i1}$ and
corresponds to the specification in Table 4, column II, row 1 in
\textcite{imbens2001lottery}. However, the long regression estimate is very
noisy: the 95\% confidence interval includes positive values for the average
\ac{MPE}, which economic theory rules out, and it is over three times as long as
the short regression \ac{CI}. To increase precision, \textcite{li_linear_2020}
consider restricting the sum of the squared effects,
$\sum_{i=1}^{n}(\tilde{z}_{i2}'\gamma_{2})^{2}$, where $\tilde{z}_{i2}$ are the
residuals from projecting $z_{i2}$ onto $z_{i1}$. To make this restriction more
interpretable, we implement this penalty by restricting
$R^{2}_{Y-w\beta}\leq \Rmax^{2}$ as in \Cref{example:l2}, and calibrate the
bound to $\Rmax^{2}=0.05$. Given that the $R^{2}$ in the short regression equals
$0.244$, this choice of $\Rmax^{2}$ limits the explanatory power of the
interaction terms to be about a fifth of the main effects, formalizing the
empirical regularity that main effects are more important. Ex post, this
regularity turns out to hold in this dataset as well, given the long regression
$R^{2}$ is only $0.285$. This choice of $\Rmax^{2}$ yields the
\textcite{li_linear_2020} \ac{CI} $(-0.117, 0.016)$, with essentially the same
length as the long regression \ac{CI}.\footnote{\textcite{li_linear_2020} show
  that their method is close to optimal in terms of weighted average length
  under a homoskedastic benchmark. This may no longer be the case under
  heteroskedasticity. Their \ac{CI} is variable-length, and may be longer than
  the long regression \ac{CI} in some samples even in the homoskedastic case. In
  contrast, our construction guarantees length improvements over the long
  regression in all samples under homoskedasticity.} By contrast, the AKK2
\ac{CI}, $(-0.110, 0.010)$, is 9.5\% shorter than the long regression \ac{CI},
although it still includes zero. This moderate improvement over the long
regression is consistent with the theory in \Cref{sec:scope-impr-over}, which
under a homoskedastic benchmark, predicts about a 5.5\% improvement.

\begin{figure}[tp]
  \centering
  \input{mpe_plot.tex}
  \caption{95\% \acp{CI} for the \acs{MPE} out of unearned income under the
    $\ell_{1}$ penalty.}\label{fig:emp}
  \floatfoot{\emph{Notes:} Orange solid vertical lines at the left and right
    endpoints of the $x$-axis mark \acp{CI} based on the short and long
    regressions, respectively. The black solid vertical line at the left endpoint
    of the $x$-axis marks the BCH \ac{CI}. The blue shaded area depicts the bias-aware
    \acp{CI}, while the dotted black line shows the point estimates as a
    function of $\rho=C/\norm{\widehat{\gamma}_{\textnormal{short}}}_{1}$, the
    relative total-effect bound. The default $\rho=1$ and breakdown value
    $\rho^{*}$ are marked on the $x$-axis (log scale).}
\end{figure}

Limiting the total effect of the additional controls by imposing an $\ell_{1}$
penalty bound as in \Cref{example:l1} yields more substantive precision gains.
At the baseline calibration $\rho=1$, the point estimate is similar to the long
regression estimate, but the \ac{CI} is only about 60\% as long, and excludes
zero. At $\rho=2$, the \ac{CI} widens a bit, yielding $(-0.107, -0.007)$.
\Cref{fig:emp} examines the sensitivity of the method to the choice of $\rho$
more thoroughly, showing that moderate changes in $\rho$ have little impact. The
breakdown value for statistical significance is $\rho^{*}=2.59$: the total
effect of the additional controls would need to be more than 2.6 times that of
the baseline controls in order for the \ac{CI} to include
zero.\footnote{\label{fn:lind1}\Cref{remark:empirical-lindeberg} in
  \Cref{sec:additional-results} examines the sensitivity of the results to
  imposing a bound on the maximal Lindeberg weight, defined in
  \cref{eq:lind_weight}. The AKK1 and AKK2 \acp{CI} widen somewhat, as does the
  short regression \ac{CI}, but the qualitative conclusions remain unchanged.
  However, the long regression \ac{CI} becomes uninformative, and so does the
  Li--M\"uller procedure.} The double lasso \ac{CI} of
\textcite{belloni_inference_2014}, depicted by a black vertical line in the
figure, is very short. However, as the simulation results in
\Cref{sec:simulation-results} demonstrate, its coverage guarantees may be
unreliable when the underlying sparsity conditions are not met, and in this
application, it is not clear why the coefficients on the additional controls
should admit a sparse approximation.

The point estimates are remarkably stable across specifications. Thus, while the
penalized regression results indicate that there is more uncertainty about the
precise magnitude of the average \ac{MPE} than the short regression suggests, we
find little evidence that the short regression is biased.

\subsection{Right heart catheterization}\label{sec:right-heart-cath}

We revisit the observational study of \textcite{connors_rhc_1996}, who found
that \ac{RHC} led to an increase in 30-day mortality among critically ill
patients. We follow the empirical design in \textcite{hirano_rhc_2001}. After
removing observations with missing values, the final analysis sample consists of
$n=5{,}735$ patients. The treatment is an indicator for \ac{RHC}, and the
outcome is a binary indicator for 30-day mortality. The short regression
includes the treatment indicator and the original 72 baseline covariates used in
\textcite{hirano_rhc_2001}. Here we examine the robustness of their results to
functional form misspecification by including second- and third-order
interactions of the baseline controls.\footnote{Specifically, the additional
  covariates consist of squares and cubes of all continuous baseline covariates,
  all pairwise interactions (continuous$\times$continuous,
  continuous$\times$dummy, and dummy$\times$dummy), and higher-order cross-terms
  (products of polynomial transformations up to third order), after removing
  constant columns.} The long regression thus includes $k_{1}=72$ baseline
controls and $k_{2}=4{,}078$ higher-order terms.\footnote{The long regression
  drops 200 covariates due to collinearity. Since our procedure does not require
  the design matrix to be nonsingular, we keep all covariates when applying our
  methods.}

We consider the same six methods as in the previous application. Panel~B of
\Cref{tab:empirical_results} reports the results.\footnote{\label{fn:lind2}The
  maximal Lindeberg weights for all methods are small; see
  \Cref{remark:empirical-lindeberg}.} All methods yield positive point estimates
and \acp{CI}, consistent with the original finding that RHC is associated with
higher mortality. However, there is some disagreement between the methods about
the magnitudes. The long regression estimate that \ac{RHC} leads to a 4.8
percentage point increase in 30-day mortality is close to the lower endpoint of
the short regression \ac{CI}. At the same time, the long regression \ac{CI} is
over 80\% wider than the short regression \ac{CI}. The short regression $R^{2}$
is 0.223. We again calibrate the $\Rmax^{2}$ bound for AKK2 and the
\textcite{li_linear_2020} method to $\Rmax^{2}=0.05$ to capture the notion that
the interactions have about a fifth of the explanatory power of the main
effects. Consistent with the theoretical calculations under the homoskedastic
benchmark in \Cref{sec:scope-impr-over}, however, due to the large sample size,
both methods yield \acp{CI} that are very close to the long regression. The BCH
\ac{CI} is very short, and close to the short regression. However, it is not
clear why the long regression should admit a sparse approximation. The
$\ell_{1}$ penalty calibrated with $\rho=1$ yields a 22\% precision gain
relative to the long regression, with the \ac{CI} given by $(0.005, 0.077)$; the
point estimate, $0.041$, implies a treatment effect that is about two percentage
points lower than the short regression estimate. At $\rho=2$, the \ac{CI}
changes little. Overall, the analysis indicates that the short regression result
that \ac{RHC} increases mortality is robust to functional form assumptions,
though the magnitude of the effect is likely about two percentage points lower
than the short regression estimate suggests.

\begin{appendices}
\crefalias{section}{appsec}
\crefalias{subsection}{appsec}
\crefalias{subsubsection}{appsec}

\section{Extensions and additional results}

\subsection{Partially linear model}\label{sec:partly-linear-model}

To flexibly control for a low-dimensional set of covariates $\tilde{z}_{i}$, one
may specify a semiparametric model
\begin{equation*}
  Y_{i} = w_{i}\beta + h(\tilde{z}_{i}) + \varepsilon_{i},
  \quad \widetilde{\Pen}(h)\le \widetilde C,
\end{equation*}
where the penalty $\widetilde{\Pen}(h)$ is a seminorm on functions $h(\cdot)$
that penalizes the ``roughness'' of $h$, such as the Hölder or Sobolev seminorm.
Minimax linear estimation in this model for particular choices of
$\widetilde{\Pen}(h)$ has been considered in \textcite{heckman_minimax_1988}.
This setting is covered by our setup if we define $Z=I_{n}$,
$\gamma_{i}=h(\tilde{z}_{i})$, and
$\Pen(\gamma)=\min_{h\colon h(\tilde z_i)=\gamma_i, \; i=1,\dotsc n}
\widetilde{\Pen}(h)$ (assuming the minimum is attained).
  Note, in particular, that the assumption that $\widetilde{\Pen}(\cdot)$ is a
  seminorm on the space of functions $h$ guarantees that $\Pen(\cdot)$ is a
  seminorm on $\mathbb{R}^n$.
  \Cref{thm:optimization_equivalence} then implies that the optimal estimator
  takes the form
  \begin{equation*}
    \hat{\beta}_\lambda = \frac{\sum_{i=1}^{n}
      (w_i-g_\lambda(\tilde z_i))Y_i}{\sum_{i=1}^n(w_i-g_\lambda(\tilde z_i))w_i},
  \end{equation*}
  where $g_\lambda(\cdot)$ solves
  \begin{equation*}
    \min_{g} \sum_{i=1}^n(w_i-g(\tilde z_i))^2
    \quad\text{s.t.}\quad
    \widetilde{\Pen}(g) \le t
  \end{equation*}
  for some value of $t$; $g_\lambda(\cdot)$ is thus analogous to the regularized
  regression estimate $\pi_\lambda$ in~\cref{eq:pi_optimization}.

  When $\widetilde{\Pen}$ is the Sobolev seminorm, this yields a spline estimate
  $g_{\lambda}$ \parencite[see, for example][]{wahba90}. Interestingly, the
  estimator proposed in the seminal work by
  \textcite{robinson_root-n-consistent_1988} takes a similar form to the estimator
  $\hat\beta_\lambda$, involving residuals from a nonparametric regression of
  $w$ on $\tilde z_i$. While the analysis in
  \textcite{robinson_root-n-consistent_1988} is asymptotic, our results imply that
  a version of this estimator has sharp finite-sample optimality properties.

\subsection{Heterogeneous treatment effects}\label{sec:heter-treatm-effects}

A more general approach to handling treatment effect heterogeneity relative to
the one discussed in \Cref{remark:hte} is to consider the model
\begin{equation}\label{het_te_regression_eq}
  Y_i = w_i\tau(z_i) + z_i'\gamma + \varepsilon_i,
\end{equation}
where $\tau(z)$ represents the (possibly nonlinear in $z$) conditional \ac{ATE}
for units $i$ with $z_{i}=z$. The parameter of interest is given by
$\beta=\int \tau(z)\, d\mu(w, z)$, where $\mu$ is a signed measure defined on
some set that includes the empirical support $\{(w_{i}, z_{i})\}_{i=1}^{n}$.
Allowing $\mu$ to be signed allows for inference on nonconvex averages of
$\tau(z)$. We allow $\mu$ to place nonzero mass outside the empirical support
$\{(w_{i}, z_{i})\}_{i=1}^{n}$, thereby allowing for extrapolation. The general
theory of estimation and inference on linear functionals developed in
\textcite{donoho94} and \textcite{ArKo18optimal} underlying
\Cref{thm:optimization_equivalence,optimality_corollary} can be applied to
inference on this parameter under any convex and symmetric restriction on the
function $\tau(\cdot)$ and the parameter $\gamma$---only the worst-case bias
calculation in \cref{eq:maxbias_def} changes.

Assuming $\tau(z)$ is linear and setting $\mu$ to be a weighted empirical
measure with weights $c_{i}$ that sum to one yields the specification discussed
in \Cref{remark:hte}. Another case of interest arises when we compute the same
linear estimator and bias-aware \acp{CI} as in the homogeneous \ac{TE} model in
\cref{eq:linear_regression}. Allowing for \ac{TE} heterogeneity then requires
changing the interpretation of the estimand as targeting a particular weighted
average of \acp{TE} given in the next theorem. To set the stage for the theorem,
let us denote the worst-case bias of a linear estimator $\hat{\beta}=a'Y$
relative to an estimand $\int \tau(z)d\mu(w, z)$ when the heterogeneity is
completely unrestricted by
$\widetilde{\maxbias}_{\Gamma}(\hat{\beta};\mu)=\sup_{\tau(\cdot), \gamma}
\left[ \sum_{i=1}^{n} a_{i}w_i \tau(z_i)+a'Z\gamma - \int \tau(z)\, d\mu(w, z)
\right]$, where the supremum is over $\gamma\in\Gamma$ and all measurable
functions $\tau(\cdot)$. For an $n$-vector $a$, let $\mu_{a, w}^{*}(w, z)$
denote a weighted empirical measure with (possibly negative) weights
$a_{i}w_{i}$, so that the parameter of interest becomes
$\beta_{a, w}=\sum_{i=1}^{n}a_{i}w_{i}\tau(z_{i})$.

For a signed measure $\mu$ on $(w, z)$, let $\mu_{z}$ denote its marginal on
$z$, $\mu_{z}(A)=\mu(\{(w, z)\colon z\in A\})$. Since $\tau(\cdot)$ depends only
on $z$, the estimand depends on $\mu$ only through $\mu_{z}$.

\begin{theorem}\label{theorem:weighted_te}
  Let $\mu$ be a signed measure with $\int d\mu(w, z)=1$, and let
  $\hat{\beta}=a'Y$ be an estimator with $a'w=1$. If
  $\mu_{z}=(\mu^{*}_{a, w})_{z}$, so that
  $\int\tau(z)\, d\mu(w, z)=\beta_{a, w}$, then
  $\widetilde{\overline{\bias}}_{\Gamma}(\hat\beta;\mu)
  =\maxbias_{\Gamma}(\hat\beta)$, with $\maxbias_{\Gamma}(\hat\beta)$ given
  in~\cref{eq:maxbias_def}. If $\mu_{z}\neq(\mu^{*}_{a, w})_{z}$, then
  $\widetilde{\overline{\bias}}_{\Gamma}(\hat\beta;\mu)=\infty$. Furthermore,
  the estimator $\hat{\beta}_{\lambda}$ given in \cref{eq:optimal_estimator}
  solves
  \begin{equation}\label{moving_goalposts_eq}
    \min_{\hat\beta=a'Y, a\in\mathbb{R}^{n}} \var(\hat\beta)
    \quad\text{s.t.}\quad
    \min_\mu \widetilde{\overline{\bias}}_{\Gamma}(\hat\beta;\mu)
    \le C\overline B_\lambda,
  \end{equation}
  where the minimization over $\mu$ is over all signed measures such that
  $\int d\mu(w, z)=1$.
\end{theorem}
\begin{proof}
  Since the bias depends on $\mu$ only through $\mu_{z}$, it suffices to
  consider $\mu=\mu^{*}_{a, w}$. Note that
  $ \widetilde{\overline{\bias}}_{\Gamma}(\hat\beta;\mu^*_{a, w})
  =\sup_{\gamma\in\Gamma} a'Z\gamma$, which proves the first part of the
  theorem.

   To prove the second part of the theorem, let
   $\nu=(\mu^{*}_{a, w})_{z}-\mu_{z}$, so that the bias equals
   $\int\tau\, d\nu+a'Z\gamma$. If $\nu\neq 0$, there is a measurable set $A$
   with $\nu(A)\neq 0$. Setting $\gamma=0\in\Gamma$, and $\tau(z)=c$ for
   $z\in A$ and $\tau(z)=0$ otherwise, gives bias $c\,\nu(A)$, which can be
   made arbitrarily large by taking $c$ to have the same sign as $\nu(A)$ and
   $\abs{c}$ large. Thus
   $\widetilde{\overline{\bias}}_{\Gamma}(\hat\beta;\mu)=\infty$ if
   $\mu_{z}\ne(\mu^*_{a, w})_{z}$.

   To prove the final assertion, recall that the weights $a_\lambda$ minimize
   the variance of the linear estimator $\hat\beta$ subject to the bound
   $C\overline B_\lambda$ on worst-case bias $\maxbias_{\Gamma}(\hat\beta)$ for
   $\beta$ in the constant treatment effects model in
   \cref{eq:linear_regression}. It follows from the first two assertions that
   $\maxbias_{\Gamma}(\hat\beta)=\min_{\mu}
   \widetilde{\overline{\bias}}_{\Gamma}(\hat\beta;\mu)
   =\widetilde{\overline{\bias}}_{\Gamma}(\hat\beta;\mu^{*}_{a, w})$, where the
   minimization is over all signed measures $\mu$ that integrate to one. Thus,
   under the heterogeneous \ac{TE} model~\eqref{het_te_regression_eq}, the same
   weights $a_\lambda$ solve~\cref{moving_goalposts_eq}.
\end{proof}
The theorem gives three results for inference when $\tau(\cdot)$ is
unrestricted. First, given a linear estimator $a'Y$, the only estimand for which
the bias is finite is $\beta_{a, w}$. Second, for this estimand, the bias is the
same as in the homogeneous \ac{TE} model with $\tau(z)=\beta$. Thus, assuming
homogeneous \acp{TE} leads to valid inference for ${\beta}_{a, w}$ when \acp{TE}
are in fact heterogeneous. In particular, the bias-aware \ac{CI} based on the
estimator $\hat{\beta}_{\lambda}$ provides inference on the weighted average
${\beta}_{a_{\lambda}, w}= \sum_{i=1}^{n}\tilde{w}_{\lambda,
  i}w_{i}\tau(z_{i})/\sum_{i=1}^{n}\tilde{w}_{\lambda, i}w_{i}$. Observe that,
when the treatment $w_{i}$ is binary, the weights are nonnegative if and only if
the residual in the propensity score regression $\tilde{w}_{\lambda, i}$ is
nonnegative for all treated units---this can be easily verified in a given
application. Equivalently, the weights are nonnegative if and only if
$z_{i}'\pi_{\lambda}\leq 1$ for all treated units: the fitted values in the
propensity score regression must respect the population constraint that
treatment probabilities cannot exceed one. This finding is a finite-sample
analog of the identification result in \textcite{ghk24}, who show that the
estimand in the partially linear model has an analogous weighted \ac{ATE}
interpretation under heterogeneous \acp{TE}. An analogous identification result
in a random design setting dates to at least \textcite{angrist_estimating_1998},
who gave a weighted \ac{ATE} interpretation to the \ac{OLS} estimand.

Third, the estimator $\hat{\beta}_{\lambda}$ remains optimal in the
heterogeneous \ac{TE} model in that it solves a bias-variance tradeoff in a
problem where we can pick the estimand to make the estimation problem as easy as
possible. The problem~\eqref{moving_goalposts_eq} is a finite-sample version of
the ``moving the goalposts'' problem considered by \textcite[][Section
5.4]{crump_moving_2006}.\footnote{The optimization
  problem~\eqref{moving_goalposts_eq} was also used by \textcite{ImWa19} in the
  context of a regression discontinuity design with multiple cutoffs, although
  they did not explicitly note the optimality properties of the resulting
  estimator.} \textcite{crump_moving_2006} derived the measure $\mu$ that minimizes
the asymptotic variance of a particular class of inverse propensity score
weighted estimators of weighted \acp{ATE} in a random design setup.
\textcite{ghk24} show this measure in fact minimizes the variance among all regular
estimators; the measure coincides with the weighting of the treatment effects
given in \textcite{angrist_estimating_1998} under homoskedasticity, giving an
optimality property to the \ac{OLS} estimator.

\subsection{Connection between double lasso and optimal estimator
under \texorpdfstring{$\ell_1$}{ell1} constraints}\label{connection_with_l1_estimator_sec}

When $\Pen(\gamma)=\norm{\gamma}_1$ (\cref{example:l1}), the solution
$\pi_\lambda$ to \cref{eq:pi_optimization} is the lasso estimate in the
propensity score regression of $w$ on $Z$, and our
estimator~\eqref{eq:optimal_estimator} uses residuals from this lasso
regression. This is related to ``double lasso'' estimators used to form \acp{CI}
for $\beta$ under sparsity constraints on $\gamma$ \parencite[see, among
others,][]{belloni_inference_2014,javanmard_confidence_2014,van_de_geer_asymptotically_2014,zhang_confidence_2014}.
For concreteness, we focus on the estimator in \textcite{zhang_confidence_2014},
which is given by
\begin{equation*}
  \hat\beta_{\textnormal{ZZ}} =\hat\beta_{\textnormal{lasso}}
  +\frac{\tilde{w}_{\lambda}'(Y-w\hat\beta_{\textnormal{lasso}}
    -Z\hat\gamma_{\textnormal{lasso}})}{\tilde{w}_{\lambda}'w},
\end{equation*}
where $\hat\beta_{\textnormal{lasso}}, \hat\gamma_{\textnormal{lasso}}$ are the lasso
estimates from regressing $Y$ on $X$:
\begin{equation*}
  \hat\beta_{\textnormal{lasso}}, \hat\gamma_{\textnormal{lasso}}
  = \argmin_{\beta, \gamma} \norm{Y-w\beta-Z\gamma}_{2}^{2} + \tilde\lambda (\abs{\beta}
  +\norm{\gamma}_{1})
\end{equation*}
for some penalty parameter $\tilde \lambda>0$.

\begin{remark}\label{nonlinearity_remark}
  Note that $\hat\beta_{\textnormal{ZZ}}$ is nonlinear in $Y$, due to
  nonlinearity of the lasso estimates
  $\hat\beta_{\textnormal{lasso}}, \hat\gamma_{\textnormal{lasso}}$, which is
  consistent with the goal of efficiency in the nonconvex parameter
  space~\eqref{eq:sparsity_parameter_space}. In contrast,
  \Cref{optimality_corollary} shows that under the convex parameter space
  $\Gamma=\{\gamma\colon \norm{\gamma}_{1}\le C\}$, the estimator
  $\hat\beta_{\lambda}$ in~\eqref{eq:optimal_estimator} which only uses lasso in
  the propensity score regression of $w$ on $Z$, is already highly efficient
  among all estimators, so that there is no further role for substantive
  efficiency gains from the lasso regression of $Y$ on $X$, or from the use of
  other nonlinear estimators.
\end{remark}

To further understand the connection between these estimators, we note that
\textcite{zhang_confidence_2014} motivate their approach by bounds of the form
\begin{equation}\label{eq:lasso_gamma_l1_bound}
  \norm{\hat\gamma_{\textnormal{lasso}}-\gamma}_1 \le \tilde C\quad
  \text{where}\quad\tilde C= \text{const.} \cdot s\sqrt{\log k}/\sqrt{n},
\end{equation}
which hold with high probability with the constant depending on certain ``compatibility
constants'' that describe the regularity of the design matrix $X$
\parencite[see][Theorem 6.1, and references in the surrounding
discussion]{buhlmann_statistics_2011}.  This suggests correcting the initial
estimate $\hat\beta_{\textnormal{lasso}}$ by estimating
$\tilde\beta=\beta-\hat\beta_{\textnormal{lasso}}$ in the regression
\begin{equation*}
  \tilde Y= w(\beta-\hat\beta_{\textnormal{lasso}})+Z(\gamma-\hat\gamma_{\textnormal{lasso}})+\varepsilon
   = w\tilde \beta + Z\tilde \gamma + \varepsilon,
\end{equation*}
where
$\tilde Y=Y-w\hat\beta_{\textnormal{lasso}}-Z\hat\gamma_{\textnormal{lasso}}$.
Heuristically, we can treat the bound in \cref{eq:lasso_gamma_l1_bound} as a
constraint $\norm{\tilde\gamma}_{1}\le \tilde C$ on the unknown parameter
$\tilde\gamma=\gamma-\hat\gamma_{\textnormal{lasso}}$ and search for an optimal
estimator of $\tilde\beta=\beta-\hat\beta_{\textnormal{lasso}}$ under this
constraint. Applying the optimal estimator derived in
\Cref{thm:optimization_equivalence} then suggests estimating
$\beta-\hat\beta_{\textnormal{lasso}}$ with
$\frac{\tilde{w}_{\lambda}'\tilde Y}{\tilde{w}_{\lambda}'w}$. Adding this
estimate to $\hat\beta_{\textnormal{lasso}}$ gives the estimate
$\hat\beta_{\textnormal{ZZ}}$ proposed by \textcite{zhang_confidence_2014}. Whereas
\textcite{zhang_confidence_2014} motivate their approach as one possible way of
correcting the initial estimate $\hat\beta_{\textnormal{lasso}}$ using the bound
in \cref{eq:lasso_gamma_l1_bound}, the above analysis shows that their
correction is in fact identical to an approach in which one optimizes this
correction numerically.\footnote{The estimator proposed by
  \textcite{javanmard_confidence_2014} performs a numerical optimization of this
  form, but with the constraint~\eqref{eq:lasso_gamma_l1_bound} replaced by a
  constraint on
  $\abs{\hat\beta_{\textnormal{lasso}}-\beta}+\norm{\hat\gamma_{\textnormal{lasso}}-\gamma}_{1}$.
  Thus, \Cref{thm:optimization_equivalence} shows that a modification of the
  constraint used in \textcite{javanmard_confidence_2014} yields the same estimator
  as \textcite{zhang_confidence_2014}.}%

Under the bound in \cref{eq:lasso_gamma_l1_bound} it follows that
$\hat\beta_{\textnormal{ZZ}}-\beta =\tilde b + {a_\lambda}'\varepsilon$ where
$a_\lambda=\frac{\tilde{w}_\lambda}{\tilde{w}_{\lambda}'w}$ are the optimal
weights under the $\ell_1$ constraint $\norm{\tilde\gamma}_{1}\le \tilde C$,
given in \Cref{thm:optimization_equivalence}. Furthermore,
$\abs{\tilde b}\le \tilde C\overline B_\lambda$, with $\overline B_\lambda$
given in \Cref{thm:optimization_equivalence} and $\tilde C$ given in
\cref{eq:lasso_gamma_l1_bound}, and the variance of the random term
${a_\lambda}'\varepsilon$ is given by $\sd_{\lambda}^{2}$ in
\Cref{thm:optimization_equivalence}. Using arguments similar to those used to
prove \Cref{thm:upper_rate_bound}, it follows that
$\tilde C\overline B_\lambda/\sd_\lambda$ is bounded by a constant times
$s(\log k)/\sqrt{n}$, so that one can ignore bias in large samples as long as
this term converges to zero. This leads to the \ac{CI} proposed by
\textcite{zhang_confidence_2014}, which takes the form
\begin{equation}\label{eq:zz_CI}
\{\hat{\beta}_{\textnormal{ZZ}} \pm z_{1-\alpha/2} \widehat{\se}_\lambda\},
\end{equation}
where $\widehat{\se}_\lambda$ is an estimate of $\sd_\lambda$.

\begin{remark}\label{bias-aware_double_lasso_remark}
  To avoid the assumption that $s(\log k)/\sqrt{n}\to 0$ one could, in
  principle, extend our approach and the above analysis to form valid bias-aware
  \acp{CI} as
  $\{\hat\beta_{\textnormal{ZZ}} \pm [ \tilde{C}
  \overline{B}_\lambda+z_{1-\alpha/2}\widehat{\se}_\lambda] \}$.\footnote{We use
    the slightly more conservative approach of adding and subtracting the bound
    $\tilde {C}\overline{B}_\lambda$ rather than using the critical value
    $\cv_{\alpha}(\tilde {C}\overline{B}_\lambda/\widehat{\se}_\lambda)$ as in
    \cref{eq:linear-CI}, since the ``bias'' term for
    $\hat\beta_{\textnormal{ZZ}}$ is correlated with $\varepsilon$ through the
    first step estimates
    $\hat\beta_{\textnormal{lasso}}, \hat\gamma_{\textnormal{lasso}}$.}
  Unfortunately, finding a computable constant $\tilde C$
  in~\eqref{eq:lasso_gamma_l1_bound} that is sharp enough to yield useful bounds
  in practice appears to be difficult, although it is an interesting area for
  future research.
\end{remark}

\subsection{Verification of high-level conditions}\label{sec:verif-high-level}

The next result gives a set of low-level sufficient conditions under which
$X\in\mathcal{E}^{p}_{n}(\eta)$ with high probability.
\begin{lemma}\label{upper_bound_event_lemma}
  Suppose $w_i, z_i$ are drawn i.i.d.\ over $i$, and let
  $\delta=\argmin_b E[(w_i-z_i'b)^2]$ so that $z_{i}'\delta$ is the population
  best linear predictor of $w_{i}$. Suppose that the linear prediction
  error $E[(w_{i}-z_{i}'\delta)^{2}]$ is bounded away from zero as
  $k\to \infty$, $E[w_i^2]<\infty$, and that
  $\sup_j E[\abs{(w_{i}-z_{i}'\delta)z_{ij}}^{\max\{2,q\}}]<\infty$ when $p>1$, and, for some
  $c>0$, $P\left(\abs{(w_{i}-z_{i}'\delta)z_{ij}}\ge t \right)\le 2\exp(-c t^2)$ for all $j$
  when $p=1$. Then, for any $\tilde\eta>0$, there exists $\eta$ such that
  $X\in \mathcal{E}^{p}_n(\eta)$ with probability at least $1-\tilde\eta$ for large
  enough $n$.
\end{lemma}
\begin{proof}
  By the orthogonality condition for the best linear predictor, we have
  $E[w_{i}v_{i}]=E[v_{i}^{2}]$, where $v_{i}=w_{i}-z_{i}'\delta$, which is
  bounded from below uniformly over $k$ by assumption. Since $E[w_{i}v_{i}]$ is
  bounded from above by $Ew_i^2<\infty$, it follows from the law of large
  numbers for triangular arrays that
  $\frac{1}{n}\sum_{i=1}^n w_{i}v_{i}\ge \eta$ with probability approaching one
  once $\eta$ is small enough. Similarly,
  $\frac{1}{n}\sum_{i=1}^{n} v_{i}^{2}\le 1/\eta$ for small enough $\eta$ by the
  law of large numbers for triangular arrays.

For the last inequality in the definition of $\mathcal{E}^{p}_n(\eta)$, first
consider the case $p>1$ so that $q<\infty$. We then have
$E\norm{\frac{1}{\sqrt{n}}\sum_{i=1}^n z_{i}v_{i}}_{q}^{q} = E
\sum_{j=1}^k\abs{\sum_{i=1}^n v_{i}z_{ij}/\sqrt{n}}^{q} \le k\cdot K$ by
\textcite{von_bahr_convergence_1965}, where $K$ is a constant that depends only on
an upper bound for $\max_{j} E[\abs{v_{i}z_{ij}}^{\max\{q,2\}}]$. Applying
Markov's inequality gives the required bound. When $p=1$, then $q=\infty$ so
that
\begin{equation*}
  P\left(\norm*{\frac{1}{\sqrt{n}}\sum_{i=1}^n z_{i}v_{i}}_{q}\ge
    \eta^{-1}\sqrt{\log k} \right) \le \sum_{j=1}^k P\left(\abs*{
      \frac{1}{\sqrt{n}}\sum_{i=1}^n v_{i} z_{ij}} > \eta^{-1}\sqrt{\log k}
  \right),
\end{equation*}
which is bounded by
$2k\exp\left(-K\cdot \eta^{-2}\log k \right)=2k^{1-K\eta^{-2}}$ for some
constant $K$ by Hoeffding's inequality for sub-Gaussian random variables
\parencite[][Theorem 2.7.3]{vershynin26}. This can be made
arbitrarily small uniformly in $k$ by making $\eta$ small, as required.
\end{proof}

\Cref{upper_bound_event_lemma} requires that $w_i$ and $z_i$ be drawn from a
distribution such that the population prediction error from regressing $w_i$ on
$z_i$ is bounded away from zero as $k\to\infty$. This condition will hold, for
example, if $w_i=z_i'\delta+u_i$ with $u_i$ a nondegenerate random variable that
is independent of $z_i$. Note that we do not impose any conditions on the norm
or sparsity level of $\delta$ in this regression.

\section{Proofs}\label{sec:proofs}

This \namecref{sec:proofs} gives proofs for all results in the main text.

\subsection{Proof of Theorem~\ref{thm:optimization_equivalence}}

To prove \Cref{thm:optimization_equivalence}, we first explain how our results
fall into the general setup used in \textcite{donoho94}, \textcite{low95} and
\textcite{ArKo18optimal}. In the notation of \textcite{ArKo18optimal},
$(\beta, \gamma')'$ plays the role of the parameter $f$, the functional of
interest is given by $L(\beta, \gamma)=\beta$ and
$K(\beta, \gamma)=w\beta+Z\gamma$. The parameter space
$\mathbb{R}\times \Gamma$ is centrosymmetric, so that the modulus of continuity
\parencite[eq.~(25) in][]{ArKo18optimal} is given by
\begin{equation*}
  \omega(\delta) = \sup_{\beta, \gamma} 2\beta \quad\text{s.t.}\quad \norm{w\beta+Z\gamma}_{2}
  \le \delta/2,
  \quad
  \Pen(\gamma)\le C.
\end{equation*}
Using the substitution $\pi=-\gamma/ \beta$, we can write this as
\begin{equation}\label{eq:modulus_pi}
  \omega(\delta)=\sup_{\beta, \pi} 2\beta \quad\text{s.t.}\quad \beta \norm{w-Z\pi}_{2}
  \le \delta/2,
  \quad
  \beta \Pen(\pi)\le C.
\end{equation}
Let $\beta^{\text{mod}}_\delta, \gamma^{\text{mod}}_\delta$ and
$\pi^{\text{mod}}_\delta=-\gamma^{\text{mod}}_\delta/\beta^{\text{mod}}_\delta$
denote a solution to this problem when it exists. In the notation of
\textcite{ArKo18optimal},
$(\beta^{\text{mod}}_\delta, {\gamma^{\text{mod}}_\delta}')'$ plays the role of $g^*_\delta$, and the solution $(f^*_\delta, g^*_\delta)$ satisfies
$f^*_\delta=-g^*_\delta=-(\beta^{\text{mod}}_\delta, {\gamma^{\text{mod}}_\delta}')'$
by centrosymmetry.

This optimization problem is clearly related to the problem in
\cref{eq:pi_optimization}: we want to make $\norm{w-Z\pi}_{2}$ and $\Pen(\pi)$
small so that large values of $\beta$ satisfy the constraint
in~\eqref{eq:modulus_pi}. The following lemma formalizes the connection.

\begin{lemma}\label{modulus_solution_lemma}
  If there exists $\pi\in \sRk$ such that $w=Z\pi$ and $\Pen(\pi)=0$, then
  $\omega(\delta)=\infty$ for all $\delta\ge 0$. Otherwise, (i) for any
  $\delta>0$, the modulus problem in \cref{eq:modulus_pi} has a solution
  $\beta^{\text{mod}}_\delta, \pi^{\text{mod}}_\delta$ with
  $\beta^{\text{mod}}_\delta>0$. For
  $t=C/\beta^{\text{mod}}_\delta=2C/\omega(\delta)$, this solution
  $\pi^{\text{mod}}_\delta$ is also a solution to the penalized
  regression~\eqref{eq:pi_optimization} with optimized objective
  $\norm{w-Z\pi^{\text{mod}}_\delta}_{2}=\delta/\omega(\delta)>0$; and (ii) for
  any $t=t_\lambda>0$, the penalized regression
  problem~\eqref{eq:pi_optimization} has a solution $\pi_\lambda$. Setting
  $\beta_\lambda=C/t_\lambda$ and
  $\delta_\lambda=2\beta_\lambda\norm{w-Z\pi_\lambda}_2=(2C/t_\lambda)\norm{w-Z\pi_\lambda}_{2}$,
  the pair $\beta_\lambda, \pi_\lambda$ solves the modulus
  problem~\eqref{eq:modulus_pi} at $\delta=\delta_\lambda$, with optimized
  objective $\omega(\delta_\lambda)=2C/t_\lambda$, so long as
  $\norm{w-Z\pi_\lambda}_2>0$.
\end{lemma}
\begin{proof}
  If there exists $\pi\in \sRk$ such that $w=Z\pi$ and $\Pen(\pi)=0$, then the
  result is immediate. Suppose there does not exist such a $\pi$.

  First, we show that the problem in \cref{eq:pi_optimization} has a solution.
  Let $\sRk^{(0)}$ denote the linear subspace of vectors $\pi\in \sRk$ such that
  $Z\pi=0$ and $\Pen(\pi)=0$, and let $ \sRk^{(1)}$ be a subspace such that
  $\sRk= \sRk^{(0)}\oplus \sRk^{(1)}$, so that we can write $\pi\in \sRk$
  uniquely as $\pi=\pi^{(0)}+\pi^{(1)}$ where $\pi^{(0)}\in \sRk^{(0)}$ and
  $\pi^{(1)}\in \sRk^{(1)}$. Note that $Z\pi=Z\pi^{(1)}$ and, applying the
  triangle inequality twice,
  $\Pen(\pi^{(1)})=\Pen(\pi^{(1)}) - \Pen(-\pi^{(0)}) \le \Pen(\pi)\le
  \Pen(\pi^{(0)})+\Pen(\pi^{(1)})=\Pen(\pi^{(1)})$ so that
  $\Pen(\pi)=\Pen(\pi^{(1)})$. Thus, the problem~\eqref{eq:pi_optimization} can
  be written in terms of $\pi^{(1)}\in \sRk^{(1)}$ only. The level sets of this
  optimization problem are bounded and are closed by continuity of the seminorm
  $\Pen(\cdot)$ \parencite{goldberg_continuity_2017}, and so it has a solution,
  which is also a solution in the original problem. Similarly, to show that the
  problem~\eqref{eq:modulus_pi} has a solution, note that feasible values of
  $\beta$ are bounded by a constant times the inverse of the minimum of
  $\max\{\norm{w-Z\pi}_2, \Pen(\pi)\}$ over $\pi$, which is strictly positive by
  continuity of $\Pen(\pi)$ and the fact that there does not exist $\pi$ with
  $\max\{\norm{w-Z\pi}_2, \Pen(\pi)\}=0$. Thus, we can restrict
  $\beta, \pi^{(1)}$ to a compact set without changing the optimization problem.

  To show the first statement in the lemma, note that
  $\beta^{\text{mod}}_\delta>0$, since it is feasible to set $\pi=0$ and
  $\beta=\delta/(2\norm{w}_2)$, and that
  $\norm{w-Z\pi^{\text{mod}}_\delta}_2>0$, since otherwise a strictly larger
  value of $\beta$ could be achieved by multiplying $\pi^{\text{mod}}_\delta$ by
  $1-\eta$ for $\eta>0$ small enough. Now, if the first statement did not hold,
  there would exist a $\tilde\pi$ with
  $\Pen(\tilde \pi)\le C/\beta^{\text{mod}}_\delta$ such that
  $\norm{w-Z\tilde\pi}_2\le \norm{w-Z\pi^{\text{mod}}_\delta}_2-\nu$ for small
  enough $\nu>0$. Then, letting $\tilde\pi_\eta=(1-\eta)\tilde \pi$, we would
  have
  $\norm{w-Z\tilde\pi_\eta}_2\le \norm{w-Z\tilde\pi}_2 + \eta\norm{Z\tilde\pi}_2
  \le \norm{w-Z\pi^{\text{mod}}_\delta}_2 - \nu + \eta\norm{Z\tilde\pi}_2\le
  \delta/(2\beta^{\text{mod}}_\delta)-\nu+\eta\norm{Z\tilde\pi}_2$. Thus,
  $\norm{w-Z\tilde\pi_\eta}_{2}$ will be strictly less than
  $\delta/(2\beta^{\text{mod}}_\delta)$ for small enough $\eta$ and
  $\Pen(\tilde{\pi}_\eta)\le
  (1-\eta)C/\beta^{\text{mod}}_\delta<C/\beta^{\text{mod}}_\delta$. This is a
  contradiction, since it would allow a strictly larger value of $\beta$ by
  setting $\pi=\tilde\pi_\eta$.

  The second statement follows immediately, since any pair
  $\tilde\beta, \tilde\pi$ satisfying the constraints in the
  modulus~\eqref{eq:modulus_pi} for $\delta=\delta_\lambda$ with
  $\tilde\beta>\beta_\lambda$ would have to have
  $\norm{w-Z\tilde\pi}_2<\norm{w-Z\pi_\lambda}_{2}$ while maintaining the
  constraint $\Pen(\tilde{\pi})\le t_\lambda$.
\end{proof}

We now prove \Cref{thm:optimization_equivalence}. The class of bias-variance
optimizing estimators, $\hat L_\delta$ in the notation of \textcite{ArKo18optimal},
is given by
$\frac{(w\beta^{\text{mod}}_\delta +
  Z\gamma^{\text{mod}}_\delta)'Y}{(w\beta^{\text{mod}}_\delta +
  Z\gamma^{\text{mod}}_\delta)'w}$, where we use eq.~(26) in
\textcite{ArKo18optimal} to compute the form of this estimator under
centrosymmetry, and Lemma D.1 in \textcite{ArKo18optimal} to calculate the
derivative $\omega'(\delta)$, since the problem is translation invariant with
$\iota$ given by the parameter $\beta=1$, $\gamma=0$. Given $\lambda$ with
$\norm{w-Z\pi_\lambda}_{2}>0$, it follows from \Cref{modulus_solution_lemma}
that, for $\delta_\lambda$ given in the lemma, this estimator
$\hat{L}_{\delta_\lambda}$ is equal to $\hat\beta_\lambda=a_\lambda'Y$ where
$a_\lambda=\frac{w-Z\pi_{\lambda}}{(w-Z\pi_{\lambda})'w}$, as defined in
\Cref{thm:optimization_equivalence}. The worst-case bias formula in
\Cref{thm:optimization_equivalence} then follows from the fact that the maximum
bias is attained at
$\gamma=-\gamma^{\text{mod}}_{\delta_\lambda}=Ct_\lambda^{-1}\pi_\lambda$ by
Lemma A.1 in \textcite{ArKo18optimal} \parencite[or Lemma 4 in][]{donoho94}.

\subsection{Proof of Corollary~\ref{optimality_corollary}}\label{sec:proof-coroll-optimality}

Part~\ref{item:bias-variance} of \Cref{optimality_corollary} follows from
\textcite{low95}. In particular, consider the one-dimensional submodel
$\beta\in[-C/\Pen(\pi_{\lambda}), C/\Pen(\pi_{\lambda})]$,
$\gamma=-\pi_{\lambda}\beta$. Let
$b_{\lambda}=(w-Z\pi_{\lambda})/\norm{w-Z\pi_{\lambda}}^{2}_{2}$, and let
$B\in\mathbb{R}^{n\times (n-1)}$ be a matrix with orthonormal columns that is
orthogonal to $b_{\lambda}$. Note that in this submodel,
$B'Y=B'(w-Z\pi_{\lambda})\beta+B'\varepsilon=B'\varepsilon$, which does not
depend on the unknown parameter $\beta$, and is independent of $b_{\lambda}'Y$.
Therefore, $b_{\lambda}'Y\sim\ND(\beta, \norm{b_{\lambda}}_{2}^{2}\sigma^{2})$
is a sufficient statistic in this submodel. By Theorem 1 in \textcite{low95}, in
this submodel, the estimator
$\hat{\beta}_{\lambda}={a_{\lambda}}'Y=\kappa b_{\lambda}'Y$, where
$\kappa=\norm{w-Z\pi_{\lambda}}^{2}_{2}/(w-Z\pi_{\lambda})'w$ minimizes
$\sup_{\beta}\var(\delta(Y))$ among all estimators $\delta(Y)$ with
$\sup_{\beta}\abs{E_{\beta}[\delta(Y)]-\beta}\leq
(1-\kappa)C/\Pen(\pi_{\lambda})=C\overline{B}_{\lambda}$, and, likewise, it
minimizes $\sup_{\beta}\abs{E_{\beta}[\delta(Y)]-\beta}$ among all estimators
with
$\sup_{\beta}\var(\delta(Y))\leq
\kappa^{2}\sigma^{2}\norm{b_{\lambda}}_{2}^{2}=\sd_{\lambda}^{2}$. Since the
worst-case bias
$\maxbias_{\Gamma}(\hat{\beta}_{\lambda})= C\overline{B}_{\lambda}$ and standard
deviation $\sd(\hat{\beta}_{\lambda})= \sd_{\lambda}$ are the same in the full
model by~\Cref{thm:optimization_equivalence}, the result follows.

Part~\ref{item:mse} of \Cref{optimality_corollary} is immediate from
\textcite{donoho94}. In particular, it holds with
\begin{equation*}
  \kappa^{*}_{\textnormal{MSE}}(X, \sigma, \Gamma) =
  \frac{\sup_{\delta>0}(\omega(\delta)/\delta)^2
    \rho_{N}(\delta/2,\sigma)}{\sup_{\delta>0}(\omega(\delta)/\delta)^{2}
    \rho_{A}(\delta/2,\sigma)}
  \ge 0.8,
\end{equation*}
where $\omega(\delta)$ is defined in \cref{eq:modulus_pi}, and $\rho_{A}$ and
$\rho_{N}$ are the minimax risk among affine estimators, and among all estimators,
respectively, in the bounded normal means problem $Y\sim \ND(\theta, \sigma^2)$,
$\abs{\theta}\le \tau$, defined in \textcite{donoho94}, and the last inequality
follows from eq.~(4) in \textcite{donoho94}.

Finally, Part~\ref{item:flci} of \Cref{optimality_corollary} follows from
Corollary 3.3 in \textcite{ArKo18optimal}, with
\begin{equation*}
  \kappa^*_{\textnormal{CI}}(X, \sigma,
  \Gamma,\alpha)=\frac{(1-\alpha)E\left[\omega(2\sigma(z_{1-\alpha}-Z))\mid Z\le
      z_{1-\alpha}\right]}{2\min_\delta
    \cv_\alpha\left(\frac{\omega(\delta)}{2\sigma\omega'(\delta)}
      -\frac{\delta}{2\sigma}\right)\cdot\sigma\omega'(\delta)},
\end{equation*}
where $Z\sim\ND(0,1)$, $\omega(\delta)$ is given in \cref{eq:modulus_pi}, and by
Lemma D.1 in \citeauthor{ArKo18optimal}, since the problem is translation
invariant with $\iota$ given by the parameter $\beta=1$, $\gamma=0$,
$\omega'(\delta)=\delta/[w'(w-Z\pi^{\text{mod}}_{\delta})\cdot \omega(\delta)]$.
The universal lower bound 0.717 when $\alpha=0.05$ follows from Theorem 4.1 in
\textcite{ArKo20sensitivity}.

\subsection{Proof of Theorem~\ref{theorem:li_mueller_rate}}

Let
$L_{\lambda}= \lambda R_{w}\kappa
+\sqrt{1+\frac{R_{w}^{2}}{1-R^{2}_{w}}(1-\lambda)^{2}}$. Since
$b+z_{1-\alpha}\leq \cv_{\alpha}(b)\leq b+z_{1-\alpha/2}$, it follows that the ratio
$\cv_{\alpha}(b_{\lambda})\sd_{\lambda}/\sigma_{\textnormal{short}}L_{\lambda}$
is bounded above and below by positive constants. Hence, it suffices to show
that
$\min_{\lambda\in[0,1]}L_{\lambda}\asymp 1+\min\{\kappa,1/\sqrt{1-R_{w}^{2}}\}$.

Suppose first that $\kappa^{2}/(1+\kappa^{2})>R_{w}^{2}$. Then the result
follows from
$\min_{\lambda}L_{\lambda}=L_{0}=1/\sqrt{1-R_{w}^{2}}$. Since
$1+1/\sqrt{1-R_{w}^{2}}\leq 2/\sqrt{1-R_{w}^{2}}\leq 2 \sqrt{1+\kappa^{2}}\leq
2(1+\kappa)$, it follows that $1/\sqrt{1-R_{w}^{2}}\asymp
1+\min\{\kappa,1/\sqrt{1-R_{w}^{2}}\}$, which yields the claim.

Consider now the case $\kappa^{2}/(1+\kappa^{2})\leq R_{w}^{2}$. Then
$L_{\lambda}$ is minimized at
$\lambda=1-\sqrt{\frac{\kappa^{2}(1-R_{w}^{2})^{2}}{
    R_{w}^{2}-\kappa^{2}R_{w}^{2}(1-R_{w}^{2})}}$, with
$\min_{\lambda}L_{\lambda}=\ell(R_{w})$, where
$\ell(R_{w})=R_{w}\kappa+\sqrt{1-\kappa^{2}(1-R_{w}^{2})}$. We claim that
$\ell(R_{w})\asymp \kappa+1$. The upper bound follows from $\ell(R_{w})\leq
\kappa+1$ and the lower bound from the fact that $\ell(R_{w})$ is increasing in
$R_{w}$ so that
\begin{equation*}
  \ell(R_{w})\geq \min_{R_{w}\colon \kappa^{2}/(1+\kappa^{2})\leq R_{w}^{2}}\ell(R_{w})
  =\ell(\kappa/\sqrt{1+\kappa^{2}})=\sqrt{1+\kappa^{2}}\geq (1+\kappa)/\sqrt{2}.
\end{equation*}
Hence, $\min_{\lambda}L_{\lambda}\asymp 1+\kappa$. Since
$\kappa\leq R_{w}/\sqrt{1-R_{w}^{2}}\leq 1/\sqrt{1-R_{w}^{2}}$, this yields the
claim.

\subsection{Proof of Theorem~\ref{thm:upper_rate_bound}}

To prove that the claimed upper bound holds for $X\in \mathcal{E}^{p}_n(\eta)$,
we first note that, since the \ac{CI} based on
$\hat\beta_{\lambda^{*}_{\textnormal{CI}}}$ is shorter than the \ac{CI} based on
any linear estimator $a'Y$, it suffices to show that there exists a sequence of
weight vectors $a$ such that the worst-case bias and standard deviation are
bounded by constants times $R^{*}_{p}(k,C,R_{w}^{2})$. Consider first the
weights $\tilde{a}_i=\frac{v_i}{\sum_{j=1}^{n}v_{j}w_{j}}$, where
$v_i=w_i-z_i'\delta$, with $\delta$ given in the definition of
$\mathcal{E}^{p}_{n}(\eta)$. The variance of the estimator $\tilde a'Y$ is
$\frac{\sigma^{2}\sum_{i=1}^n v_i^2}{\left(\sum_{i=1}^{n}v_{i} w_i \right)^2}\le
\sigma^{2}\eta^{-3}/n\leq \eta^{-5}/n$. The worst-case bias is
\begin{equation*}
  \sup_{\gamma\colon \norm{\gamma}_p\le C} \tilde a'Z\gamma
  =C\norm{Z'\tilde a}_{q}
  =n^{-1/2}C\frac{n^{-1/2}\norm{Z'(w-Z\delta)}_{q}}{n^{-1}\abs{w'(w-Z\delta)}}
  \leq C \frac{r_{p}(k)/\sqrt{n}}{\eta^{2}},
\end{equation*}
where the first equality follows by Hölder's inequality, and the last inequality
follows by definition of $\mathcal{E}^{p}_{n}(\eta)$. Since
$z_{1-\alpha}+b\leq \cv_{\alpha}(b)\leq z_{1-\alpha/2}+b$, it follows that
choosing these weights yields the convergence rate $(1+Cr_{p}(k))/\sqrt{n}$.

Next, consider the long regression estimator, which has no bias and standard
deviation given by
\begin{equation*}
  \frac{\sigma}{\norm{\tilde{w}}_{2}}\asymp
  n^{-1/2}\frac{1}{\sqrt{1-R_{w}^{2}}}
  \asymp
  n^{-1/2}\left(1+\frac{1}{\sqrt{1-R_{w}^{2}}}\right),
\end{equation*}
where $\tilde{w}=(I-ZZ^{+})w$, and we use the fact that
$\norm{\tilde{w}}_{2}^{2}=(1-R_{w}^{2})\norm{w}_{2}^{2}\asymp (1-R_{w}^{2})n$.
Hence, the claimed upper bound is achieved by running the long regression if
$Cr_{p}(k)>1/\sqrt{1-R_{w}^{2}}$, and using the weights $\tilde{a}$ otherwise.

\subsection{Proof of Theorem~\ref{theorem:l2_rates}}

\paragraph{Step 1: Closed-form bias and variance expressions}

Solving the minimization problem yields
$\pi_{\lambda}=(Z'Z+\lambda I_{k})^{-1}Z'w$, so that letting
$H_{\lambda}=Z(Z'Z+\lambda I_{k})^{-1}Z'$, we have
\begin{equation*}
  a_{\lambda}=\frac{(I_{n}-H_{\lambda})w}{
    w'(I_{n}-H_{\lambda})w},
\end{equation*}
and
\begin{equation*}
  \overline{B}_{\lambda}
  =\frac{w'(H_{\lambda}-H_{\lambda}^{2})w}{
    w'(I_{n}-H_{\lambda})w\cdot \norm{(Z'Z+\lambda I_{k})^{-1}Z'w}_{2}}
  =\frac{
    \lambda \norm{(Z'Z+\lambda I_{k})^{-1}
      Z'w}_{2}}{
    w'(I_{n}-H_{\lambda})w}.
\end{equation*}
where the second equality uses the observation that
\begin{equation*}
  H_{\lambda}-  H_{\lambda}^{2}=H_{\lambda}-Z(Z'Z+\lambda I_{k})^{-1}
  Z'Z(Z'Z+\lambda I)^{-1}
  Z'
  =\lambda Z(Z'Z+\lambda I_{k})^{-2} Z'.
\end{equation*}
Let $Z=UDV'$ denote the singular value decomposition of $Z$, where $D$ is an
$r\times r$ matrix collecting its singular values, $r=\operatorname{rank}(Z)$,
$U$ is $n\times r$, and $V$ is $k\times r$. Then we may write
\begin{equation*}
  (Z'Z+\lambda I_{k})^{-1}=V(D^{2}+\lambda I_{r})^{-1}V'+(I_{k}-VV')/\lambda,
\end{equation*}
so that $H_{\lambda}=U \operatorname{diag}(d_{j}^{2}/(\lambda+d_{j}^{2}))U'$.
Hence, $  (I_{n}-H_{\lambda})w=
  \tilde{w} +U \operatorname{diag}(\lambda/(\lambda+d_{j}^{2}))U'w$, so that
\begin{equation*}
  \norm{a_{\lambda}}_{2}=\frac{
    \left(\norm{\tilde{w}}_{2}^{2}
      +
      \lambda^{2}    \sum_{j=1}^{r}\frac{(U'w)_{j}^{2}}{(\lambda+d_{j}^{2})^{2}}\right)^{1/2}
  }{
    \norm{\tilde{w}}_{2}^{2} +
    \lambda \mathcal{D}_{\lambda}
  },\quad
  \mathcal{D}_{\lambda}=\sum_{j=1}^{r}\frac{(U'w)_{j}^{2}}{\lambda+d_{j}^{2}}.
\end{equation*}
Similarly, using the fact that
$\norm{(Z'Z+\lambda I_{k})^{-1}Z'w}_{2}^{2}=\norm{V(D^{2}+\lambda
  I_{r})^{-1}DU'w}_{2}^{2} =\sum_{j=1}^{r}\frac{d_{j}^{2}(U'w)_{j}^{2}
}{(d_{j}^{2}+\lambda)^{2}} $, we obtain
\begin{equation*}
  \overline{B}_{\lambda}=\frac{
    \lambda \norm{(Z'Z+\lambda I_{k})^{-1}
      Z'w}_{2}
  }{
    \norm{\tilde{w}}_{2}^{2} +
    \lambda \mathcal{D}_{\lambda}
  }=
  \frac{
    \lambda\left(\sum_{j=1}^{r}\frac{d_{j}^{2}(U'w)_{j}^{2}
}{(d_{j}^{2}+\lambda)^{2}}\right)^{1/2}
  }{
    \norm{\tilde{w}}_{2}^{2} +
    \lambda \mathcal{D}_{\lambda}
  }
\end{equation*}
Combining the expressions for $\norm{a_{\lambda}}_{2}$ and
$\overline{B}_{\lambda}$ gives
\begin{equation*}
  \sd_{\lambda}+C\overline{B}_{\lambda}
  \asymp f(\lambda),\qquad
  f(\lambda)=\frac{
    \norm{\tilde{w}}_{2}
    +\lambda    \left(
      \sum_{j=1}^{r}\frac{(1+C^{2} d_{j}^{2})(U'w)_{j}^{2}}{(\lambda+d_{j}^{2})^{2}}\right)^{1/2}
  }{
    \norm{\tilde{w}}_{2}^{2} +
    \lambda \mathcal{D}_{\lambda}
  }
\end{equation*}
Therefore, to prove the lower bound assertion, it suffices to show that
$\min_{\lambda}f(\lambda)$ matches the claimed rates.

\paragraph{Step 2: lower bounds}

Observe that
\begin{equation*}
  \frac{R_{w}^{2}n}{\lambda+d_{\max}^{2}}
  \lesssim
  \frac{R_{w}^{2}\norm{w}_{2}^{2}}{\lambda+d_{\max}^{2}}
  \leq
  \mathcal{D}_{\lambda}\leq \frac{R_{w}^{2}\norm{w}_{2}^{2}}{\lambda+d_{\min}^{2}}\lesssim
  \frac{R_{w}^{2}n}{\lambda+d_{\min}^{2}},
\end{equation*}
Using these bounds, we obtain the lower bound
\begin{multline*}
  f(\lambda)\geq
  \frac{
    \norm{\tilde{w}}_{2}
    +\lambda\left(\min_{j}\frac{1+C^{2}d_{j}^{2}}{\lambda+d_{j}^{2}}\right)^{1/2}   \mathcal{D}_{\lambda}^{1/2}
  }{
    \norm{\tilde{w}}_{2}^{2} +
    \lambda \mathcal{D}_{\lambda}
  }\gtrsim
  \frac{
    \norm{\tilde{w}}_{2}
    +\lambda\left(\min_{j}\frac{1+C^{2}d_{j}^{2}}{\lambda+d_{j}^{2}}\right)^{1/2}
    \frac{R_{w}\sqrt{n}}{\sqrt{\lambda+d_{\max}^{2}}}
  }{
    \norm{\tilde{w}}_{2}^{2} +
    \frac{\lambda R_{w}^{2}n}{\lambda+d_{\min}^{2}}
  }\\
  \asymp  \frac{
    \norm{\tilde{w}}_{2}
    +\lambda\left(\min_{j}\frac{1+C^{2}d_{j}^{2}}{\lambda+d_{j}^{2}}\right)^{1/2}
    \frac{R_{w}\sqrt{n}}{\sqrt{\lambda+d_{\max}^{2}}}
  }{
    \norm{\tilde{w}}_{2}^{2} +
    \frac{\lambda R_{w}^{2}n}{\lambda+d_{\max}^{2}}
  }
\end{multline*}
where the second line uses the fact that since $d_{\min}/d_{\max}\asymp 1$,
$(\lambda+d_{\min}^{2})/(\lambda+d_{\max}^{2})\asymp 1$ uniformly over
$\lambda\geq 0$.

Since the function $(1+C^{2}x)/(x+\lambda)$ is decreasing if
$\lambda C^{2}\leq 1$, it follows that if $\lambda C^{2}\leq 1$,
$\min_{j}\frac{1+C^{2}d_{j}^{2}}{\lambda+d_{j}^{2}}=\frac{1+C^{2}d_{\max}^{2}}{\lambda+d_{\max}^{2}}$,
so that
\begin{equation*}
  f(\lambda)\gtrsim
  \frac{
    \norm{\tilde{w}}_{2}(\lambda+d_{\max}^{2})
    +\lambda(1+C^{2}d_{\max}^{2})^{1/2}
    R_{w}\sqrt{n}
  }{
    \norm{\tilde{w}}_{2}^{2}(\lambda+d_{\max}^{2}) +
    \lambda R_{w}^{2}n
  }.
\end{equation*}
Since this is a ratio of two affine functions, it is monotone in $\lambda$, so
that the minimum over $\lambda\in[0,1/C^{2}]$ is taken at one of the extremes.
Thus, letting $h(C)=\frac{R_{w}^{2}n}{1+ C^{2}d_{\max}^{2}}$,
\begin{multline}\label{eq:f1c}
  \min_{\lambda\in[0,1/C^{2}]}f(\lambda)\gtrsim
  \min\left\{  \frac{
      1
    }{
      \norm{\tilde{w}}_{2}
    },
    \frac{
      \norm{\tilde{w}}_{2}+
      h(C)^{1/2}
    }{     \norm{\tilde{w}}_{2}^{2}+
      h(C)  }
  \right\}
  \asymp
  \min\left\{
    \frac{   1    }{      \norm{\tilde{w}}_{2}    }
    ,
    \frac{
      1
    }{
      \norm{\tilde{w}}_{2}+
      h(C)^{1/2}
    }
  \right\}\\=    \frac{
    1
  }{
    \norm{\tilde{w}}_{2}+
    h(C)^{1/2}
  }\asymp
  \frac{1}{n^{1/2}}\frac{
  1+        C d_{\max}
  }{
    \sqrt{1-R_{w}^{2}}C d_{\max}+
    1
  }
  ,
\end{multline}
where we use $\sqrt{a+b}\asymp \sqrt{a}+\sqrt{b}$, and
$\norm{\tilde{w}}_{2}^{2}=(1-R_{w}^{2})\norm{w}_{2}^{2}\asymp (1-R_{w}^{2})n$.
Similarly, if $\lambda C^{2}\geq 1$,
$\min_{j}\frac{1+C^{2}d_{j}^{2}}{\lambda+d_{j}^{2}}=\frac{1+C^{2}d_{\min}^{2}}{\lambda+d_{\min}^{2}}$,
so that
\begin{equation*}
  f(\lambda)  \gtrsim
  \frac{
    \norm{\tilde{w}}_{2}
    +\lambda\left(\frac{1+C^{2}d_{\min}^{2}}{\lambda+d_{\min}^{2}}\right)^{1/2}
    \frac{R_{w}\sqrt{n}}{\sqrt{\lambda+d_{\max}^{2}}}
  }{
    \norm{\tilde{w}}_{2}^{2} +
    \frac{\lambda R_{w}^{2}n}{\lambda+d_{\max}^{2}}
  }\asymp
  \frac{
    \norm{\tilde{w}}_{2}(\lambda+d_{\max}^{2})
    +\lambda\left(1+C^{2}d_{\min}^{2}\right)^{1/2}
    R_{w}\sqrt{n}
  }{
    \norm{\tilde{w}}_{2}^{2}(\lambda+d_{\max}^{2}) +
    \lambda R_{w}^{2}n
  },
\end{equation*}
which is again a ratio of affine functions, so that the minimum over
$\lambda\in[1/C^{2},\infty]$ is taken at one of the extremes. Using the same
arguments as in~\cref{eq:f1c}, we therefore obtain
\begin{multline*}
  \min_{\lambda\in[1/C^{2},\infty]}f(\lambda)
  \gtrsim
  \min\{  \frac{1}{n^{1/2}}\frac{
    1+        C d_{\min}
  }{
    \sqrt{1-R_{w}^{2}}C d_{\max}+
    1
  },
  \frac{
    \norm{\tilde{w}}_{2}
    +\left(1+C^{2}d_{\min}^{2}\right)^{1/2}
    R_{w}\sqrt{n}
  }{
    \norm{\tilde{w}}_{2}^{2}+
     R_{w}^{2}n
  }
  \}\\
  \asymp
  \min\left\{  \frac{1}{n^{1/2}}\frac{
    1+        C d_{\min}
  }{
    \sqrt{1-R_{w}^{2}}C d_{\max}+
    1
  },
  \frac{
    1
    +C d_{\min}  R_{w}
  }{
    \sqrt{n}
  }\right\}.
\end{multline*}
Since this minimum is smaller than the right-hand side of \cref{eq:f1c}, it
follows that
\begin{multline*}
  \min_{\lambda\geq 0}f(\lambda)\gtrsim
  \frac{1}{\sqrt{n}}  \min\left\{  \frac{
      1+        C d_{\min}
    }{
      \sqrt{1-R_{w}^{2}}C d_{\max}+
      1
    },
    1
    +C d_{\min}  R_{w}
  \right\}\\\asymp
  \frac{1}{\sqrt{n}}  \min\left\{  \frac{
      1+        C d_{\max}
    }{
      \sqrt{1-R_{w}^{2}}C d_{\max}+
      1
    },
    1
    +C d_{\max}  R_{w}
  \right\}.
\end{multline*}
Observe that for $a\in[0,1]$ and $x\geq 0$,
$\frac{1+x}{1+ax}\geq \frac{1+\min\{x,1/a\}}{4}$. Indeed, if $ax\leq 1$, then
$\frac{1+x}{1+ax}\geq \frac{1+x}{2} =\frac{1+\min\{x,1/a\}}{2}$, while if
$ax>1$,
$\frac{1+x}{1+ax}\geq \frac{1+x}{2ax}\geq \frac{1}{2a}\geq
\frac{1}{4}\left(1+\frac{1}{a}\right)=\frac{1}{4}\left(1+\min\{x,1/a\}\right)$.
Applying this inequality to the first minimand in the preceding display then
yields that it is bounded below by
$\frac{1+\min\{Cd_{\max},1/\sqrt{1-R_{w}^{2}}\}}{4}$. The result then follows
from observing that
$1+\min\{\frac{1}{\sqrt{1-R_{w}^{2}}},R_{w}C d_{\max}\}\asymp
1+\min\{\frac{1}{\sqrt{1-R_{w}^{2}}},C d_{\max}\}$.

\subsection{Proof of Theorem~\ref{thm:l1_lower_bound}}
Since the lower bound $c\cdot n^{-1/2}$ follows from standard efficiency bounds
with finite dimensional parameters (e.g., taking the submodel where
$\delta=\gamma=0$), it suffices to show that
$E_{\vartheta^*}\hat\chi \ge c\cdot C \sqrt{\log k}/\sqrt{n}$ for some $c>0$
once $n$ is large enough. To show this, we follow essentially the same arguments
as \textcite[][Theorem 3]{cai_lasso_2017} and \textcite[][Proposition
4.2]{javanmard_debiasing_2018}, checking that the distributions used in the
lower bound satisfy the constraints that define
$\Theta(C, C\cdot K\sqrt{n/\log k}, \eta)$. Under a given parameter vector
$\vartheta=(\beta, \gamma', \delta', \sigma^2,\sigma^2_{v})$, the data
$(Y_i, w_i, z_i')'$ are i.i.d.\ normal with mean zero and variance matrix
\begin{equation*}
  \Sigma_\vartheta=
  \begin{pmatrix}
    \sigma^2 + \beta^2(\sigma^2_{v}+\norm{\delta}_2^2) + 2\beta\delta'\gamma
    +\norm{\gamma}_{2}^{2} & \beta(\sigma^2_{v} + \norm{\delta}_{2}^{2}) +
    \gamma'\delta
    & \beta\delta' + \gamma' \\
    \beta(\sigma^2_{v} + \norm{\delta}_{2}^{2}) + \gamma'\delta
    & \sigma^2_{v} + \norm{\delta}_2^2 & \delta' \\
    \beta\delta + \gamma & \delta & I_{k}
  \end{pmatrix}.
\end{equation*}
Let $f_\pi$ denote the distribution of the data $\{Y_i, w_i, z_i\}_{i=1}^n$ when
the parameters follow a prior distribution $\pi$, and let
$\chi^2(f_{\pi_0}, f_{\pi_1})$ denote the chi-square distance between these
distributions for prior distributions $\pi_0$ and $\pi_1$. By Lemma 1 in
\textcite{cai_lasso_2017}, it suffices to find a prior distribution $\pi_1$
supported on the parameter space $\Theta(C, C\cdot K\sqrt{n/\log k}, \eta)$
such that $\pi_1$ places probability one on $\beta=\beta_{1,n}$ for some
sequence with $\abs{\beta_{1,n}}$ bounded from below by a constant times
$C \sqrt{\log k}/\sqrt{n}$, and such that $\chi^2(f_{\pi_0}, f_{\pi_1})\to 0$,
where $\pi_0$ is the distribution that places probability one on $\vartheta^*$.

We first note that we can assume $\sigma^2_0=\sigma^2_{v, 0}=1$ without loss of
generality. Dividing $Y_i$ by $\sigma_0$ and $w_i$ by $\sigma_{v, 0}$ leads to
the same model with $\beta$, $\gamma$, $\delta$, $\sigma^{2}$ and
$\sigma^{2}_{v}$ multiplied by constants that depend only on $\sigma_0$ and
$\sigma_{v, 0}$.

We may also assume $\tilde\eta<1/2$, since decreasing $\tilde\eta$ only weakens
the condition $C\le \sqrt{k/n}\cdot k^{-\tilde\eta}$. Fix a constant
$c_{\nu}\in(0,\min\{2\tilde{\eta}, c_{0}^{2}\})$, and set
$\nu=\frac{\sqrt{c_{\nu}}}{2}\cdot\frac{\sqrt{\log k}}{\sqrt{n}}$ and
$s=\left\lfloor \frac{C}{2\nu}\right\rfloor$. Let $\pi_1$ be defined by a
uniform prior for $\delta$ over the set of vectors with $\norm{\delta}_{0}=s$
and each element $\delta_j\in\{0,\nu\}$. We set the remaining parameters as
deterministic functions of $\delta$:
\begin{equation*}
  \beta_{1, n}=-\frac{\norm{\delta}_{2}^{2}}{1-\norm{\delta}_{2}^{2}},\qquad
  \gamma=(1-\beta_{1, n})\delta=\frac{\delta}{1-\norm{\delta}_{2}^{2}},\qquad
  \sigma^2_{v}=1-\norm{\delta}_{2}^{2},\qquad
  \sigma^2=\frac{1-2\norm{\delta}_{2}^{2}}{1-\norm{\delta}_{2}^{2}}.
\end{equation*}
Note that $\norm{\delta}_{2}^{2}=s\nu^{2}$ is constant under this prior, so the
prior for $\beta_{1, n}$ is a unit point mass as required. Also, by the lower
bound on $C$, and $c_{\nu}\le c_{0}^{2}$, $C/(2\nu)\geq 1$, so that $s\geq 1$.
Furthermore, $\norm{\delta}_{2}^{2}=s \nu^{2}\leq C\nu/2$, with
$C\nu=\frac{\sqrt{c_{\nu}}}{2}\cdot C\sqrt{\frac{\log k}{n}}\to 0$, which
implies that $\norm{\delta}_{2}^{2}<1/2$ for $n$ large enough, so that the
variance parameters are positive. We now verify that $\pi_1$ is supported on
$\Theta(C, C\cdot K\sqrt{n/\log k}, \eta)$ for $n$ large enough. Since
$s\le C/(2\nu)$, $\norm{\delta}_{1}=s\nu\le C/2$ and, once
$\norm{\delta}_{2}^{2}\le 1/2$,
$\norm{\gamma}_{1}=s\nu/(1-\norm{\delta}_{2}^{2})\le C$. Also,
$\norm{\gamma}_{0}=\norm{\delta}_{0}=s\le C/(2\nu) =c_{\nu}^{-1/2}C\sqrt{n/\log
  k}\le C\cdot K\sqrt{n/\log k}$ once $K\ge c_{\nu}^{-1/2}$. Finally,
$\abs{\sigma^{2}_{v}-1}=\norm{\delta}_{2}^{2}$ and
$\abs{\sigma^{2}-1}=\norm{\delta}_{2}^{2}/(1-\norm{\delta}_{2}^{2})\le C\nu$,
which is at most $\eta$ for $n$ large provided $\eta\to 0$ more slowly than
$C\sqrt{(\log k)/n}$.

For $\vartheta$ in the support of $\pi_1$, these choices lead to the variance
matrix
\begin{equation*}
  \Sigma_\vartheta=
 \begin{pmatrix}
   1 & 0 & \delta' \\
   0 & 1 & \delta' \\
   \delta & \delta & I_{k}
 \end{pmatrix},
\end{equation*}
while $\Sigma_{\vartheta^*}=I_{k+2}$ under the point mass $\pi_0$. It now
follows from eqs.~(118) and (119) in \textcite{javanmard_debiasing_2018} (which
are applications of Lemmas 2 and 3 in \textcite{cai_lasso_2017}) that
\begin{equation*}
  \chi^2(f_{\pi_0}, f_{\pi_1})\le
  e^{\frac{s^2}{k-s}}\left(1 + \frac{s}{k} (e^{4n\nu^2}-1) \right)^s - 1
  \leq
  \exp\left(\frac{s^{2}}{k-s}+\frac{s^{2}k^{c_{\nu}}}{k}\right)-1,
\end{equation*}
where the second inequality uses $e^{4n\nu^{2}}=k^{c_{\nu}}$ and $1+x\le e^{x}$.
Let $\psi=1/2-\tilde{\eta}\in(0,1/2)$. Since
$s\leq k^{1/2-\tilde{\eta}}/\sqrt{c_{\nu}\log(k)}\leq k^{\psi}$ once
$\log(k)\geq 1/c_{\nu}$, it follows that
$ \frac{s^{2}}{k-s}+\frac{s^{2}k^{c_{\nu}}}{k} \leq
\frac{k^{2\psi-1}}{1-k^{\psi-1}}+k^{2\psi+c_{\nu}-1}\to 0$, so that the
right-hand side of the preceding display converges to zero.

Finally, we note that under $\pi_{1}$,
$\abs{\beta_{1, n}}=\norm{\delta}_{2}^{2}/(1-\norm{\delta}_{2}^{2})\geq
\norm{\delta}_{2}^{2}\geq C\nu/4$, where the last inequality uses the fact that
since $\lfloor x\rfloor\ge x/2$ for $x\ge 1$, $s\nu^{2}\geq C\nu/4$. Hence,
$\abs{\beta_{1, n}}\geq C \sqrt{c_{\nu}}/8\cdot \sqrt{\log(k)/n}$, which yields the
result.

\section{Additional results}\label{sec:additional-results}

This \namecref{sec:additional-results} presents additional results that are
useful for implementing \Cref{algorithm:baseline}, and for assessing the
plausibility of the assumption $\Pen(\gamma)\le C$.
\Cref{sec:global_estimation_algebra} derives the properties of a regularized
estimator of the regression function $w_{i}\beta+z_{i}'\gamma$.
\Cref{sec:se_consistency} gives conditions under which this estimator can be
used to construct initial estimates of residuals in \Cref{algorithm:baseline}.
\Cref{sec:C_lower_CI} presents a lower \ac{CI} for $C$ that can be used to
assess the plausibility of the assumption $\Pen(\gamma)\le C$.

In this \namecref{sec:additional-results}, we focus primarily on the $\ell_{p}$
penalty $\Pen(\gamma)=\norm{\gamma_{2}}_{p}$, with $k_{2}\to \infty$ and
$k_{1}/n\to 0$. To state the results concisely, we use the notation
$\theta=(\beta, \gamma')'$ and let $\Theta=\mathbb{R}\times \Gamma$ denote its
parameter space. Let $X=(X_{1}, X_{2})$, where $X_{1}=(w, Z_{1})$, and
$X_{2}=Z_{2}$. We partition $\theta$ accordingly, with
$\theta_{1}=(\beta, \gamma_{1}')'$, and $\theta_{2}=\gamma_{2}$. Let
$H_{X_{1}}=X_{1}X_{1}^{+}$ and $M_{X_{1}}=I-H_{X_{1}}$ denote projections
onto the column space of $X_{1}$ and its orthogonal complement, where
$X_{1}^{+}$ denotes the pseudo-inverse (so that $X_{1}^{+}=(X_1'X_1)^{-1}X_1'$
if $X_{1}$ is full rank).

We allow the distribution $Q$ of $\varepsilon$ to be unknown and possibly
non-Gaussian, and only maintain the assumption that $\varepsilon_i$ is mean zero
and independent across $i$. The class of possible distributions for $Q$ is
denoted by $\mathcal{Q}_n$. We use $P_{\theta, Q}$ and $E_{\theta, Q}$ to denote
probability and expectation when $Y$ is drawn according to $Q\in \mathcal{Q}_n$
and $\theta\in\Theta$, and we use the notation $P_Q$ and $E_Q$ for expressions
that depend on $Q$ only and not on $\theta$.

We use the following assumption repeatedly throughout this appendix.

\begin{assumption}\label{global_estimation_assump}
  There exists $\eta>0$ such that, for all $i$ and $n$ and all
  $Q\in\mathcal{Q}_n$,
  \begin{equation*}
    P_{Q}(\abs{\varepsilon_i}>t)\le 2\exp(-\eta t^{2}) \quad\text{when}\quad
    p=1,\quad
    E_{Q}[\abs{\varepsilon_i}^{\max\{2+\eta, q\}}] < 1/\eta \quad\text{when}\quad
    p>1
  \end{equation*}
  and $\eta < E_Q\varepsilon_i^2$. In addition, the elements of $M_{X_{1}}X_2$
  are bounded by some constant $K_X$ uniformly over $n$.
\end{assumption}

\subsection{Estimating the regression function globally}\label{sec:global_estimation_algebra}

Consider the regularized regression estimator of $\theta$, given by
\begin{equation}\label{eq:YX_regularized_regression}
  \hat\theta = \argmin_{\vartheta} \norm{Y-X\vartheta}_2^2/n + \lambda \norm{\vartheta_2}_{p}.
\end{equation}
In order to derive the rate of convergence of $\hat{\theta}$ in
\Cref{l2_rate_thm} below, we first give an elementary property of this
estimator, following standard arguments \parencite[e.g.][Section
6.2]{buhlmann_statistics_2011}.

\begin{lemma}\label{theorem:elementary_property_lp}
  If $\norm{2X_2'M_{X_{1}}\varepsilon}_q/n\le \lambda_0$, then
  $\norm{M_{X_{1}}X_2(\hat\theta_2-\theta_2)}_2^2/n + (\lambda-\lambda_0)
  \norm{\hat\theta_2}_{p} \le (\lambda+\lambda_0) \norm{\theta_2}_{p}$.
\end{lemma}
\begin{proof}
  Write the objective function as
  \begin{equation*}
    \norm{H_{X_{1}}(Y-X_2\vartheta_2)-X_1\vartheta_1}_2^2/n
    + \norm{M_{X_{1}}Y-M_{X_{1}}X_2\vartheta_2}_2^2/n
    + \lambda \norm{\vartheta_2}_p.
  \end{equation*}
  The first summand can be set to zero for any $\vartheta_2$ by taking
  $\vartheta_1=X_{1}^{+}(Y-X_2\vartheta_2)$. Therefore,
  \begin{equation*}
    \hat\theta_2=\argmin_{\vartheta}
    \norm{M_{X_{1}}Y-M_{X_{1}}X_2\vartheta_2}_{2}^{2} / n
    + \lambda \norm{\vartheta_2}_{p},
  \end{equation*}
  with $\hat\theta_1=X_{1}^{+}(Y-X_2\hat\theta_{2})$. This implies
  $H_{X_{1}}\varepsilon=H_{X_{1}}Y-H_{X_{1}}X\theta
  =H_{X_{1}}X(\hat{\theta}-\theta)$, so that
  \begin{equation}\label{eq:theta2_error_orthogonal_decomposition}
    \norm{X(\hat\theta-\theta)}_{2}^{2}/n
    =\norm{H_{X_{1}}\varepsilon}_2^2/n
    +\norm{M_{X_{1}}X_2(\hat\theta_2-\theta_2)}_2^2/n,
  \end{equation}
  Using the fact that $\hat\theta_2$ attains a lower value of the objective than
  the true parameter value $\theta_2$, we obtain an $\ell_{p}$ version of what
  in the $\ell_{1}$ case \textcite[Lemma 6.1]{buhlmann_statistics_2011} term ``the
  Basic Inequality'',
  \begin{equation*}
    \norm{M_{X_{1}}X_2(\hat\theta_2-\theta_2)}_2^2/n + \lambda \|\hat\theta_2\|_p
    \le 2\varepsilon'M_{X_{1}}X_2(\hat\theta_2-\theta_2)/n + \lambda \|\theta_2\|_p.
  \end{equation*}
  By Hölder's inequality,
  $2\varepsilon'M_{X_{1}}X_2(\hat\theta_2-\theta_2)\le
  \|2X_2'M_{X_{1}}\varepsilon\|_q\|\hat\theta_2-\theta_2\|_p$ so that, on the
  event $\norm{2X_2'M_{X_{1}}\varepsilon}_q/n\le \lambda_0$, we have
  \begin{equation*}
      \|M_{X_{1}}X_2(\hat\theta_2-\theta_2)\|_2^2/n + \lambda \|\hat\theta_2\|_p
      \le \lambda_0\|\hat\theta_2-\theta_2\|_p + \lambda \|\theta_2\|_p
      \le \lambda_0\|\hat\theta_2\|_p + (\lambda+\lambda_0) \|\theta_2\|_{p},
  \end{equation*}
  which implies the result.
\end{proof}

We now use \Cref{theorem:elementary_property_lp} to derive rates of convergence
for the regularized regression estimator in \cref{eq:YX_regularized_regression}
for estimating the regression function in $\ell_{2}$ loss. For simplicity, we
use a fixed sequence for the penalty parameter $\lambda$ satisfying certain rate
conditions. This yields simple sufficient conditions that allow $\hat{\theta}$
to be used for auxiliary purposes such as standard error construction. In
practice, data-driven methods such as cross-validation may be appealing. We
discuss another possible choice based on moderate deviations bounds in
\Cref{lambda_alpha_remark} in \Cref{sec:C_lower_CI} below. We leave the analysis
of $\hat{\theta}$ under such choices of $\lambda$ for future research.

\begin{theorem}\label{l2_rate_thm}
  Suppose that \Cref{global_estimation_assump} holds.
  Let $\hat\theta$ be the penalized regression
  estimator defined in \cref{eq:YX_regularized_regression} with
  $\lambda=K_n r_{p}(k_{2})/\sqrt{n}$, where $K_n\to\infty$ and $r_{p}(k)$ given in
  \cref{eq:convergence-rates}. Then
  \begin{equation*}
    \sup_{\theta\in\mathbb{R}^{k+1}}
    \sup_{Q\in\mathcal{Q}_n} P_{\theta, Q}\left(\|X(\hat\theta-\theta)\|_2^2/n >
      K_n((k_1+1)/n+2\norm{\theta_2}_{p} r_{p}(k_{2})/\sqrt{n}) \right)\to 0,
  \end{equation*}
\end{theorem}
\begin{proof}
  By \Cref{lambda0_bound_lemma} below, if we set
  $\lambda_{0}=\lambda=K_n r_{p}(k_{2})/\sqrt{n}$, the condition of
  \Cref{theorem:elementary_property_lp}, and hence the conclusion that
  $\norm{M_{X_{1}}X_2(\hat\theta-\theta)}_2^2/n\le
  2K_n\norm{\theta_2}_{p}r_{p}(k_2)/\sqrt{n}$, holds with probability
  approaching one uniformly over $\theta\in\mathbb{R}^{k+1}$ and
  $Q\in\mathcal{Q}_n$. In addition, since $H_{X_{1}}$ is idempotent with rank at
  most $k_1+1$ and $E_Q\varepsilon\varepsilon'$ is diagonal with elements
  bounded uniformly over $Q\in\mathcal{Q}_n$, we have
  $E_Q\|H_{X_{1}}\varepsilon\|_2^2/n\le \tilde{K} (k_1+1)/n$ for some constant
  $\tilde K$. The result follows by Markov's inequality and
  \cref{eq:theta2_error_orthogonal_decomposition}.
\end{proof}

\begin{lemma}\label{lambda0_bound_lemma}
  Under \Cref{global_estimation_assump},
  $\inf_{Q\in\mathcal{Q}_n}P_Q(\norm{2X_2'M_{X_{1}}\varepsilon}_{q}/n\le K_n
  r_{p}(k_{2})/\sqrt{n}) \to 1$.
\end{lemma}
\begin{proof}
  Let $\tilde x_{ij}=(2M_{X_{1}}X_2)_{ij}$.  For $q<\infty$,
  we have
  \begin{equation*}
    E_{Q}\norm{2X_2'M_{X_{1}}\varepsilon}_{q}^{q}
    =E_Q \sum_{j=1}^{k_2}\abs*{\sum_{i=1}^n \tilde x_{ij}\varepsilon_i}^{q}
    \le k_2\cdot K\cdot n^{q/2}
  \end{equation*}
  for some constant $K$ that depends only on $\eta$, $q$ and $K_X$, by
  \textcite{von_bahr_convergence_1965}. The result then follows by Markov's
  inequality. For $q=\infty$, we have
  \begin{equation*}
    P_Q\left(\norm{2X_2'M_{X_{1}}\varepsilon}_q/n> K_n\sqrt{\log k_2}/\sqrt{n} \right)
    =P_Q\left(\max_{j}\left| \sum_{i=1}^n \tilde x_{ij} \varepsilon_i \right|/n>
    K_n\sqrt{\log k_2}/\sqrt{n} \right),
  \end{equation*}
  which, for some $\tilde K>0$, is bounded by
  $2k_2\exp(-\tilde K\cdot K_n^2 \log k_2)=2k_2^{1-\tilde K\cdot K_n^2}\to 0$ by
  Hoeffding's inequality for sub-Gaussian random variables \parencite[][Theorem
  2.7.3]{vershynin26}.
\end{proof}

\subsection{Feasible \texorpdfstring{\acp{CI}}{CIs} with an unknown error distribution}\label{sec:se_consistency}

This \namecref{sec:se_consistency} presents formal results for feasible \acp{CI}
when the error distribution is unknown. \Cref{sec:se_general} presents general
results for feasible \acp{CI} for linear estimators in our setting.
\Cref{sec:se_optimized_weights} specializes these results to the feasible
\acp{CI} in \Cref{sec:impl-with-non}, with some technical modifications.

\subsubsection{General results}\label{sec:se_general}

We consider standard errors for linear estimators $\hat\beta_{a}=a'Y$, deviating
slightly from the notation in the main text by making the dependence on the
weights explicit with the subscript $a$. As in the main text, the weights $a$
are nonrandom: they can depend on $X$ but not on $Y$. We consider asymptotics
where the weights $a$ are allowed to depend on $n$ so that
$a_{1}, \dotsc, a_{n}$ is a triangular array rather than a sequence, but we
leave this implicit in the notation.

Let $\hat\theta$ be an estimate of $\theta$, and let
$\hat\varepsilon=Y-X\hat\theta$. Consider the estimator
$\hat{V}_{a}=\sum_{i=1}^n a_i^2\hat\varepsilon_i^2$ of
$V_{Q}=\var_{Q}(\hat{\beta}_{a})=\sum_{i=1}^{n} a_i^2E_{Q}\varepsilon_i^2$. We
consider coverage of the feasible bias-aware \ac{CI}
\begin{equation}\label{eq:feasible_FLCI_unknown_error}
  \hat\beta_{a}\pm \cv_\alpha(\maxbias_{\Gamma}(\hat\beta_a)/\hat{V}_{a}^{1/2})\cdot
  \hat{V}_{a}^{1/2},
\end{equation}
where $\maxbias_{\Gamma}(\hat\beta_{a})$ is the worst-case bias, given in
\cref{eq:maxbias_def}. Under non-Gaussian errors, valid coverage will require
the maximal Lindeberg weight defined in \cref{eq:lind_weight} to converge to
zero to ensure asymptotic normality of the estimator. The maximal
Lindeberg weight turns out to also be relevant for controlling the contribution
of estimation error in $\hat\theta$ in the variance estimate $\hat{V}_{a}$. In
particular, in the following theorem, there is a tradeoff between the rate at
which $\operatorname{Lind}(a)\to 0$ and the $\ell_2$ rate of convergence of the
estimator $X\hat\theta$ of the regression function.

\begin{theorem}\label{general_se_thm}
  Suppose that, for some $\eta>0$, $\eta\le E_Q\varepsilon_i^2$ and
  $E_Q\abs{\varepsilon_i}^{2+\eta} \le 1/\eta$ for all $i$ and all
  $Q\in\mathcal{Q}_n$.  Suppose also that,
  for some sequence $c_n$ with $c_n=\mathcal{O}(\sqrt{n})$, we have
  \begin{enumerate}[label=({\roman*})]
  \item $\max\left\{\sqrt{n}c_n,1\right\}\cdot \operatorname{Lind}(a) \to
    0$; and
  \item $\inf_{\theta\in\Theta, Q\in\mathcal{Q}_n} P_{\theta, Q}
    (\norm{X(\hat\theta-\theta)}_2 \le c_n)\to 1$.
  \end{enumerate}
  Then, for any $\delta>0$,
  $\inf_{\theta\in\Theta, Q\in\mathcal{Q}_n} P_Q\left(\abs{(\hat{V}_{a}-V_Q)/V_Q} <
    \delta \right) \to 1$. Furthermore,
  \begin{equation}\label{eq:feasible_flci_coverage}
    \liminf_n\inf_{\theta\in\Theta, Q\in\mathcal{Q}_n} P_Q\left(
      \beta\in \left\{\hat\beta_{a}\pm
        \cv_\alpha(\maxbias_{\Gamma}(\hat\beta_{a})/\sqrt{\hat{V}_{a}})\cdot
        \sqrt{\hat{V}_{a}} \right\} \right) \ge 1-\alpha.
  \end{equation}
\end{theorem}
\begin{proof}
  We have
  \begin{equation}\label{eq:Vhat_error_decomposition}
    \frac{\hat{V}_{a} - V_Q}{V_Q} = \frac{\sum_{i=1}^n a_i^2(\hat\varepsilon_i^2-\varepsilon_i^2)}{V_Q}
   + \frac{\sum_{i=1}^n a_i^2(\varepsilon_i^2-E_Q\varepsilon_i^2)}{V_Q}.
  \end{equation}
  Let $\tilde{b}_i=a_i^2/\sum_{j=1}^{n} a_j^2$ so that
  $\max_{1\le i\le n}\tilde{b}_i=\operatorname{Lind}(a)$. The second term in
  \cref{eq:Vhat_error_decomposition} is bounded by
  $|\sum_{i=1}^n \tilde b_i (\varepsilon_i^2-E_Q\varepsilon_i^2)|/\eta$. The
  absolute $1+\eta/2$ moment of this quantity is bounded by a constant times
  $\sum_{i=1}^n \tilde b_i^{1+\eta/2}\cdot 1/\eta^{1+\eta/2}$ by
  \textcite{von_bahr_inequalities_1965}. This is bounded by
  $\max_{1\le i\le n} \tilde b_i^{\eta/2}\cdot \sum_{i=1}^{n}
  \tilde{b}_i/\eta^{1+\eta/2} = \max_{1\le i\le n}
  \tilde{b}_i^{\eta/2}/\eta^{1+\eta/2}
  =\operatorname{Lind}(a)^{\eta/2}/\eta^{1+\eta/2}\to 0$. The first term in
  \cref{eq:Vhat_error_decomposition} is bounded by $\operatorname{Lind}(a)/\eta$
  times
  \begin{equation*}
    \sum_{i=1}^n \abs{\hat{\varepsilon}_{i}^{2}-\varepsilon_i^2}
    = \sum_{i=1}^n\abs{\hat{\varepsilon}_i+ \varepsilon_{i}}\cdot
    \abs{\hat\varepsilon_i-\varepsilon_i}
    \le \norm{\hat\varepsilon+\varepsilon}_2\norm{\hat\varepsilon-\varepsilon}_2
    \le \left(\|\hat\varepsilon-\varepsilon\|_2 + 2\|\varepsilon\|_2 \right)
    \norm{\hat\varepsilon-\varepsilon}_2.
  \end{equation*}
  For some constant $K$ that depends only on $\eta$, we have
  $2\|\varepsilon\|_2\le K\sqrt{n}$ with probability approaching one uniformly
  over $Q\in\mathcal{Q}_n$. Since
  $\|\hat\varepsilon-\varepsilon\|_2=\|X(\hat\theta-\theta)\|_2\le c_n$ it
  follows that, with probability approaching one uniformly over
  $\theta\in\Theta$ and $Q\in\mathcal{Q}_n$, the first term in
  \cref{eq:Vhat_error_decomposition} is bounded by
  $\operatorname{Lind}(a)\cdot (K\sqrt{n}+c_n)\cdot c_n\to 0$. It follows that
  for any $\delta>0$,
  $\inf_{\theta\in\Theta, Q\in\mathcal{Q}_n} P_Q\left( \left| (\hat{V}_{a}-V_Q)/V_Q
    \right| < \delta \right) \to 1$. Coverage of the \ac{CI} then follows from
  Theorem F.1 in \textcite{ArKo18optimal}, with the central limit theorem condition
  following by using the weights and moment bounds to verify the Lindeberg
  condition (see Lemma F.1 in \textcite{ArKo18optimal}).
\end{proof}

If $\hat{\theta}$ is based on the regularized regression in
\cref{eq:YX_regularized_regression}, it follows from \Cref{l2_rate_thm} that
(ii) in \Cref{general_se_thm} will hold with
$c_n=\sqrt{K_n(k_1+1+C n^{1/2}r_{p}(k_2))}$ for a slowly increasing constant
$K_n$. Hence, condition (i) in \Cref{general_se_thm} will then hold so long as
\begin{equation}\label{eq:lind_weight_constraint}
  \sqrt{K_{n}n \cdot (k_1+1+C n^{1/2}r_{p}(k_2))}\cdot  \operatorname{Lind}(a)\to 0.
\end{equation}
The following \namecref{general_se_corollary} records this result.

\begin{corollary}\label{general_se_corollary}
  Suppose that \Cref{global_estimation_assump} holds. Let $\hat\varepsilon$ be
  the residuals from the regularized regression in
  \cref{eq:YX_regularized_regression}, with $\lambda$ given in
  \Cref{l2_rate_thm} for some $K_n\to \infty$, and suppose
  \cref{eq:lind_weight_constraint} holds. Then \cref{eq:feasible_flci_coverage}
  holds with
  $\Theta=\mathbb{R}^{k_1+1}\times \{\gamma_2\colon \norm{\gamma_2}_p\le C\}$.
\end{corollary}

\subsubsection{Optimized weights}\label{sec:se_optimized_weights}

We now apply the results in \Cref{sec:se_general} to the feasible CIs based on
optimized weights in \Cref{algorithm:baseline}. For technical reasons, we assume
that the initial variance estimator in
Step~\ref{initial_variance_estimator_step} of the algorithm is replaced by a
nonrandom guess $\tilde\sigma^2$.\footnote{Alternatively, one could use
  sample-splitting or cross-fitting. In our Monte Carlo simulations, we find
  that the feasible \acp{CI} have good coverage without imposing these technical
  modifications.}

It follows directly from \Cref{general_se_corollary} that the feasible \ac{CI}
will be asymptotically valid provided that the maximal Lindeberg weight shrinks
quickly enough relative to error in the estimator used to form the residuals, so
that $\operatorname{Lind}(a_{\lambda^{*}_{\textnormal{CI}}})$ satisfies
\cref{eq:lind_weight_constraint}. This is equivalent to checking that
$\operatorname{Lind}(a_{\lambda^{*}_{\textnormal{CI}}})\leq b_{n}$ for some
sequence of constants $b_{n}$ such that
$\sqrt{K_{n} n (k_1+1+C n^{1/2} r_{p}(k_2))} \cdot b_n\to 0$. Since the
Lindeberg weight is observable, this condition can be checked directly in a
given application. Since the theory only gives a rate condition on $b_{n}$,
operationalizing the check requires one to take a stance on the constants. In
our empirical work, we follow the proposal in \textcite[Section~6.2]{NoRo24} and
use $b_{n}=0.075$ as a simple benchmark for the maximal Lindeberg weight,
leaving a more principled implementation to future work.

If the optimal weights violate the constraint
$\operatorname{Lind}(a_{\lambda^{*}_{\textnormal{CI}}})\leq b_{n}$, then the
feasible \ac{CI} based on the optimized weights may asymptotically undercover.
To prevent this scenario, a principled approach is to constrain the
bias-variance tradeoff in \cref{eq:a_optimization}, and only optimize it over a
set of weights $a$ that satisfy $\operatorname{Lind}(a)\leq b_{n}$. Constraining
the optimization problem to guarantee that the maximal Lindeberg weight of the
solution is not too large is analogous to proposals by \textcite{NoRo24}, and
\textcite{javanmard_confidence_2014} in other settings. For our setting, given
an initial guess $\tilde{\sigma}^{2}$ for the variance, the optimal
bias-variance tradeoff is then given by minimizing
\begin{equation*}
  2 \tilde\sigma\norm{a}_{2} \cdot
  \cv_\alpha(\maxbias_\Gamma(\hat\beta_{a})/(\tilde\sigma \norm{a}_{2}))
\end{equation*}
over weights $a$ that satisfy $\operatorname{Lind}(a)\le b_{n}$. We denote the
solution to this minimization problem by $a^{*}_{b_{n}}$. Imposing the Lindeberg
weight condition can in general affect the performance of the \ac{CI}. We show
in \Cref{thm:feasible_ci_upper_bound} below, however, that the optimal rate of
convergence derived in \Cref{thm:upper_rate_bound} still obtains for a range of
choices of $b_{n}$.

As a robustness check in our empirical work, we use a simpler approach that can
be viewed as approximating the optimal weights $a^{*}_{b_{n}}$. In particular,
if the constraint
$\operatorname{Lind}(a_{\lambda^{*}_{\textnormal{CI}}})\leq b_{n}$ is violated,
we sequentially delete the observation with the largest Lindeberg weight until
the constraint holds. This has the advantage that it still allows us to use
existing computational algorithms discussed in~\Cref{remark:computation}.
Furthermore, as we discuss in \Cref{remark:empirical-lindeberg} below, in
practice the number of deletions is very small, suggesting that the efficiency
cost relative to using the weights $a^{*}_{b_{n}}$ is minimal. This approach can
also be viewed as a more principled version of the common empirical practice of
running robustness checks using a sample that drops extreme values of $X$; we
tie the definition of ``extreme'' directly to the Lindeberg weight, rather than
using an ad-hoc definition.

\begin{remark}[Lindeberg weight diagnostics in the empirical applications]\label{remark:empirical-lindeberg}
  In the lottery application, the full sample values for
  $\operatorname{Lind}(a)$ are 0.149 for the short regression, 0.114 for AKK1,
  0.349 for AKK2, and 0.376 for the long regression. For each method, we
  sequentially remove the observation with the largest individual Lindeberg
  weight until $\operatorname{Lind}(a)$ lies below the $b_{n}=0.075$ benchmark.
  After each deletion, we re-estimate using the same procedure, including
  reoptimizing $\lambda$ for AKK1 and AKK2. The short regression \ac{CI} changes
  from $(-0.074,-0.031)$ to $(-0.086,-0.033)$ after three deletions; the AKK1
  \ac{CI} changes from $(-0.098,-0.022)$ to $(-0.108,-0.029)$ after one
  deletion; and the AKK2 \ac{CI} changes from $(-0.110,0.010)$ to
  $(-0.156,0.035)$ after six deletions. While all \acp{CI} widen somewhat, the
  qualitative conclusions are largely unchanged.

  By contrast, when applied to the long regression, the procedure yields, after
  three deletions, a \ac{CI} that is too wide to be informative.
  Correspondingly, the \textcite{li_linear_2020} procedure, which requires
  asymptotically normal short and long regression estimators, also yields
  uninformative inference. The fact that both the long regression and the
  \textcite{li_linear_2020} procedure are sensitive to outliers provides an
  additional motivation for using our procedure.

  In the RHC application, the corresponding values for the short regression,
  long regression, AKK1, and AKK2 are 0.0013, 0.0024, 0.0023, and 0.0024,
  respectively, all well below the benchmark $b_{n}=0.075$.
\end{remark}

We now show that the optimal rate of convergence derived in
\Cref{thm:upper_rate_bound} is unaffected by imposing a Lindeberg weight
constraint in the optimization.

\begin{theorem}\label{thm:feasible_ci_upper_bound}
  Suppose that \Cref{global_estimation_assump} holds. Let $\hat\varepsilon$ be
  the residuals from the regularized regression in
  \cref{eq:YX_regularized_regression}, with $\lambda$ given in
  \Cref{l2_rate_thm} for some $K_n\to \infty$.
  Let $b_n$ be a sequence satisfying $b_{n}=o(1/\sqrt{K_{n} n (1+C n^{1/2} r_{p}(k_2))})$.
  Suppose $k_{1}=0$, so that
  $\Gamma=\{\gamma\colon \norm{\gamma}_{p}\le C\}$. Suppose also that, for
  some $\eta>0$, the design matrix $X$ is in the set $\mathcal{E}^{p}_n(\eta)$
  defined in \Cref{sec:asymptotics} for large enough $n$, and that
  $\max_{1\le i\le n} (w_i-z_i'\delta)^2/n=o(b_n)$ and
  $\frac{1}{n}\sum_{i=1}^n(w_i-z_i'\delta)^2$ is bounded away from zero where
  $\delta$ is given in the definition of $\mathcal{E}^{p}_n(\eta)$ in
  \Cref{sec:asymptotics}. Then there exists a constant $K$ such that
  the CI in \cref{eq:feasible_FLCI_unknown_error} with $a=a^{*}_{b_{n}}$
  satisfies
  \begin{equation*}
    \lim_{n\to\infty}\sup_{\theta\in\Theta, Q\in\mathcal{Q}_n} P_Q\left(
    2\cv_\alpha(\maxbias_{\Gamma}(\hat\beta_a)/\hat{V}_{a}^{1/2})\cdot
    \hat{V}_{a}^{1/2}
    \ge Kn^{-1/2} (1 + C r_{p}(k))\right) = 0.
  \end{equation*}
\end{theorem}
\begin{proof}
  Let $v_i=w_i-z_i'\delta$ and $\tilde a_i=v_i/\sum_{j=1}^n v_{j}w_j$ where
  $\delta$ is given in the definition of $\mathcal{E}^{p}_n(\eta)$. Note that
  $\operatorname{Lind}(\tilde a)=\max_{1\le i\le n}
  (w_i-z_i'\delta)^2/\sum_{j=1}^n (w_j-z_j'\delta)^2$. Under the assumptions of
  the theorem, this is bounded by a constant times
  $\max_{1\le i\le n}(w_i-z_i'\delta)^2/n=o(b_n)$. Thus, the weights $\tilde a$
  are feasible for the constrained optimization problem that defines
  $a^*_{b_n}$. It follows that
  \begin{equation*}
    2 \tilde\sigma \|a^*_{b_n}\|_2 \cdot \cv_\alpha(\maxbias_\Gamma(\hat\beta_{a^*_{b_n}})/(\tilde\sigma
    \|a^*_{b_n}\|_2))
    \le     2 \tilde\sigma \|\tilde a\|_2 \cdot \cv_\alpha(\maxbias_\Gamma(\hat\beta_{\tilde a})/(\tilde\sigma
    \|\tilde a\|_2)).
  \end{equation*}
  It follows from the proof of \Cref{thm:upper_rate_bound} that the right-hand
  side of the above display is bounded by a constant times
  $n^{-1/2}(1+C r_{p}(k))$. Furthermore, by the uniform consistency of $\hat V$
  (which follows from \Cref{general_se_thm}, where we note that the conditions
  of \Cref{general_se_thm} follow by the discussion leading up to \Cref{general_se_corollary}) and the fact that the variance of
  $\hat\beta_{a^*_{b_n}}$ is bounded from above uniformly over $\mathcal{Q}_n$,
  the width
  $2\cdot \cv_\alpha(\maxbias_{\Gamma}(\hat\beta_a)/\sqrt{\hat V})\cdot
  \sqrt{\hat V}$ is bounded by a constant times the left-hand side of the above
  display with probability approaching one uniformly over $\theta\in\Theta$ and
  $Q\in\mathcal{Q}_n$. The result follows.
\end{proof}

\Cref{thm:feasible_ci_upper_bound} imposes both an upper bound on $b_n$ (it must
decrease at least as quickly as $1/\sqrt{K_{n} n (1+C n^{1/2} r_{p}(k_2))}$)
as well as a lower bound on $b_n$ (it must decrease more slowly than $\max_{1\le i\le n}(w_i-z_i'\delta)^2/n$).
To
interpret the latter condition, note that $w_i-z_i'\delta$ plays the role of the
residual in a best linear predictor regression of $w_i$ on $z_i$ in a random
design setting. Thus, the condition
$\max_{1\le i\le n}(w_i-z_i'\delta)^2/n=o(b_n)$ is a tail condition on this best
linear predictor error.

Depending on how quickly $\max_{1\le i\le n}(w_i-z_i'\delta)^2$ increases, there
will be a range of choices of $b_n$ that satisfy both conditions. For example,
if $\max_{1\le i\le n}(w_i-z_i'\delta)^2$ is bounded, then the lower bound on
$b_n$ will hold if $b_n$ is equal to $1/n$ times a slowly increasing sequence.
Suppose, for simplicity, that we take $b_n=K_n/n$ where $K_n$ is the same
sequence in the upper bound. Then the upper bound
$b_n=o(1/\sqrt{K_{n} n (1+C n^{1/2} r_{p}(k_2))})$ will then hold so long as
\begin{align*}
  \sqrt{K_{n} n (1+C n^{1/2} r_{p}(k_2))}\cdot (K_n/n)
  =K_n\cdot \sqrt{K_{n}( 1/n+C n^{-1/2}r_{p}(k_2))}\to 0.
\end{align*}
Since $C n^{-1/2}r_{p}(k_2)$ is the rate at which the optimal \ac{CI} shrinks,
this condition is essentially the same as requiring that the optimal \ac{CI}
shrinks towards zero as $n\to\infty$.

\subsection{Lower CIs for \texorpdfstring{$C$}{C}}\label{sec:C_lower_CI}
We present a lower \ac{CI} for the penalty bound $C$, which can be used to
assess the plausibility of the assumption $\Pen(\gamma_2)\le C$. Let
$\hat\theta_2(\lambda)$ denote the regularized regression estimator of
$\gamma_2$, given in \cref{eq:YX_regularized_regression}, with penalty
$\lambda$. Let $\lambda^*_\alpha$ denote an upper bound for the $1-\alpha$
quantile of $\|2X_2'M_{X_{1}}\varepsilon\|_q/n$. Let
\begin{equation}\label{eq:C_lower_CI}
  \hat{\underline C}=\sup_{\lambda>\lambda^*_\alpha} \frac{\lambda-\lambda^*_\alpha}{\lambda+\lambda^*_\alpha}\|\hat\theta_2(\lambda)\|_{p}.
\end{equation}
In the idealized finite sample setting with $\varepsilon\sim \ND(0,\sigma^2 I_n)$
with $\sigma^2$ known, $\lambda^*_\alpha$ can be computed exactly, so that
$\hat{\underline C}$ is feasible.

\begin{theorem}
  Consider $\hat{\underline C}$ in \cref{eq:C_lower_CI} with $\lambda^*_\alpha$
  given by the $1-\alpha$ quantile of $\|2X_2'M_{X_{1}}\varepsilon\|_q/n$. Then,
  for any $\beta, \gamma_1,\gamma_2$ with $\|\gamma_2\|_p\le C$, we have
  $P_{\beta, \gamma_1,\gamma_2}(C\in \hor{\hat{\underline C}, \infty})\ge
  1-\alpha$.
\end{theorem}
\begin{proof}
  It follows from \Cref{theorem:elementary_property_lp} that, on the event
  $\|2X_2'M_{X_{1}}\varepsilon\|_q/n\le \lambda^*_\alpha$ (which holds with probability
  at least $1-\alpha$ by assumption), we have
  $\frac{\lambda-\lambda^*_\alpha}{\lambda+\lambda^*_\alpha}\|\hat\theta_2(\lambda)\|_p\le
  \|\gamma_2\|_p\le C$
  for all $\lambda>\lambda^*_\alpha$.  Thus, the supremum of this quantity over
  $\lambda$ in this set is also no greater than $C$ on this event.
\end{proof}

We now present a feasible version of this \ac{CI} when the error distribution is
unknown and possibly heteroskedastic in the case where $p=1$. Let
$\tilde x_{ij}=(M_{X_{1}}X_2)_{ij}$. Since $q=\infty$ in this case, we need to
choose $\hat\lambda^*_\alpha$ such that
\begin{equation*}
  2\norm{X_2'M_{X_{1}}\varepsilon}_{\infty} / n = \max_{1\le j\le k_2}
  \abs*{\sum_{i=1}^n 2\tilde x_{ij} \varepsilon_i/n }\le \hat\lambda^*_\alpha
\end{equation*}
with probability at least $1-\alpha$ asymptotically. Let
$\hat{V}_j=\sum_{i=1}^n (2\tilde x_{ij}/n)^2\hat{\varepsilon}_i^2$, where
$\hat\varepsilon_i$ is the residual from an initial regularized regression with
$\lambda$ chosen as in \Cref{l2_rate_thm} for some slowly increasing $K_n$. This
leads to the moderate deviations critical value $\hat\lambda^*_\alpha$, which
sets
\begin{equation}\label{eq:lambda_alpha_unknown_error_distribution}
  \alpha = \sum_{j=1}^{k_2} 2\Phi(-\hat\lambda^*_\alpha/\hat{V}_j^{1/2}).
\end{equation}

\begin{remark}\label{lambda_alpha_remark}
  The analysis in \Cref{l2_rate_thm} of the regularized regression estimator in
  \cref{eq:YX_regularized_regression} relies on choosing a penalty parameter
  greater than $2\|X_2'M_{X_{1}}\varepsilon\|_\infty/n$ with high probability,
  which is precisely the goal of the critical value $\hat\lambda^*_\alpha$ given
  in \cref{eq:lambda_alpha_unknown_error_distribution}. This suggests an
  iterative procedure in which one uses $\hat\lambda^*_\alpha$ (perhaps with
  some sequence $\alpha_n$ converging slowly to zero) as a data-driven penalty
  parameter in the regression in \cref{eq:YX_regularized_regression} after using
  some initial penalty choice satisfying the conditions of \Cref{l2_rate_thm} to
  form the residuals used to compute $\hat\lambda^*_\alpha$.
\end{remark}

The penalty choice $\hat\lambda^*_\alpha$ is related to data-driven choices of
the lasso penalty in the case with an unknown error distribution.
\textcite{belloni_sparse_2012} use similar ideas to choose the penalty parameter
in this setting under $\ell_0$ constraints, although our implementation is
somewhat different, since our parameter space constrains the penalty loadings we
place on each parameter. While $\hat\lambda^*_\alpha$ does not take into account
correlations between the moments, one could take into account these correlations
using a bootstrap implementation, as suggested by
\textcite{chernozhukov_gaussian_2013}.
\begin{theorem}
  Suppose \Cref{global_estimation_assump} holds with $p=1$ and that
  $\frac{1}{n}\sum_{i=1}^n \tilde x_{ij}^2\ge \eta$ for $j=1,\dotsc, k_{2}$ for
  all $n$, where $\tilde x_{ij}=(M_{X_{1}}X_2)_{ij}$. Let $\hat\lambda^*_\alpha$
  be given in \cref{eq:lambda_alpha_unknown_error_distribution} with $\hat V_j$
  formed using residuals $\hat\varepsilon$ from the regularized regression in
  \cref{eq:YX_regularized_regression} with penalty $\lambda$ chosen as in
  \Cref{l2_rate_thm} for some $K_n\to\infty$ with
  $K_n(k_1/n+(C+1)\sqrt{\log k_2}/\sqrt{n})\cdot (\log k_2)^2\to 0$. Then
  \begin{equation*}
    \limsup_n\sup_{\beta, \gamma\colon \norm{\gamma_2}_{1}\le
      C}\sup_{Q\in\mathcal{Q}_n}P_{\theta, Q}\left(\max_{1\le j\le k_2}
      \abs*{\sum_{i=1}^n 2\tilde x_{ij} \varepsilon_i/n}> \hat\lambda^*_\alpha
    \right)\le \alpha.
  \end{equation*}
  In particular, define $\hat{\underline C}$ as in \cref{eq:C_lower_CI}, but
  with $\lambda^*_\alpha$ replaced with $\hat\lambda^*_\alpha$. We then have
  \begin{equation*}
  \liminf_n\inf_{\beta, \gamma\colon \norm{\gamma_2}_1\le
    C}\inf_{Q\in\mathcal{Q}_n}P_{\theta, Q}\left(C\in \hor{\hat{\underline C}, \infty}
  \right)\ge 1-\alpha.
  \end{equation*}
\end{theorem}
\begin{proof}
  Let $\tilde V_j=\sum_{i=1}^n (2\tilde x_{ij}/n)^2\varepsilon_{i}^{2}$ and let
  $V_{Q, j}=\sum_{i=1}^n (2\tilde x_{ij}/n)^2E_Q\varepsilon_i^2$. Note that
\begin{multline*}
  \abs{\hat{V}_j-\tilde{V}_j} = \left| \sum_{i=1}^n
    (2\tilde{x}_{ij}/n)^2(\hat\varepsilon_i^2-\varepsilon_i^2) \right|
  =\left| \sum_{i=1}^n (2\tilde x_{ij}/n)^2(\hat\varepsilon_i+\varepsilon_i)(\hat\varepsilon_i-\varepsilon_i) \right| \\
  \le (2K_X/n)^2
  \|\hat\varepsilon+\varepsilon\|_2\|\hat\varepsilon-\varepsilon\|_2 \le
  (2K_X/n)^2
  (2\|\varepsilon\|_2+\|\hat\varepsilon-\varepsilon\|_2)\|\hat\varepsilon-\varepsilon\|_2.
\end{multline*}
On the event that $2\norm{\varepsilon}_2\le \sqrt{n}\tilde K$ and
\begin{equation*}
\norm{\hat\varepsilon-\varepsilon}_2=\norm{X(\hat{\theta}-\theta)}_2\le \sqrt{n
  K_n}\cdot (k_1/n+2C \sqrt{\log k_{2}}/\sqrt{n})^{1/2},
\end{equation*}
which holds with probability approaching one uniformly over $Q\in\mathcal{Q}_n$
when $\tilde K$ is large enough (by Chebyshev's inequality), this is bounded by
$(2K_X/n)^2(\tilde K \sqrt{n} + \sqrt{n K_n}(k_1/n+2C \sqrt{\log
  k_2}/\sqrt{n})^{1/2})\cdot \sqrt{n K_n}(k_1/n+2C \sqrt{\log
  k_2}/\sqrt{n})^{1/2}$. Since $V_{Q, j}\ge \tilde{\eta}/n$ uniformly over $j$
and over $n$ for some $\tilde{\eta}>0$, this implies that, on this event,
$\max_{1\le j\le k_2} \left| \hat V_j-\tilde V_j \right|/V_{Q, j}$ is bounded by
\begin{equation*}
4\tilde{\eta}^{-1}(K_X^2/n)(\tilde K \sqrt{n} + \sqrt{n K_n}(k_1/n+2C
\sqrt{\log k_2}/\sqrt{n})^{1/2})\cdot \sqrt{n K_n}(k_1/n+2C
\sqrt{\log k_2}/\sqrt{n})^{1/2},
\end{equation*}
which in turn is bounded by a constant times
$K_n^{1/2}(k_1/n+2C\sqrt{\log k_2}/\sqrt{n})^{1/2}$ so long as this quantity converges to zero.

In addition, note that
$(\tilde V_j-V_{Q, j})/V_{Q, j}=\sum_{i=1}^n \tilde a_{ij}
(\varepsilon_{i}^{2}-E_Q\varepsilon_{i}^{2})/n$, where
$\tilde a_{ij}=\frac{4\tilde x_{ij}^2}{nV_{Q, j}}\le 4K_X^2\tilde{\eta}^{-1}$ and
$\tilde{\eta}$ is a lower bound for $nV_{Q, j}$. Using this bound on
$\tilde a_{ij}$ and the tail bound on $\varepsilon_{i}$, it follows from
Bernstein's inequality for sub-exponential random variables that, for
$\delta<1$, $P_Q(|\tilde V_j-V_{Q, j}|/V_{Q, j}\ge \delta)$ is bounded from
above by $2\exp(- c n\delta^2)$ for some constant $c$ that depends only on
$K_X$, $\tilde{\eta}$ and $\eta$. Thus, for any sequence $\delta_n$, we have
$P_Q(\max_{1\le j\le k_2}|\tilde V_j-V_{Q, j}|/V_{Q, j}\ge \delta_{n})\le
2k_2\exp(-c n\delta_n^2)$, which converges to zero so long as $\delta_n$ is
bounded from below by a large enough constant times $\sqrt{\log k_2}/\sqrt{n}$.

This gives the rate of convergence for $\hat V_j/V_{Q, j}$ to one which, by
continuous differentiability of $t\mapsto\sqrt{t}$ at $t=1$, gives the same
rates for $\sqrt{\hat V_j}/\sqrt{V_{Q, j}}$. In particular, letting $c_n$ be
given by a large enough constant times
$K_n^{1/2}(k_1/n+(C+1)\sqrt{\log k_2}/\sqrt{n})^{1/2}$, the event
$\max_{1\le j\le k_2}\left| \sqrt{\hat V_j}/\sqrt{V_{Q, j}}-1 \right|\le c_n$
holds with probability approaching one uniformly over $Q\in\mathcal{Q}_n$ and
$\beta, \gamma$ with $\|\gamma_2\|\le C$. On this event, we have
\begin{equation*}
  \alpha=\sum_{j=1}^{k_2} 2\Phi(-\hat\lambda^*_\alpha/\sqrt{\hat V_{j}})
  \ge \sum_{j=1}^{k_2} 2\Phi(-\hat\lambda^*_\alpha/(\sqrt{V_{Q_n, j}}(1-c_n))).
\end{equation*}
Thus,
letting $\lambda_{\alpha, n}$ solve $\alpha=\sum_{j=1}^{k_2}
2\Phi(-\lambda_{\alpha, n}/(\sqrt{V_{Q_n, j}}))$,
we have
$\hat\lambda^*_\alpha/(1-c_n)=\lambda_{\tilde\alpha, n}$ for some
$\tilde\alpha\le \alpha$, so that $\hat\lambda^*_\alpha/(1-c_n)\ge
\lambda_{\alpha, n}$.  It follows that the non-coverage probability
under any sequence of parameters with $\|\gamma_2\|_p\le C$ and any sequence $Q_n\in\mathcal{Q}_n$
is bounded by a term that converges to zero plus
\begin{multline*}
  P_{Q_n}\left(\max_{1\le j\le k_2}\left| \frac{1}{n}\sum_{i=1}^n2\tilde{x}_{ij}
      \varepsilon_i \right| > (1-c_n)\lambda_{\alpha, n} \right)
  \le \sum_{j=1}^{k_2} F_{n, j}((1-c_n)\lambda_{\alpha, n}/\sqrt{V_{Q_n, j}}) \\
  = \sum_{j=1}^{k_2} 2\Phi(-\lambda_{\alpha, n}/\sqrt{V_{Q_n, j}})\cdot
  A_{n, j}\cdot B_{n, j},
\end{multline*}
where
$F_{n, j}(t)=P_{Q_n}\left(\abs{n^{-1}\sum_{i=1}^n 2 \tilde{x}_{ij}
    \varepsilon_i/\sqrt{V_{Q_n, j}}} > t \right)$,
$A_{n, j}=\frac{\Phi(-(1-c_n)\lambda_{\alpha, n}/\sqrt{V_{Q_n, j}})}{\Phi(-\lambda_{\alpha, n}/
  \sqrt{V_{Q_n, j}})}$ and
$B_{n, j}=\frac{F_{n, j}((1-c_n)\lambda_{\alpha, n}/\sqrt{V_{Q_n, j}})}{2\Phi(-(1-c_n)
  \lambda_{\alpha, n}/\sqrt{V_{Q_n, j}})}$. Since
$\sum_{j=1}^{k_2} 2\Phi(-\lambda_{\alpha, n}/\sqrt{V_{Q_n, j}})=\alpha$ by
definition, it suffices to show that
$\limsup_{n\to\infty}\max_{1\le j\le {k_2}} \max\{A_{n, j}, B_{n, j}\}\le 1$.

For $A_{n, j}$, we use the bound
$\Phi(-s)/\Phi(-t)\le
[s^{-1}/(t^{-1}-t^{-3})]\exp((t^2-s^2)/2)$ (this follows from the bound $(t^{-1}-t^{-3})\exp(-t^2/2)/\sqrt{2\pi}\le \Phi(-t)\le
t^{-1}\exp(-t^2/2)/\sqrt{2\pi}$ given in Lemma 2, Section 7.1 in
\textcite{feller_introduction_1968}), which gives
\begin{equation*}
A_{n, j}\le \frac{(1-c_n)^{-1}}{1-(\lambda_{\alpha, n}/\sqrt{V_{Q_n, j}})^{-2}}\exp\left([1-(1-c_n)^2]\lambda_{\alpha, n}^2/(2V_{Q_n, j}) \right).
\end{equation*}
Using standard calculations and the fact that $nV_{Q_n, j}$ is uniformly bounded
from above and below, we have $(\log k_2)/ K \le
\lambda_{\alpha, n}^2/V_{Q_n, j}\le K \log k_2$ for some constant $K$.
Thus, the right-hand side of the above
display converges to $1$ uniformly over $n$ and $1\le j\le k_{2}$ so long as
$c_n\log k_2\to 0$, which is guaranteed by the assumptions of the theorem.

For $B_{n, j}$, we use a moderate deviations bound as in \textcite[][Chapter
16.7]{feller_introduction_1971}. In particular, the bound
$|F_{n, j}(t)/(2\Phi(-t))-1|\le \tilde K t^3/\sqrt{n}$ holds for all
$1\le t<\overline t_n$, where $\overline t_n$ is any sequence with
$\overline t_n/n^{1/6}\to 0$, and $\tilde K$ depends only on $\overline t_n$ and
the moment conditions and tail bounds on $\varepsilon_i$ \parencite[][Lemma
B.5]{armstrong_multiscale_2016}. Using the fact that
$\lambda_{\alpha, n}/\sqrt{V_{Q_n, j}}$ is bounded by a constant times
$\sqrt{\log k_2}$, it follows that
$\limsup_{n\to\infty}\max_{1\le j\le {k_2}} B_{n, j}\le 1$ so long as
$(\log k_2)^{3/2}/\sqrt{n}\to 0$, which is guaranteed by the conditions of the
theorem.
\end{proof}

\end{appendices}

\printbibliography

\end{document}